\documentclass[12pt,onecolumn,twoside]{IEEEtran}

\usepackage{calc}

\usepackage{multirow}
\usepackage{cite}
\usepackage{graphicx,subfigure}
\usepackage{psfrag}
\usepackage{amsmath,amssymb,amsthm}
\usepackage{color}
\usepackage{tikz} 

\usepackage{url}

\DeclareMathOperator*{\defeq}{\triangleq}

\newtheorem{theorem}{Theorem}

\newtheorem{lemma}{Lemma}

\newtheorem{proposition}{Proposition}

\newcommand{\bit}{\begin{itemize}}
\newcommand{\eit}{\end{itemize}}

\newcommand{\bc}{\begin{center}}
\newcommand{\ec}{\end{center}}

\newcommand{\ba}{\begin{array}}
\newcommand{\ea}{\end{array}}

\newcommand{\beq}{\begin{equation}}
\newcommand{\eeq}{\end{equation}}

\newcommand{\beqn}{\begin{equation*}}
\newcommand{\eeqn}{\end{equation*}}

\newcommand{\bean}{\begin{eqnarray*}}
\newcommand{\eean}{\end{eqnarray*}}
\newcommand{\bea}{\begin{eqnarray}}
\newcommand{\eea}{\end{eqnarray}}

\def\E{\mathbb{E}}

\def\av{\boldsymbol{a}}
\def\bv{\boldsymbol{b}}
\def\cv{\boldsymbol{c}}

\def\ev{\boldsymbol{e}}

\def\hv{\boldsymbol{h}}

\def\sv{\boldsymbol{s}}

\def\uv{\boldsymbol{u}}
\def\vv{\boldsymbol{v}}

\def\xv{\boldsymbol{x}}
\def\yv{\boldsymbol{y}}
\def\zv{\boldsymbol{z}}
\def\Am{\boldsymbol{A}}
\def\Bm{\boldsymbol{B}}

\def\Km{\boldsymbol{K}}

\def\Um{\boldsymbol{U}}

\newcommand{\Cc}{{\mathcal C}}

\newcommand{\Nc}{{\mathcal N}}

\newcommand{\Rc}{{\mathcal R}}

\newcommand{\Xc}{{\mathcal X}}

\newcommand{\Zc}{{\mathcal Z}}
\newcommand{\T}{{\scriptscriptstyle\mathsf{T}}}
\renewcommand{\H}{{\scriptscriptstyle\mathsf{H}}}

\newtheorem{remark}{Remark}

\renewcommand{\Bmatrix}[1]{\begin{bmatrix}#1\end{bmatrix}}

\newcommand{\Tt}{T_{\tau}}
\newcommand{\Td}{T_{d}}

\newcommand{\non}{\nonumber}

\newcommand{\Hen}{\mathbb{H}}
\newcommand{\hen}{\mathrm{h}}
\newcommand{\Imu}{\mathbb{I}}

\newcommand{\Ci}{C_{\text{ideal}}}

\usepackage{mathtools}
\usepackage{xifthen}

\newcommand{\bunderline}[1]{\underline{#1\mkern-4mu}\mkern4mu }

\newcommand{\hvu}{\bunderline{\boldsymbol{h}}}
\newcommand{\xvu}{\bunderline{\boldsymbol{x}}}
\newcommand{\uvu}{\bunderline{\boldsymbol{u}}}
\newcommand{\vvu}{\bunderline{\boldsymbol{v}}}
\newcommand{\yvu}{\bunderline{\boldsymbol{y}}}
\newcommand{\zvu}{\bunderline{\boldsymbol{z}}}
\newcommand{\svu}{\bunderline{\boldsymbol{s}}}
\newcommand{\qvu}{\bunderline{\boldsymbol{q}}}
\newcommand{\evu}{\bunderline{\boldsymbol{e}}}
\newcommand{\bvu}{\bunderline{\boldsymbol{b}}}
\newcommand{\cvu}{\bunderline{\boldsymbol{c}}}
\newcommand{\fvu}{\bunderline{\boldsymbol{f}}}
\newcommand{\avu}{\bunderline{\boldsymbol{a}}}

\newcommand{\uu}{\bunderline{u}}

\newcommand{\au}{\bunderline{a}}

\newcommand{\zerou}{\underline{0}}

\newcommand{\xvut}[1][]{\ifthenelse{\isempty{#1}}{\xvu_{t}}{\xvu_{#1}}}
\newcommand{\hvut}[1][]{\ifthenelse{\isempty{#1}}{\hvu_{t}}{\hvu_{#1}}}
\newcommand{\uvut}[1][]{\ifthenelse{\isempty{#1}}{\uvu_{t}}{\uvu_{#1}}}
\newcommand{\vvut}[1][]{\ifthenelse{\isempty{#1}}{\vvu_{t}}{\vvu_{#1}}}
\newcommand{\svut}[1][]{\ifthenelse{\isempty{#1}}{\svu_{t}}{\svu_{#1}}}
\newcommand{\qvut}[1][]{\ifthenelse{\isempty{#1}}{\qvu_{t}}{\qvu_{#1}}}
\newcommand{\zvut}[1][]{\ifthenelse{\isempty{#1}}{\zvu_{t}}{\zvu_{#1}}}
\newcommand{\yvut}[1][]{\ifthenelse{\isempty{#1}}{\yvu_{t}}{\yvu_{#1}}}
\newcommand{\fvut}[1][]{\ifthenelse{\isempty{#1}}{\fvu_{t}}{\fvu_{#1}}}

\newcommand{\hvt}[1][]{\ifthenelse{\isempty{#1}}{\hv_{t,m}}{\hv_{t,#1}}}
\newcommand{\hvot}[1][]{\ifthenelse{\isempty{#1}}{\hv_{1,m}}{\hv_{1,#1}}}
\newcommand{\zt}[1][]{\ifthenelse{\isempty{#1}}{\zv_{t}}{\zv_{#1}}}
\newcommand{\yt}[1][]{\ifthenelse{\isempty{#1}}{\yv_{t}}{\yv_{#1}}}
\newcommand{\st}[1][]{\ifthenelse{\isempty{#1}}{\sv_{t}}{\sv_{#1}}}

\newcommand{\Wme}{\boldsymbol{w}}

\newcommand{\trace}{\mathrm{tr}}

\IEEEoverridecommandlockouts

\begin{document}
\sloppy

\title{{\huge On the MISO Channel  with Feedback: Can Infinitely Massive Antennas Achieve  Infinite Capacity?}} 
\author{Jinyuan Chen 
\thanks{Jinyuan Chen is with Louisiana Tech University, Department of Electrical Engineering, Ruston, USA (email: jinyuan@latech.edu).  This work was presented in part at the 2017 IEEE International Symposium on Information Theory.}
}


\maketitle
\pagestyle{headings}

\begin{abstract}
We consider communication over a multiple-input single-output (MISO) block fading channel in the presence of an independent noiseless feedback link. We assume that the transmitter and receiver have no prior knowledge of the channel state realizations, but the transmitter and receiver can acquire the channel state information (CSIT/CSIR) via downlink training  and feedback.
For this channel, we show that increasing the number of transmit antennas  to infinity  will \emph{not} achieve an infinite capacity,  for a finite channel coherence length and a finite input constraint on the second or fourth moment.
This insight follows from our new capacity bounds that hold for any \emph{linear} and \emph{nonlinear} coding strategies, and any channel training schemes.
In addition to the channel capacity bounds, we also provide a characterization on the beamforming gain that is also known as array gain or power gain, at the regime with a large number of antennas. 

\end{abstract}

\section{Introduction}

Motivated by the increasing demand for higher data rates in wireless communication systems,  a significant effort is being made to study the use of massive multiple-input multiple-output (massive MIMO) systems \cite{Marzetta:10,Andrews+:14,NLM:13}.  As  equipped with a large number of antennas, the massive MIMO system has potential to boost the channel's beamforming gain that is also known as array gain or power gain (cf.~\cite{Goldsmith:05, TV:05}).  In the massive MIMO channels, for example, in a massive multiple-input single-output (MISO) channel,  the capacity may increase logarithmically with the number of  antennas (cf.~\cite{Goldsmith:05, TV:05, Tel:99,JP:03}), which implies that \emph{infinitely} massive  antennas may allow us to achieve  an \emph{infinite} capacity, even with a \emph{finite} power constraint at the transmitter.

However, the above exciting result  is based on the key assumption that the instantaneous fading coefficients are perfectly known to the receiver/transmitter  (perfect CSIR/CSIT).
In general, CSIT and CSIR entail channel training and feedback. In a typical system with frequency-division-duplex (FDD) mode, CSIT comes from  channel training  and feedback operating over the downlink channel and   feedback channel respectively. The overhead of the training and feedback may in turn affect the channel capacity.
Therefore, it remains open if a massive MIMO system could still provide a significant capacity benefit as we expected.   Specifically, we might ask the following question: 
\emph{Can infinitely massive  antennas always achieve an infinite capacity in a massive MIMO channel?}

In this work, we  study this question by focusing on a massive MISO block fading channel with output feedback. 
We assume that the transmitter and receiver have no prior knowledge of the channel state realizations, but the transmitter and receiver can acquire the channel state information via  downlink training  and feedback.
Let us begin with a simple case where the channel coherence length is $T_c =2$ (channel uses) and the input signals are limited  by a finite \emph{second}-moment  constraint that is also known as long-term average power constraint.
Since the coherence length is $T_c =2$,  the transmitter could use the first and the second channel uses of  each channel block for channel training and data transmission, respectively.  Based on this scheme, one might tentatively expect an infinite rate for the case with  infinite number of transmit antennas, because a little channel state information might be very useful for this case.
However, we show that in this setting increasing the  transmit-antenna number  to infinity  will \emph{not}  yield an infinite capacity.  
This  result is in sharp contrast to the result of the setting with perfect CSIT/CSIR (e.g. through a genie-aided training and feedback),  in which the capacity will go to infinity as the antenna number grows to infinity (cf.~\cite{JP:03, GV:97, Gallager:68, KAC:90, AC:91}).

As a main contribution of this work, we derive capacity upper bound and lower bound for the MISO channel with feedback under the second  moment and the fourth moment input constraints, respectively.  
The result reveals that  increasing the  transmit-antenna number to infinity  will \emph{not}  yield an infinite capacity, for the case with a finite channel coherence and a finite input constraint on the second or fourth moment.
In addition to the capacity bounds, this work also provides a characterization on the channel's beamforming gain at the regime with a large number of antennas. 
Similarly to the degrees-of-freedom metric (cf.~\cite{ZT:02}) that usually captures the prelog factor  of  capacity at the high power regime,
beamforming gain is used in this work to capture the prelog factor of capacity at the high antenna-number regime.

{\bf Related works:} The  capacity of the channels with feedback, or with imperfect CSIT/CSIR, has been studied extensively in the literature for varying settings, e.g.,  the point-to-point channels  
(cf.~\cite{LS:02, Medard:2000, Narula+:98, Love+:03, Erkip+:03,TG:06, RR:06, XGR:06, SJ:07, JS:07, DL:06, YL:07, KVM:14}) 
and the broadcast channels (cf.~\cite{LSW:05, YJG:06, Jindal:06m, Love+:08, DJ:14}).
However, a common assumption in  those  works above is that  imperfect CSIT and CSIR were acquired without considering the overhead in channel training.
The channel training overhead cannot be negligible  when the number of channel parameters to be estimated is large and the channel coherence is relatively small.
This work  is categorized  in the line of works studying the multiple-antenna networks where CSIT and CSIR were acquired via  channel training and feedback,  such as \cite{SH:10, Caire+:10m, KCJ:08, KJC:09, CJKR:07,HKD:11,KJC:11, Caire+:13, Caire+:14,Choi+:14}. 
To the best of our knowledge, the previous capacity upper bounds on this topic hold only for \emph{linear schemes}. 
Specifically, the work in \cite{SH:10} considered, among others, a MISO block fading channel with limited feedback, under the assumptions of \emph{linear} coding schemes and  a \emph{fixed ratio} of the coherence length to the antenna number.  
The work in \cite{SH:10} also assumed a dedicated training, i.e.,  a certain fraction of each channel block is used specifically  for the channel training.
For that setting, the work in \cite{SH:10} showed that the (linear) capacity can increase logarithmically with the number of  antennas. 
In a similar direction,  the work in \cite{Caire+:10m} investigated the achievable ergodic rates of a MIMO block fading broadcast channel with dedicated training and noisy feedback, under the assumption of \emph{linear} coding schemes. The work in \cite{Caire+:10m} derived  the lower  and upper bounds of the achievable rate  as the \emph{expectation} of some functions of the channel estimates.  
Our channel can be considered as a specific block fading channel with in-block memory, due to feedback, in which the capacity is generally \emph{NP-hard} to compute (cf.~\cite{Kramer:14}). 
Specifically, the capacity of our setting is a \emph{multiletter} expression 
and finding  the optimal input distribution is NP-hard  (cf.~\cite{Kramer:14}). In our setting, the channel input  at each time  is a function of the previous channel outputs and the message. 
Note that, under the assumptions of  linear coding schemes and a dedicated channel training,  the capacity bound may be reduced to  a \emph{single-letter} expression (cf.~\cite{SH:10,Caire+:10m}). 
That is because,  with linear coding schemes and a dedicated channel training, the channel can be considered as a \emph{non-feedback} channel with imperfect CSIT/CSIR. 
However, in our setting, feedback \emph{cannot} be removed at any point of time.  Therefore, the previous approaches used in the settings with linear schemes and dedicated channel training (cf.~\cite{SH:10,Caire+:10m}) might not be directly applied in our setting.
In our converse proof, we transform the \emph{NP-hard} capacity problem into a relaxed problem that is  computable.  In our work we focus on the beamforming gain performance, as tight capacity bounds are still hard to compute. 

In one different direction, the previous work in \cite{HH:03} studied the capacity of  a MIMO channel with training but \emph{without feedback}. In that setting, the receiver can acquire the CSIR via channel training but the transmitter will not have channel state information due to the lack of the feedback link.  
Finally, in  another different direction, some previous  works considered the noncoherent communication without channel training and feedback (cf.~\cite{ZT:02, MH:99, MRYDLSB:13,XGH:13, CCLM:13, CMGEG:15, CMG:16, FQW:17, SSD:17, BAZHH:18, GP:17, VR:97, Raphaeli:96, GDME:17} and the references therein).  Specifically, the authors in \cite{ZT:02} studied a MIMO  noncoherent block fading channel, where the receiver and the transmitter have no channel side information, and computed the capacity prelog (degrees-of-freedom) of this channel at high power regime. 
In \cite{MRYDLSB:13},  the authors investigated the degrees-of-freedom of a single-input multiple-output (SIMO) channel with temporally correlated  block fading, in the noncoherent setting. In \cite{CMGEG:15} and  \cite{CMG:16},  the two works  studied the capacity scaling laws for noncoherent  communications in the wideband massive SIMO channel and the massive SIMO multiple access channel, respectively, at the regime with a large number of receiving antennas. 
Our work is very different from those works, as we consider both training and feedback in our setting.

The remainder of this work is organized as follows. 
Section~\ref{sec:system} describes the system model. 
Section~\ref{sec:result}  provides  the main results of this work. 
The converse  and  achievability proofs are described in Section~\ref{sec:converse2g}, Section~\ref{sec:misoc},  and the appendices.
The conclusion and discussion are provided  in Section~\ref{sec:concl}.
Throughout this work, $(\bullet)^\T$, $(\bullet)^{*}$, $(\bullet)^{\H}$ and $(\bullet)^{-1}$  denote the transpose, conjugate,  conjugate transpose and inverse operations, respectively. $||\bullet||$ denotes the Euclidean norm, $\det(\bullet)$ denotes the determinant, $\trace(\bullet)$ denotes the trace, and $|\bullet|$ denotes the magnitude. 
We use $\Am \succeq \mathbf{0}$ to denote that matrix $\Am $ is  Hermitian positive semidefinite, and use $\Am \preceq \Bm $ to mean that  $\Bm- \Am \succeq \mathbf{0}$.
Logarithms are in base~$2$.
We let $\ev^{j}_{i} = (\ev_i, \ev_{i+1}, \cdots, \ev_j)$ if $i\leq j$, else, let $\ev^{j}_{i}$ denote an empty term. Let $\ev^{j} = \ev^{j}_{1}$.
$\Imu(\bullet)$, $\Hen(\bullet)$ and $\hen(\bullet)$ denote the mutual information, entropy and differential entropy,  respectively.  
$\lfloor \bullet \rfloor$ denotes the largest integer not greater than the argument and $\lceil \bullet \rceil$ denotes the smallest integer not less than  the argument. 
   $\Zc$,  $\Rc$ and $\Cc$ denote the sets of integers, real numbers and  complex numbers, respectively.    
   $o(\bullet)$ comes from the standard Landau notation, where 
$f(x)=o(g(x))$ implies that $\lim_{x \to \infty} f(x)/g(x) =0$. 
 $[ a\  \text{mod} \  m]$ denotes the modulo operation, i.e., $[ a\  \text{mod} \  m] = r$ if the number of $a$ can be represented as  $a= \ell m + r $ for $\ell \in \Zc$ and $|r| < |m|$.  
 $\uvu \sim \mathcal{CN}(\uu_0, \Omega_0)$  denotes that the random vector $\uvu$ is  proper complex Gaussian  distributed with mean $\uu_0$ and covariance $\Omega_0$. $\uvu$ is said to be proper if $\E [(\uvu -\E[\uvu] )(\uvu -\E[\uvu] )^\T] = \mathbf{0}$.
When a complex Gaussian vector is proper and with zero mean,  it is said to be circularly symmetric complex Gaussian. 
$\uv \sim \Xc^2 (k)$  denotes that  $\uv$ is a chi-squared random variable that is defined as the sum of squares of $k$ independent and identically distributed (i.i.d.)  standard normal  $\Nc(0, 1)$ random variables.
Unless for some specific parameters, the random matrix, random variable and random vector are usually denoted by the  bold italic uppercase symbol (e.g., $\Um$),  bold italic lowercase symbol (e.g., $\uv$) and  bold italic lowercase symbol with underline (e.g., $\uvu$) respectively,  while the corresponding realizations are non-bold (e.g., $U$, $u$ and $\uu$).

\section{System model \label{sec:system} }

We consider a   MISO  channel where a transmitter with $M $ ($M\geq 2$) antennas sends  information to a single-antenna user, as illustrated in Fig.~\ref{fig:MassiveMISO}. 
The signal  received  by the user at time $t$  is given as
\begin{align}
\yv_t &=   \hvu^\T_t \xvu_t  +   \zv_t,   \label{eq:misoy}
\end{align}
$t=1,2,\cdots,n$, where $\xvut$ denotes the transmitted signal vector at time $t$, $\zt \sim \mathcal{CN}(0, 1)$ denotes the additive white Gaussian noise (AWGN), $\hvut \defeq [ \hvt[1], \hvt[2], \cdots,   \hvt[M]]^\T \sim \mathcal{CN}( \zerou,  I_M) $ denotes the $M\times 1$ channel vector at time~$t$, and $ \hvt[m]$ denotes a channel coefficient of the $m$th transmit antenna at time~$t$. 
We assume a \emph{block fading} model (cf.~\cite{BPS:98, Caire+:10m}), in which the channel coefficients remain constant during a coherence block of $T_c$ channel uses and change \emph{independently} from one block to the next, i.e., \[ \hvut[\ell T_c+1]=\hvut[\ell T_c+2]=\cdots =\hvut[\ell T_c+T_c] \quad \text{and} \quad   \hvut[\ell T_c+T_c]  \  \text{is independent of} \ \hvut[(\ell +1) T_c+1] \] for $\ell =0, 1,  \cdots, L-1$ and $L=n/T_c$, where $n, L$ and $T_c$ are assumed to be integers. 
We assume that the channel coefficients in each block are initially unknown to the transmitter and the user. 
At the end of each time $t$, the user can feed back the channel outputs to the transmitter over an independent feedback link. 
For simplicity we assume that the feedback link is noiseless (error-free) and with a unit time delay, i.e., at the beginning of time $t+1$, the transmitter knows  $\yv^{t} \defeq ( \yt[1], \yt[2], \cdots, \yt[t] )$.    

\begin{figure}
\centering
\includegraphics[width=7cm]{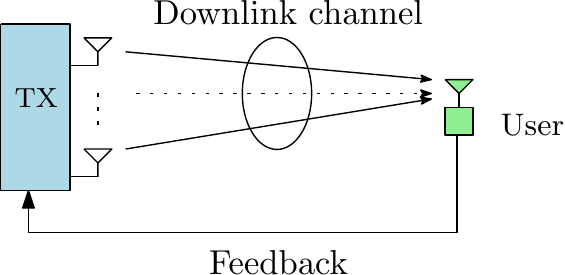}
\caption{MISO channel with a feedback link.}
\label{fig:MassiveMISO}
\end{figure}

For this feedback communication  of total $n$ channel uses, the transmitter wishes to send the user a message index $\Wme$ that is uniformly distributed over $\ \{1,2, \cdots, 2^{nR}\}$.
We specify a $(2^{nR}, n)$ feedback code with encoding maps
 \begin{align}
  \xvut : \  \{1,2, \cdots, 2^{nR}\} \times \Cc^{t-1}   \to \Cc^{M},  \quad  t=1,2, \cdots,n   \label{eq:mapx} 
  \end{align} 
that result in codewords (or code functions, more precisely) 
 \begin{align}
   \xvu^n(\Wme, \yv^{n-1} ) \!=\!  \bigl(  \xvut[1](\Wme ),  \xvut[2](\Wme, \yv_1), \cdots,    \xvut[n](\Wme, \yv^{n-1} )  \bigr).  \label{eq:mapcodeword} 
  \end{align} 
 Then the user decodes the message with decoding maps
   \begin{align}
 \hat{\Wme}_n  :  \ \Cc^n \to    \{1,2, \cdots, 2^{nR}\}.   \label{eq:mapdec} 
  \end{align} 
 We consider two cases of  constraints  on the input signals.  
 At first we consider the \emph{second moment} input constraint such that 
\begin{align}
 \frac{1}{n} \sum_{t=1}^{n}    \E \Bigl[ \|\xvut( \Wme,  \yv^{t -1})\|^2 \Bigr]  \leq  P  \label{eq:powerEve} 
  \end{align} 
 where  the expectation is over all possible noise and fading sequences as well as the message $\Wme$, for some $P \in \Rc $, $0< P < +\infty$.   This second moment  constraint is also known as the \emph{average  power constraint}. 
 We then consider the \emph{fourth moment} input constraint such that 
 \begin{align}
 \frac{1}{n} \sum_{t=1}^{n}    \E \Bigl[ \|\xvut( \Wme,  \yv^{t -1})\|^4 \Bigr]  \leq  \kappa^2 P^2   \label{eq:power4m}
  \end{align} 
  where $\kappa$ is a positive constant.  The fourth moment input constraint has been introduced in several communication scenarios (cf.~\cite{MG:02, GPV:05, GPV:02, BT:06,  PV:04}). 
For some certain cases,  imposing the fourth moment constraint  is identical to imposing a limitation on the \emph{kurtosis} that is a measure of peakedness of the signal (cf.~\cite{GPV:05, GPV:02, BT:06}).
The probability  of error $\mathrm{P}^{(n)}_e$ is defined as   
   \begin{align}
 \mathrm{P}^{(n)}_e \defeq&    \frac{1}{2^{nR}} \sum_{w=1}^{2^{nR}} \mathrm{Pr}\bigl\{ \hat{\Wme}_n (\yv^n) \neq w | \Wme=w  \bigr\}  . \non 
  \end{align} 
 A rate $R$ (bits per channel use) is said to be achievable if there exists a sequence of $(2^{nR}, n)$ codes with $ \mathrm{P}^{(n)}_e \to 0$  as $n\to \infty$.
The capacity of this channel $C$ is defined as  the supremum of all achievable rates.

\subsection{Beamforming gain}

In this work, we specifically focus on the capacity effect of the channel with a large  number of antennas, which may be captured by the metric of beamforming gain. 
For the  capacity effect of the channel with high power, one might consider the metric of degrees-of-freedom that is beyond the scope of this work.
 
 In our setting, the channel capacity (and the beamforming gain) might depend on the antenna number $M$ and the channel coherence length $T_c$.  Intuitively, when the channel coherence length $T_c$ is sufficiently large, i.e., $T_c \gg M$,  the channel might be considered as a static channel, in which the capacity (and the beamforming gain) might be the same as that of an ideal case with perfect CSIT and CSIR.  However,  when $T_c$ is decreased to a relatively small number compared with $M$, e.g., when $T_c=1$ (the case with fast fading),  then the channel capacity (and the beamforming gain) might be decreased significantly.   
In order to study the interplay between the beamforming gain,  antenna number $M$, and channel coherence length $T_c$, in this work we introduce a new parameter:  
\[\alpha  \defeq \frac{\log T_c}{\log M},  \quad  \alpha \geq 0\] 
that is the ratio between the coherence length and the antenna number in a logarithmic scale. In our setting, channel coherence length $T_c$ can be rewritten by  $T_c \defeq M^{\alpha}$.
When $M$ is very large,  $\alpha = 0$ refers to a class of channels where the coherence length $T_c$ is finite, while $\alpha = 1$ refers to a class of channels where $T_c$ and $M$ are scaled similarly.

In our setting, the beamforming gain of the channel is defined as   \[ b(\alpha)  \defeq  \limsup_{  M\to \infty}   \frac{ C(\alpha, P, M)  }{ \log M}.\]  
Similarly to the definition of  generalized degrees-of-freedom (GDoF, see \cite{ETW:08}), 
the beamforming gain $b(\alpha)$ captures the capacity prelog factor  for a class of channels with a fixed $\alpha$, at the regime with a large number of antennas. 
 This approximation on the capacity is a middle step, or perhaps the first step, for understanding the channel capacity.
In this setting $b= 0$ means zero beamforming gain, while $b = 1$ denotes a full beamforming gain.
For the ideal case with perfect CSIT and  CSIR (e.g., through a genie-aided method) one might achieve a full beamforming gain. However, for this setting where CSIR and CSIT are acquired via downlink training and feedback, the beamforming gain is generally unknown so far.  In the following we seek to characterize the  beamforming gain of this setting.

\section{Main results \label{sec:result}}

This section provides the main results for  a MISO  channel with feedback defined in Section~\ref{sec:system}.
The proofs are shown in Section~\ref{sec:converse2g}, Section~\ref{sec:misoc},  and the appendices. 
Before showing the main results of this work, let us first revisit the \emph{ideal} case of MISO  channel with  perfect CSIT and  CSIR, and with a \emph{second} moment input constraint.
According to the previous works in \cite{JP:03, GV:97, Gallager:68, KAC:90, AC:91},  the channel capacity of this ideal case, denoted by $ \Ci$, is characterized  in the following closed form
\begin{align}
 \Ci  =   \max_{\bar{P} (\gamma): \int_{\gamma}\bar{P} (\gamma) f_{\gamma}(\gamma) d\gamma =P}  \int_{\gamma}   \log \bigl(1 +   \bar{P} (\gamma) \cdot \gamma \bigr) f_{\gamma} (\gamma) d \gamma \label{eq:idealc}  
\end{align}
 where $\gamma \defeq  \|\hvu_t \|^2$; \  $f_{\gamma} (\gamma)$ is the probability  density function of $\gamma$;  $\bar{P} (\gamma)$ is the power allocation  function; and the optimal solution  of $\bar{P} (\gamma)$ is based on a water-filling algorithm (cf.~\cite{JP:03, GV:97, Gallager:68, KAC:90, AC:91}).    
 When the antenna-number $M$ is large, the capacity expressed in \eqref{eq:idealc} tends to $\log(1+ PM)$, which is summarized in the following proposition. 
 \vspace{5pt}
 \begin{proposition}  [Ideal case] \label{thm:idealc}
For the ideal case of MISO  channel (cf.~\eqref{eq:misoy}) with  perfect CSIT and  CSIR, and with a  second moment input constraint (cf.~\eqref{eq:powerEve}),  
the channel capacity $\Ci$ is approximated by:    \[ \Ci = \log(1+ PM) + o(\log M), \]
and the corresponding beamforming gain $b (\alpha)$ is characterized by 
\[   b (\alpha) = 1 ,  \quad  \forall \alpha \geq 0. \]
\end{proposition}
\vspace{5pt}
Proposition~\ref{thm:idealc} follows from the capacity expression in \eqref{eq:idealc} and the asymptotic analysis that is provided in Appendix~\ref{sec:idealc}.  Proposition~\ref{thm:idealc} reveals that the capacity of a MISO  channel with  perfect CSIT and  CSIR  will go to infinity as  the antenna-number $M$ grows to infinity, even with a finite  input constraint $P$. 
Proposition~\ref{thm:idealc} also reveals that full beamforming gain  $(b = 1)$ is achievable with  perfect CSIT and  CSIR, for any given channel coherence length $T_c \in \Zc^+$.

Let us now go back to the  MISO  channel with feedback defined in Section~\ref{sec:system}, where the transmitter and receiver have no prior knowledge of the channel state realizations, but the transmitter and receiver can acquire the CSIT/CSIR via  downlink training  and feedback.  
In this work, we specifically focus on the channel capacity effect of the system with a large  number of antennas, which may be captured by the metric of beamforming gain. 
The following  results summarize  the beamforming gain of  the channel under two input constraints, respectively.
\vspace{5pt}
\begin{theorem}  [Beamforming gain, second moment] \label{thm:MISOb2}
For the MISO  channel with feedback defined in Section~\ref{sec:system},  the  beamforming gain is bounded as 
\[    \min\{\alpha,1\}  \leq b (\alpha)    \leq   \min\{2\alpha,1\} ,    \quad   \quad   \forall \alpha \geq 0  \]
under the  second moment input constraint (cf.~\eqref{eq:powerEve}).
\end{theorem}
\vspace{5pt}
Theorem~\ref{thm:MISOb2} follows from a capacity upper bound in Theorem~\ref{thm:MISO2u} that  is shown in Section~\ref{sec:converse2g} (see also Remark~\ref{rmk:proofuMISOb2} in Section~\ref{sec:converse2g}), and a capacity lower bound in Theorem~\ref{thm:MISOlower} that is shown in Section~\ref{sec:misoc}.  
In Theorem~\ref{thm:MISOb2}, we have $b =0$ when $\alpha=0$. This implies that, given a finite second-moment input constraint and  a finite  channel coherence length (i.e., $\alpha=0$),  the capacity will \emph{not} go to infinity (i.e., $b =0$)  as  the antenna number $M$ grows to infinity.   This  result is in sharp contrast to the result of the perfect CSIT/CSIR case,  in which the capacity will be infinite when $M$ is taken to infinity, as shown in  Proposition~\ref{thm:idealc}.

\begin{figure}
\centering
\includegraphics[width=7cm]{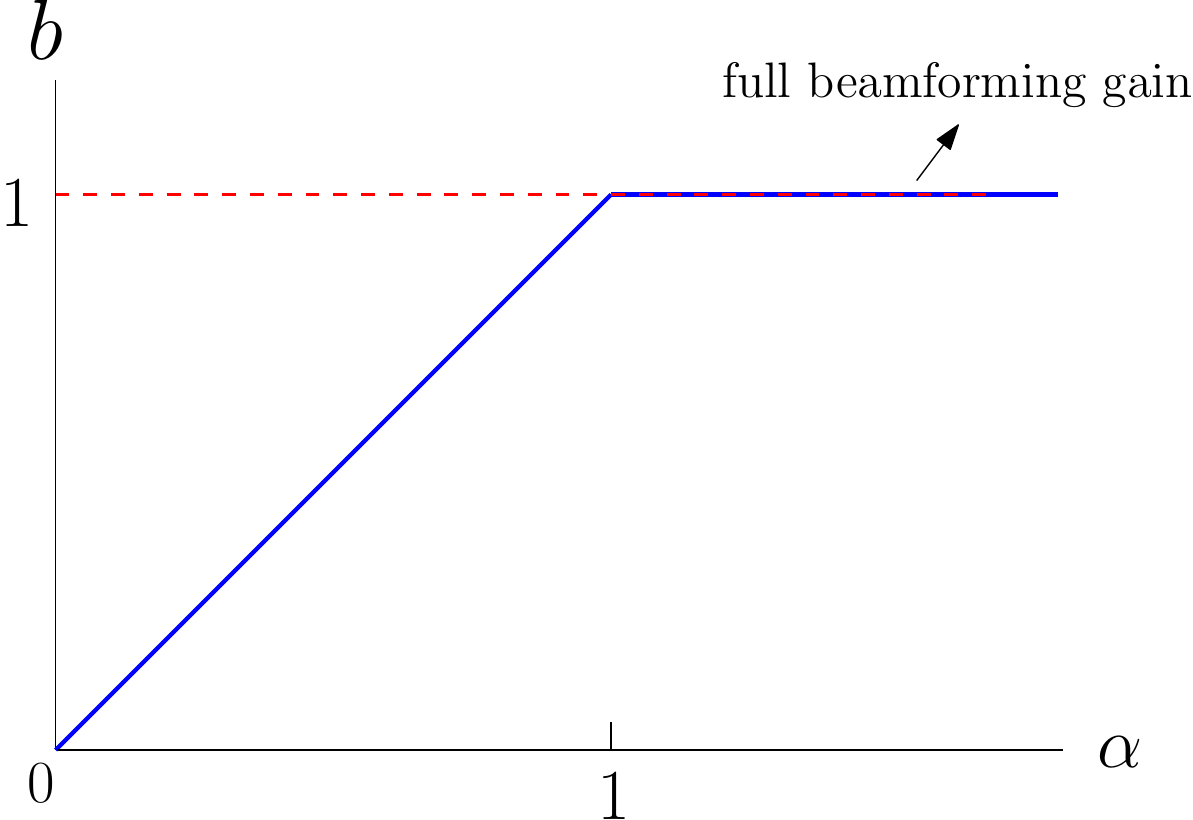}
\caption{Beamforming gain $b$ vs. $\alpha$ for the MISO channel with feedback, under the fourth moment input constraint.}
\label{fig:beamMassive}
\end{figure}

\vspace{8pt}
\begin{theorem}  [Beamforming gain, fourth moment]  \label{thm:MISOb4}
For the MISO  channel with feedback defined in Section~\ref{sec:system},  the  beamforming gain is characterized by  
\[   b (\alpha) = \min\{\alpha,1\}, \quad \quad   \forall \alpha \geq 0\]
under the fourth moment input constraint (cf.~\eqref{eq:power4m}). 
\end{theorem}
\vspace{5pt}

Theorem~\ref{thm:MISOb4} follows from a capacity upper bound  in Theorem~\ref{thm:MISO} that is shown in Appendix~\ref{sec:converse} (see also Remark~\ref{rmk:proofuMISOb4} in Appendix~\ref{sec:converse}), and a capacity lower bound in Theorem~\ref{thm:MISOlower} that is shown in Section~\ref{sec:misoc}.  Theorem~\ref{thm:MISOb4} reveals that,  given a finite fourth-moment input constraint and a finite channel coherence length, again, the capacity will \emph{not} go to infinity when  the antenna number $M$ is taken to infinity.

As illustrated in Fig.~\ref{fig:beamMassive},  a full beamforming gain, i.e., $b =1$, is achievable when $\alpha \geq 1$.  Intuitively, for the case with large $\alpha$, $\alpha \gg 1$,  the channel can be considered as a static channel, in which a full beamforming gain could be achieved easily via sufficiently long downlink training. 
Theorem~\ref{thm:MISOb4} reveals an interesting insight that, instead of a large $\alpha$,  $\alpha = 1$ is sufficient for achieving a full beamforming gain.  
From Theorem~\ref{thm:MISOb4} we note that, under a finite fourth-moment input constraint,  the channel capacity is asymptotically scaled as $\log ( 1+ \min \{ M,  T_c\})$, which reveals another interesting insight that using more transmit antennas than the coherence length does not yield a \emph{significant} gain in capacity for this setting with feedback, in an asymptotic sense.  This insight is similar to the insight  for the setting \emph{without feedback},  in which using more transmit antennas than the coherence length does not yield a gain in capacity  (cf.~\cite{MH:99}).
We conjecture that the beamforming gain for the setting with second moment input constraint (cf. Theorem~\ref{thm:MISOb2}) is the same as that for the setting with fourth moment input constraint (cf. Theorem~\ref{thm:MISOb4}), i.e., $b (\alpha) = \min\{\alpha,1\}$. If this is the case, then the above two insights also hold for the setting with  second moment input constraint.

\section{Converse: the case with second moment input constraint  \label{sec:converse2g} }

This section provides a capacity upper bound for the MISO  channel defined in Section~\ref{sec:system}, under a second moment input constraint (cf.~\eqref{eq:powerEve}).    The result  is summarized in the following theorem. 
\vspace{5pt}
\begin{theorem} [Upper bound, second moment]  \label{thm:MISO2u}
For the MISO  channel with feedback defined in Section~\ref{sec:system}, the  capacity is upper bounded by 
\[C \leq    2 \log \bigl(  4 +  3 \cdot \min\{ M,  T_c\}  \bigl)   +   \log\bigl(      1+    4 P    \bigr)     \]
under the second moment input constraint.  
\end{theorem}
\vspace{5pt}

\vspace{5pt}
\begin{remark} [Proof of Theorem~\ref{thm:MISOb2}, upper bound]  \label{rmk:proofuMISOb2}
In this work, we specifically focus on the  beamforming gain of the channel. 
From the capacity upper bound in Theorem~\ref{thm:MISO2u}, we can easily derive  an upper bound on the  beamforming gain: 
\begin{align}
b(\alpha)  &\leq   \lim_{  M\to \infty}   \frac{  2 \log \bigl(  4 +  3 \cdot \min\{ M,  M^{\alpha}\}  \bigl)   +   \log\bigl(      1+    4 P    \bigr)  }{ \log M}     \non   \\
  &=  \min \{ 2  , 2 \alpha\}   \non
\end{align}
under the second moment input constraint, recalling that $T_c= M^{\alpha}$.  
On the other hand, the beamforming gain is also  upper bounded by $b(\alpha) \leq 1$, $\forall \alpha$,  even for the ideal case with perfect CSIT and CSIR (cf.~Proposition~\ref{thm:idealc}). Therefore,  the beamforming gain is upper bounded by  \[ b(\alpha) \leq \min \{1  ,  2\alpha\}\] under the second moment input constraint.  It then proves the beamforming gain  upper bound described in Theorem~\ref{thm:MISOb2}.
\end{remark}
\vspace{5pt}

In the rest of this section we will provide the proof of Theorem~\ref{thm:MISO2u}.  
In the proof we will use Lemmas~\ref{lm:densityh}-\ref{lm:ExpCht}  shown in this section (see later on).  
We will also use some notations  given as
\begin{align}
\hat{\hvu}_{t} &\defeq  \sum_{ i = T_c \lfloor \frac{t-1}{T_c}\rfloor  + 1 }^{t -1}    \frac{\Omega_{i} \xvut[i]^{*} (\yt[i] - \xvut[i]^\T \hat{\hvu}_{i} )  }{  \xvut[i]^\T \Omega_{i} \xvut[i]^{*}  +1 }  \quad \text{for} \quad  t \neq  \ell T_c +1,   \label{eq:mmse424}  \\ 
  \Omega_{t}  &  \defeq I_M - \sum_{ i = T_c \lfloor \frac{t-1}{T_c}\rfloor  + 1 }^{t-1}   \frac{\Omega_{i} \xvut[i]^{*}  \xvut[i]^\T  \Omega_{i}}{\xvut[i]^\T \Omega_{i} \xvut[i]^{*}  +1 }   \quad \text{for} \quad t  \neq  \ell T_c +1,   \label{eq:mmse9832}  \\
  \tilde{\hvut} & \defeq  \hvut -  \hat{\hvut}   \quad   \forall t  \in \{ 1, 2,  \cdots, n\}  \label{eq:mmse924}     
  \end{align}
and  $\hat{\hvu}_{\ell T_c +1} \defeq \zerou$, \ $ \Omega_{\ell T_c +1}  \defeq  I_M $,  $ \forall \ell  \in \{0, 1,  \cdots,L-1\}$.  
An additional notation that will be used is given as 
\begin{align}
\tilde{\yv}_t  \defeq   \| \hat{\hvu}_{t}\| +  \tilde{ \zt},   \quad t=1,2, \cdots,n   \label{eq:genie5678}  
\end{align}
where $\hat{\hvu}_{t}$ is defined in \eqref{eq:mmse424}, and $ \tilde{ \zt}   \sim  \Cc\Nc (  0 ,   1) $ is a random variable that is independent of  $\Wme, \{\hvut \}_t, \{\zt \}_t $ and $\{ \tilde{\zv_{\ell}} \}_{\ell \neq t}$.  Before describing the necessary lemmas and  the proof details, let us first provide a roadmap of our proof.

\vspace{5pt}

{\bf Roadmap and intuitions of the proof:} The challenge of our proof is mainly due to the correlation between the channel input and the previous  channel outputs (see \eqref{eq:mapx} and \eqref{eq:mapcodeword}), and the high dimension of the channel inputs (with a large number of antennas).
Note that the previous approaches used in the settings with linear schemes (cf.~\cite{SH:10,Caire+:10m}) cannot be directly applied in our setting, in which the coding scheme could be nonlinear.  The proof consists of the following steps. 
\begin{itemize}
\item  \emph{Step 1:   genie-aided channel enhancement}.  In this step, we  enhance the original setting by providing a genie-aided information, i.e., $\{ \tilde{\yv}_t \}_{t=1}^n$, that is defined in  \eqref{eq:genie5678},  to the receiver at the end of the whole communication.
\item \emph{Step 2:  bound the rate of the enhanced channel}.  In this step, we bound the rate of the enhanced channel as
\begin{align}
 nR \leq   \sum_{t=1}^n   \hen(\tilde{\yv}_t ) + \sum_{t=1}^n \hen(\yv_t \big|  \tilde{\yv}_{t} )  + o(\log M) . \non  
\end{align}
In the above bound,   the differential entropies  $\hen(\tilde{\yv}_t ), t=1,2, \cdots, n,$  correspond to the \emph{penalty} terms due to the  genie-aided  channel enhancement. 
\item  \emph{Step 3: bound the penalty terms}. In this step, we prove that 
\begin{align}
   \hen\bigl( \tilde{\yv}_t   \bigr)   \leq      \min\{\alpha,1\} \cdot \log M  + o(\log M), \quad  \forall t=1,2, \cdots, n  \non
\end{align}
by using the differential entropy maximizer (i.e., Gaussian distribution)  and Lemma~\ref{lm:ExpCht} (see below).  In this step, Lemma~\ref{lm:ExpCht} is used to bound the average power of $ \tilde{\yv}_t$.  Note that the \emph{penalty} terms  lead to a gap between our beamforming gain upper bound and inner bound, as shown in Theorem~\ref{thm:MISOb2}.
\item  \emph{Steps 4-6: bound the differential entropy $\hen( \yv_t \big|  \tilde{\yv}_{t} )$}.  The  difficulty of our proof lies in bounding the differential entropy $\hen( \yv_t \big|  \tilde{\yv}_{t} )$, which is involved with Steps 4-6.   

In  \emph{Step~4}, we prove that 
\begin{align}
  \hen( \yv_t \big|  \tilde{\yv}_{t} )  \leq    \E \Bigl[  \log\Bigl(\pi e \bigl(  1 + \E \bigl[ |\hvut^\T \xvut |^2   |  \tilde{\yv}_{t}   \bigr]    \bigl) \Bigr)   \Bigr]  \non 
\end{align}
by using the differential entropy maximizer.  
    
In  \emph{Step~5}, we provide an upper bound on the expectation term $\E \bigl[ |\hvut^\T \xvut |^2  |  \tilde{\yv}_{t} \bigr]$. 
The challenge of this step is due to the correlation between $\xvut$ and $\hvut$ (see \eqref{eq:mapx} and \eqref{eq:mapcodeword}).  
In this step, we prove that 
\begin{align}
 \E \Bigl[ |\hvut^\T \xvut |^2  \big|  \tilde{\yv}_{t}\Bigr] \leq   \E \Bigl[    \|\hat{\hvut}\|^2 \cdot \| \xvut \|^2   +  \| \xvut \|^2     \Big |  \tilde{\yv}_{t}  \Bigr]   \label{eq:3trmax366111} 
\end{align}
by using Lemma~\ref{lm:densityh} and  Lemma~\ref{lm:estimate1} (see below), where $\hat{\hvut}$ is defined in \eqref{eq:mmse424}.   Lemma~\ref{lm:densityh} corresponds to the minimum mean square error (MMSE) estimator.  It reveals that $\hat{\hvu}_{t}$ is the MMSE estimate of $\hvut$  given  $(\yv^{t-1}, \Wme)$, and that
$ \tilde{\hvut} \  \big| \  (\yv^{t-1}, \Wme)     \sim       \Cc\Nc (  \zerou ,   \Omega_{t}) $,
where $ \tilde{\hvut} \defeq  \hvut -  \hat{\hvut}$ and  $\Omega_{t}$ is defined in \eqref{eq:mmse9832}.  Lemma~\ref{lm:estimate1} reveals that $\mathbf{0} \preceq \Omega_{t} \preceq  I_M$,  $\forall t \in \{1,2, \cdots, n\}$.

In \emph{Step 6}, we provide a final bound on $\hen( \yv_t \big|  \tilde{\yv}_{t} )$.   Note that  $\xvut$ and $\hat{\hvut}$ are correlated.  Without conditioning on $ \tilde{\yv}_{t}$ (a genie-aided information),  it is challenging to derive a tight bound on the expectation term  $\E \bigl[    \|\hat{\hvut}\|^2 \cdot \| \xvut \|^2  \bigr]$ in \eqref{eq:3trmax366111}.
  In this  step we take  the benefit of genie-aided channel enhancement ---  which leads to a condition $\tilde{\yv}_{t}$ in the expectation term  $\E \bigl[    \|\hat{\hvut}\|^2 \cdot \| \xvut \|^2  \bigr]$ in \eqref{eq:3trmax366111} --- and prove that 
\[  \hen( \yv_t \big|  \tilde{\yv}_{t} )   \leq      \min\{\alpha,1\} \cdot \log M    + \log\bigl(      1+   4 \cdot \E \bigl[ \| \xvut \|^2   \bigr]   \bigr)  + o(\log M)  . \]  
 \item  \emph{Step~7:  derive  a final capacity upper bound}. In the final step, we combine the previous steps and derive  a  capacity upper bound with  the optimal power allocation: 
 \begin{align}
 nR   \leq   n \cdot \min\{2 \alpha, 2\}  \cdot \log M    +  n \cdot o(\log M).    \non   
\end{align}
\end{itemize}

\vspace{5pt}

The lemmas that will be used in our proofs are provided as follows.  The first lemma corresponds to the MMSE estimator.  This lemma is the extension of the well-known result of MMSE estimator (see, for example, \cite[Chapter 15.8]{Kay:93}).

\begin{lemma}  [MMSE]\label{lm:densityh}
Consider independent complex Gaussian random vectors    $\zvut \in \Cc^{N\times 1} \sim \Cc\Nc ( \zerou ,  I_N)$, $t=1,2,\cdots,T$,  and  $\uvu \in \Cc^{M\times 1}\sim \Cc\Nc ( \hat{\uvu}_{1} ,    \Omega_{1})$, for  some fixed $\hat{\uvu}_1$ and Hermitian positive semidefinite $ \Omega_{1}$.   
Let \[\yvut=\Am_t\uvu+ \zvut,  \quad  t=1,2,\cdots,T,\] where $\Am_t\in \Cc^{N\times M}$ is a deterministic function of $(\yvu^{t-1}, w)$ and
 $w$ is a  fixed parameter. 
Then, the conditional density of $\uvu$ given  $(\yvu^{t-1}, w)$ is 
\[  \uvu \big| (\yvu^{t-1}, w)   \  \sim  \  \Cc\Nc ( \hat{\uvu}_{t},   \Omega_{t})\]
where  
 \begin{align}
 \hat{\uvu}_{t} & =  \hat{\uvu}_{1} +  \sum_{ i =  1 }^{t -1} \Omega_{i} \Am^\H_{i}( \Am_{i}\Omega_{i}\Am^\H_{i}  + I_N )^{-1} (\yvut[i]  - \Am_{i} \hat{\uvu}_{i} )  \label{eq:mmseuv11}  \\ 
  \Omega_{t}  &  =   \Omega_{1} - \sum_{ i =  1 }^{t-1} \Omega_{i} \Am^\H_{i} ( \Am_{i}\Omega_{i}\Am^\H_{i}  + I_N )^{-1}  \Am_{i}\Omega_{i} \label{eq:mmseuv22}  
\end{align}
 for $t = 2,3, \cdots, T$.
Furthermore,   $\hat{\uvu}_{t}$ and $\vvu_t \defeq \uvu  - \hat{\uvu}_{t}$ are conditionally independent given  $(\yvu^{t-2}, w)$,  and we have 
\[  \vvu_t \big| (\yvu^{t-1}, w)   \  \sim  \ \Cc\Nc ( \zerou,   \Omega_{t}) . \]
\end{lemma}
\vspace{10pt}
 \begin{proof}
 The proof is shown in Appendix~\ref{sec:lmdensityh}.  
 \end{proof} 
\vspace{10pt}

\begin{remark} 
In our setting, we consider the case of $\yv_t =  \xvu_t^\T  \hvu_t +   \zv_t$,  where $\xvu_t$ is a deterministic function of $(\yv^{t-1}, \Wme)$ given the encoding maps in \eqref{eq:mapx}.  Lemma~\ref{lm:densityh} reveals that $ \hat{\hvu}_{t}$ is the MMSE estimate of  $\hvu_t$ given  $(\yv^{t-1}, \Wme)$ and  $\hvu_t  | (\yv^{t-1}, \Wme)\sim \Cc\Nc ( \hat{\hvu}_{t},   \Omega_{t})$, where $\hat{\hvu}_{t}$ and $\Omega_{t}$ are defined in \eqref{eq:mmse424} and \eqref{eq:mmse9832} in our setting.
\end{remark}

\vspace{10pt}

\begin{lemma}   \label{lm:estimate1}
Consider any  vector $\evu_i \in \Cc^{M \times 1 }$, $i \in \Zc $, and let
\begin{align}
  \Km_{t}  &  \defeq I_M - \sum_{ i =  1 }^{t-1}   \frac{\Km_{i} \evu_i^{*}  \evu^\T_i  \Km_{i}}{\evu^\T_i \Km_{i} \evu_i^{*}  +1 },    \quad    t =2, 3, 4, \cdots      \label{eq:estm3255}
  \end{align}
and $\Km_{1}    \defeq I_M $, then we have  
\[ \mathbf{0} \preceq \Km_{t}  \preceq  I_M,  \quad \forall t \in \{1,2,3, \cdots \}.\] 
\end{lemma}
\vspace{5pt}
\begin{proof}
 The proof is shown in Appendix~\ref{sec:estimate1}.
 \end{proof}

 \vspace{5pt}
\begin{lemma}   \label{lm:max}
The solution for the following maximization problem 
\begin{align}
\text{maximize} \quad       &\sum_{t=1}^n    \log(1 +    c s_t )  \non\\
 \text{subject to} \quad  &  \sum_{t=1}^n  s_t \leq  m  \non  \\
        &   s_t \geq 0 ,  \quad  t =1,2,\cdots,n  \non
\end{align}
is  $ s^{\star}_1 =  s^{\star}_2 = \cdots =  s^{\star}_n = m/n$, for  constants $m >0 $ and $c > 0$.
\end{lemma} 
 \begin{proof}
 The proof follows directly from Jensen's inequality.   Note that $ f(x) = \log(1 + c x) $ is a concave function.  By using Jensen's inequality, we have 
\[ \frac{1}{n} \sum_{t=1}^n    \log(1 +    c s_t )   \leq     \log \Bigl(1 +      \frac{c}{n} \sum_{t=1}^n  s_t \Bigr)     \]
 which, together with the constraint of $ \sum_{t=1}^n  s_t \leq  m $, gives the bound $ \sum_{t=1}^n    \log(1 +     cs_t )   \leq  n \log (1 +      \frac{c m}{n} ) $. The equality holds when $s^{\star}_1 =  s^{\star}_2 = \cdots =  s^{\star}_n = m/n$.
\end{proof}

 \vspace{5pt}
\begin{lemma}   \label{lm:ExpCht}
For   $\hat{\hvu}_t$ defined   in \eqref{eq:mmse424}, 
we have 
\begin{align}
\E   \bigl[  \|\hat{\hvu}_t\|^2  \bigr]  & \leq     [ (t -1) \  \text{mod} \  T_c]   \label{eq:chbound1}\\ 
\E   \bigl[  \|\hat{\hvu}_t\|^2  \bigr]   &\leq    M     \label{eq:chbound2}  \\ 
\E   \bigl[  \|\hat{\hvu}_t\|^4  \bigr]   &\leq      2 [(t-1) \ \text{mod} \  T_c]^2 +5 [(t-1) \  \text{mod} \  T_c]      \label{eq:chbound3}    \\ 
\E   \bigl[  \|\hat{\hvu}_t\|^4  \bigr]   &\leq     M^2 +2M    \label{eq:chbound4}     \\ 
 \E \bigl[  ( \|\hat{\hvut}\|^2 +1)^2\bigr]  & \leq   \min\bigl\{M^2 + 4M +1 ,  \  \     2 [ (t-1) \  \text{mod} \  T_c]^2 + 7 [(t -1) \ \text{mod} \  T_c] +1    \bigr\}  \label{eq:chbound5}  
\end{align}
for $t=1,2, \cdots, n$. $[ t \  \text{mod} \  T_c]$ denotes a modulo operation.
\end{lemma} 
 \begin{proof}
 The proof is shown in Appendix~\ref{sec:lmExpCht}.
 \end{proof}

\begin{figure}
\centering
\includegraphics[width=10cm]{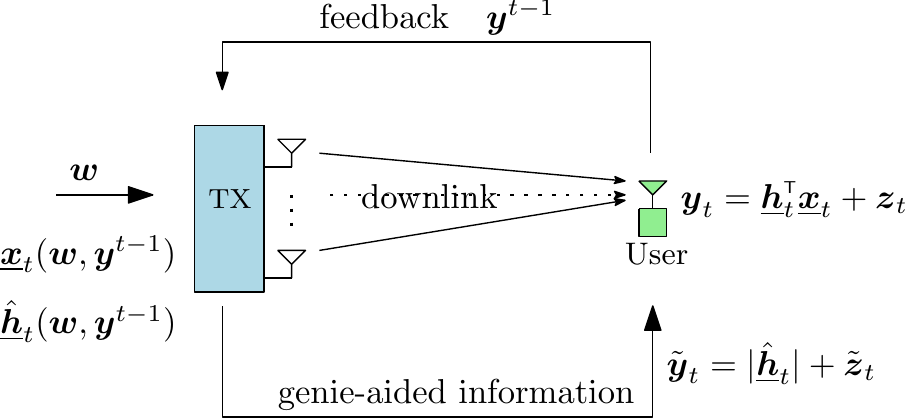}
\caption{A feedback MISO channel with a genie-aided information. The feedback information $\yv^{t-1}$ and the message $\Wme$ are available at the transmitter at time $t$, $t=1,2,\cdots, n$.  The channel outputs $ \yv^{n}$ and genie-aided information $\tilde{\yv}^{n}$ are available at the receiver after time $t=n$.  Both $\xvut$ and $\hat{\hvu}_{t}$ are the functions of $(\Wme, \yv^{t-1})$.}
\label{fig:MassiveMISOgenie}
\end{figure}

 \vspace{5pt}
 
The proof details of each step for Theorem~\ref{thm:MISO2u} are provided as follows.  Recall that the proof of  Theorem~\ref{thm:MISO2u} follows the roadmap mentioned earlier.

\subsection{Step 1:  genie-aided  channel  enhancement  \label{sec:genie} }

For the original channel model defined in Section~\ref{sec:system},  the transmitter obtains the information of $(\Wme, \yv^{t-1})$ at time $t$, while the receiver obtains the information of $\yv^{t}$ at time $t$, for $t=1,2, \cdots,n$.  At the end of the whole communication, i.e., after time $t=n$, the information of  $\yv^{n}$ is available at  the receiver.
We now  enhance the setting by providing a genie-aided information $\tilde{\yv}^{n} \defeq \{ \tilde{\yv}_t \}_{t=1}^n$ to the receiver at the end of the whole communication, where $ \tilde{\yv}_t$ is defined in \eqref{eq:genie5678}.
In the enhanced setting (see Fig.~\ref{fig:MassiveMISOgenie}), the transmitter has the same information as before at each time $t$, but the receiver has more information, i.e., \[ (\yv^{n}, \tilde{\yv}^{n})\] at the end of the whole communication. Therefore, the channel capacity (and its upper bound) of the enhanced setting will serve as the upper bound of the channel capacity of the original setting.  In what follows we will investigate the  capacity upper bound of the enhanced setting.  
As we will see later on,  this  channel enhancement step plays an important role  in deriving our capacity upper bound.

\subsection{Step 2: bound the rate of the enhanced channel     \label{sec:genierate} }

We proceed to bound the rate of the  enhanced setting as follows:
\begin{align}
nR &=  \Hen(\Wme)  \nonumber\\
&=  \Imu(\Wme; \yv^{n}, \tilde{\yv}^{n})   +   \Hen (\Wme| \yv^{n}, \tilde{\yv}^{n})   \nonumber\\
&\leq \Imu(\Wme; \yv^{n}, \tilde{\yv}^{n})  +  n \epsilon_n  \label{eq:3fano3121} \\
&= \sum_{t=1}^n \Imu(\Wme; \yv_t, \tilde{\yv}_t\big|  \yv^{t-1}, \tilde{\yv}^{t-1} )  +  n \epsilon_n  \label{eq:3chain8835} \\
&= \sum_{t=1}^n \Bigl( \Imu(\Wme; \tilde{\yv}_t\big|  \yv^{t-1}, \tilde{\yv}^{t-1} )  + \Imu(\Wme; \yv_t \big|  \yv^{t-1}, \tilde{\yv}^{t} )  \Bigr)+  n \epsilon_n  \label{eq:3chain9256} \\
&= \sum_{t=1}^n \Bigl(  \hen(\tilde{\yv}_t\big|  \yv^{t-1}, \tilde{\yv}^{t-1} ) - \hen(\tilde{\yv}_t\big|  \yv^{t-1}, \tilde{\yv}^{t-1}, \Wme )    + \hen(\yv_t \big|  \yv^{t-1}, \tilde{\yv}^{t} ) - \hen(\yv_t \big|  \yv^{t-1}, \tilde{\yv}^{t}, \Wme )   \Bigr)+  n \epsilon_n  \non \\
&\leq  \!\sum_{t=1}^n \! \Bigl(  \hen(\tilde{\yv}_t\big|  \yv^{t-1}, \tilde{\yv}^{t-1} ) \! - \underbrace{\hen(\tilde{\yv}_t\big|  \yv^{t-1}, \tilde{\yv}^{t-1}, \Wme , \hat{\hvu}_{t}) }_{= \log (\pi e ) }    \!+ \hen(\yv_t \big|  \yv^{t-1}, \tilde{\yv}^{t} ) \!- \underbrace{ \hen(\yv_t \big|  \yv^{t-1}, \tilde{\yv}^{t}, \Wme , \hvut, \xvut)}_{= \log (\pi e ) }    \Bigr) \!+\!  n \epsilon_n  \label{eq:3cond2566}  \\
&=  \sum_{t=1}^n \Bigl(  \hen(\tilde{\yv}_t\big|  \yv^{t-1}, \tilde{\yv}^{t-1} ) - \log (\pi e ) + \hen(\yv_t \big|  \yv^{t-1}, \tilde{\yv}^{t} ) -  \log (\pi e )   \Bigr)+  n \epsilon_n  \label{eq:3cond5823}  \\
&\leq   \sum_{t=1}^n   \hen(\tilde{\yv}_t ) + \sum_{t=1}^n \hen(\yv_t \big|  \tilde{\yv}_{t} )  - 2 n\log (\pi e ) +  n \epsilon_n  \label{eq:3bound1}  
\end{align}
where \eqref{eq:3fano3121} follows from Fano's inequality  and $\epsilon_n \to 0$ as $n \to \infty$;
 \eqref{eq:3chain8835} and \eqref{eq:3chain9256} result from chain rule;
\eqref{eq:3cond2566} and \eqref{eq:3bound1} use  the fact that conditioning reduces differential entropy;
\eqref{eq:3cond5823} is from that $\hen( \yt \big| \yv^{t-1}, \tilde{\yv}^{t}, \Wme , \hvut, \xvut) $ $=  \hen( \yt- \hvut^\T\xvut \big| \yv^{t-1}, \tilde{\yv}^{t}, \Wme , \hvut, \xvut) $ $=  \hen( \zt) = \log (\pi e )$ and that $\hen(\tilde{\yv}_t\big|  \yv^{t-1}, \tilde{\yv}^{t-1}, \Wme , \hat{\hvu}_{t}) $ $= \hen(\tilde{\yv}_t -  |\hat{\hvu}_{t}|  \ \big|  \yv^{t-1}, \tilde{\yv}^{t-1}, \Wme , \hat{\hvu}_{t})  $ $= \hen(\tilde{\zv}_t) =  \log (\pi e )$.

\subsection{Step 3: bound $ \hen(\tilde{\yv}_t )$ by using the differential entropy maximizer  and Lemma~\ref{lm:ExpCht}  \label{sec:ub221} }

We proceed to  upper bound the  differential entropy $ \hen(\tilde{\yv}_t )$ in \eqref{eq:3bound1}, for $t\in \{1,2, \cdots, n\}$. 
Note that the average power of $\tilde{\yv}_t $  is 
\[  \E \bigl[ | \tilde{\yv}_t  |^2  \bigr]   =   \E \bigl[   \bigl|  \| \hat{\hvu}_{t}\| +  \tilde{ \zt} \bigr|^2  \bigr]    =   1 + \E \bigl[   \| \hat{\hvu}_{t}\|^2  \bigr]  \]
(cf.~\eqref{eq:genie5678}).  Since differential entropy is maximized by a circularly symmetric complex Gaussian distribution with the same average power, we have 
\begin{align}
   \hen\bigl( \tilde{\yv}_t   \bigr)   \leq  \log\Bigl(\pi e \bigl(   1 +  \E \bigl[   \| \hat{\hvu}_{t}\|^2 \bigr] \bigl) \Bigr).  \label{eq:3guassian8143}  
\end{align}
From  \eqref{eq:chbound1} and \eqref{eq:chbound2} in Lemma~\ref{lm:ExpCht}   we have
\[ \E \bigl[   \| \hat{\hvu}_{t}\|^2 \bigr] \leq  \min\{ M,  \   [  (t -1) \  \text{mod} \  T_c]\} \]
 which, together with \eqref{eq:3guassian8143}, gives 
\begin{align}
   \hen\bigl( \tilde{\yv}_t   \bigr)  & \leq  \log\Bigl(\pi e \bigl(   1 + \min\{ M,  \   [  (t -1) \  \text{mod} \  T_c]\}   \bigl) \Bigr)  \non   \\ 
    &\leq  \log\Bigl(\pi e \bigl(   1 + \min\{ M,  T_c \}   \bigl) \Bigr)    \label{eq:3guassian8221}  
\end{align}
$\forall t\in \{1,2, \cdots, n\}$, where $[  (\bullet )  \text{mod} \  T_c]$ denotes a modulo operation; and \eqref{eq:3guassian8221}  uses the identity of $[ (t -1) \  \text{mod} \  T_c] \leq T_c$.
Then, combining  \eqref{eq:3bound1}  and  \eqref{eq:3guassian8221} yields the following  bound on the rate: 
\begin{align}
 nR - n \epsilon_n    & \leq   \sum_{t=1}^n  \log \bigl(   1 +  \min\{ M, T_c\}  \bigl)     + \sum_{t=1}^n \hen(\yv_t \big|  \tilde{\yv}_{t} )  -  n\log (\pi e ) .  \label{eq:3bound11}   
\end{align}
Note that the first  term in the right-hand side of \eqref{eq:3bound11} corresponds to  a penalty on the capacity upper bound, due to the  genie-aided  channel enhancement, because it corresponds to the differential entropy of the genie-aided information $\tilde{\yv}^{n}$.  This penalty leads to the factor $2$ in the upper bound of beamforming gain, as shown in Theorem~\ref{thm:MISOb2}, i.e., $b (\alpha)    \leq   \min\{2\alpha,1\}$.

\subsection{Step 4:  bound $\hen( \yv_t \big|  \tilde{\yv}_{t} )$ by using the differential entropy maximizer  \label{sec:ub334} }

Let us now focus on the  conditional differential entropy $\hen( \yv_t \big|  \tilde{\yv}_{t} )$ in \eqref{eq:3bound11}. 
Note that  
\begin{align}
  \hen( \yv_t \big|  \tilde{\yv}_{t} ) = \E_{\tilde{\yv}_{t}} [ \hen( \yv_t \big|  \tilde{\yv}_{t} =  \tilde{y}_{t} )  ].  \label{eq:3guassian3355}  
\end{align}
Again, by using the fact that  Gaussian distribution with the same average power  maximizes the differential entropy, we have 
\begin{align}
   \hen\bigl(  \yv_t \big|    \tilde{\yv}_{t} =  \tilde{y}_{t} \bigr)   &\leq  \log\Bigl(\pi e \cdot \E \bigl[ | \yv_t  |^2  |  \tilde{\yv}_{t} =  \tilde{y}_{t}  \bigr]  \Bigr) \non\\
   &= \log\Bigl(\pi e \bigl(  1 + \E \bigl[ |\hvut^\T \xvut |^2   |  \tilde{\yv}_{t} =  \tilde{y}_{t}  \bigr]    \bigl) \Bigr)   \label{eq:3guassian0099}  
\end{align}
 which, together with \eqref{eq:3guassian3355},  yields   
\begin{align}
  \hen( \yv_t \big|  \tilde{\yv}_{t} )  \leq    \E
  \Bigl[  \log\Bigl(\pi e \bigl(  1 + \E \bigl[ |\hvut^\T \xvut |^2   |  \tilde{\yv}_{t}   \bigr]    \bigl) \Bigr)   \Bigr]. \label{eq:3guassian8877}  
\end{align}

\subsection{Step 5: bound $\E \bigl[ |\hvut^\T \xvut |^2  |  \tilde{\yv}_{t} \bigr]$ by dealing with the correlation between $\xvut$ and $\hvut$   \label{sec:ub443} }

Since $\xvut$ and $\hvut$ are correlated, computing the value of $\E \bigl[ |\hvut^\T \xvut |^2   |  \tilde{\yv}_{t} \bigr]$ (shown in \eqref{eq:3guassian8877}) could be challenging in general. 
We now bound the value of $\E \bigl[ |\hvut^\T \xvut |^2  |  \tilde{\yv}_{t} \bigr]$ as follows: 
\begin{align}
 \E \Bigl[ |\hvut^\T \xvut |^2  \big|  \tilde{\yv}_{t}\Bigr] 
& = \E \Bigl[ \E   \Bigl[ |\hvut^\T \xvut |^2  \big| \yv^{t-1}, \Wme,  \tilde{\yv}_{t}  \Bigr]   \Big|  \tilde{\yv}_{t}  \Bigr]  \label{eq:3exex} \\
& = \E \Bigl[  \E   \Bigl[ |\hvut^\T \xvut |^2  \big| \yv^{t-1}, \Wme  \Bigr]   \Big|  \tilde{\yv}_{t}  \Bigr]  \label{eq:3exexmarkov} \\
& =  \E \Bigl[ \E   \Bigl[  \trace \bigl( \hvut^{*} \hvut^\T \xvut \xvut^{\H} \bigr)       \big| \yv^{t-1}, \Wme \Bigr]  \Big |  \tilde{\yv}_{t}  \Bigr]  \label{eq:1trace221} \\
& =  \E \Bigl[ \trace \Bigl( \E \bigl[   \hvut^{*} \hvut^\T   \big| \yv^{t-1}, \Wme \bigr] \cdot \xvut \xvut^{\H} \Bigr)   \Big |  \tilde{\yv}_{t}  \Bigr]  \label{eq:1traceex52} \\
& =  \E \Bigl[ \trace \Bigl( \E \Bigl[   (\hat{\hvut} + \tilde{\hvut})^{*} (\hat{\hvut} + \tilde{\hvut})^\T   \big| \yv^{t-1}, \Wme \Bigr] \cdot \xvut \xvut^{\H} \Bigr) \Big |  \tilde{\yv}_{t}   \Bigr]  \label{eq:1exph93} \\
& =  \E \Bigl[ \trace \Bigl( \E \Bigl[   \hat{\hvut}^{*}\hat{\hvut}^\T  + \tilde{\hvut}^{*}\tilde{\hvut}^\T  +\hat{\hvut}^{*}\tilde{\hvut}^\T  +  \tilde{\hvut}^{*} \hat{\hvut}^\T  \  \big| \yv^{t-1}, \Wme \Bigr] \cdot  \xvut \xvut^{\H}  \Bigr) \Big |  \tilde{\yv}_{t}  \Bigr]  \nonumber \\
& =  \E \Bigl[ \trace \Bigl( \bigl(\hat{\hvut}^{*}\hat{\hvut}^\T  +\Omega_{t} \bigr) \cdot  \xvut \xvut^{\H}     \Bigr) \Big |  \tilde{\yv}_{t}  \Bigr]  \label{eq:1hmean} \\
& =  \E \Bigl[ \trace \Bigl( \hat{\hvut}^{*}\hat{\hvut}^\T \xvut \xvut^{\H}  \Bigr)   + \trace \Bigl(    \Omega_{t}  \xvut \xvut^{\H}    \Bigr)  \Big |  \tilde{\yv}_{t}  \Bigr]  \nonumber \\
&\leq \E \Bigl[     \lambda_{\max}(\hat{\hvut}^{*}\hat{\hvut}^\T) \cdot \trace(\xvut \xvut^{\H})   + \lambda_{\max}( \Omega_{t})  \cdot \trace(\xvut \xvut^{\H})    \Big |  \tilde{\yv}_{t}     \Bigr]  \label{eq:1trmax54} \\
&\leq   \E \Bigl[    \|\hat{\hvut}\|^2 \cdot \| \xvut \|^2   +  \| \xvut \|^2     \Big |  \tilde{\yv}_{t}  \Bigr]  \label{eq:3trmax366} 
\end{align}
where \eqref{eq:3exex} follows from the identity that $\E[\boldsymbol{a} | \boldsymbol{c}] = \E[\E[\boldsymbol{a} | \boldsymbol{b}, \boldsymbol{c}] | \boldsymbol{c}]$ for any three random variables $\boldsymbol{a}$, $\boldsymbol{b}$ and  $\boldsymbol{c}$;
\eqref{eq:3exexmarkov} stems from the Markov chain of  $  \{\yv^{t-1}, \Wme,  \tilde{\yv}_{t}\}  \to \{\yv^{t-1}, \Wme\} \to \{\hvut,  \xvut\}$; remind that $\tilde{\yv}_{t} = \| \hat{\hvu}_{t}\| +  \tilde{ \zt}$, and both $\hat{\hvu}_{t}$ and $\xvut$ are deterministic functions of $( \yv^{t-1}, \Wme )$ given the encoding maps in \eqref{eq:mapx};
\eqref{eq:1trace221} results from the fact that $|\hvut^\T \xvut |^2 = \trace (\hvut^\T \xvut \xvut^\H \hvut^{*} ) =  \trace (\hvut^{*} \hvut^\T \xvut \xvut^\H )$ by using the identity of $\trace( A B) = \trace( BA)$ for any matrices $A \in \Cc^{m \times q}$, $B\in \Cc^{q \times m}$;
\eqref{eq:1traceex52} stems from the fact that $\xvut$ is a deterministic function of $( \yv^{t-1}, \Wme )$; in \eqref{eq:1exph93} we just replace $\hvut$ with $\hat{\hvut} + \tilde{\hvut} $, where $\hat{\hvut}$  and $\tilde{\hvut}$ are defined in \eqref{eq:mmse424}-\eqref{eq:mmse924};
\eqref{eq:1hmean} results from the fact that  $\hat{\hvut}$ is a deterministic function of $(\yv^{t-1}, \Wme)$,  and the fact that the conditional density of $ \tilde{\hvut}$ given $(\yv^{t-1}, \Wme)$ is 
\[   \tilde{\hvut} \  \big| \  (\yv^{t-1}, \Wme)   \quad  \sim    \quad     \Cc\Nc (  \zerou ,   \Omega_{t})  \]
(see Lemma~\ref{lm:densityh});
 \eqref{eq:1trmax54} follows from the identity that  $\trace( A B)  \leq  \lambda_{\max}(A ) \trace(B)$, where $\lambda_{\max}(A)$ corresponds to the maximum eigenvalue of matrix $A$, for positive semidefinite $m\times m$ Hermitian matrices $A$ and $B$; 
 \eqref{eq:3trmax366} results from the facts that $\lambda_{\max}(\hat{\hvut}^{*}\hat{\hvut}^\T) = \|\hat{\hvut}\|^2$, $\trace(\xvut \xvut^{\H}) = \| \xvut\|^2$, and  $\lambda_{\max}( \Omega_{t}) \leq 1$  (see Lemma~\ref{lm:estimate1}).

\subsection{Step 6: bound $\hen( \yv_t \big|  \tilde{\yv}_{t} )$ by dealing with the correlation between $\xvut$ and $\hat{\hvut}$   \label{sec:ub554} }

By plugging \eqref{eq:3trmax366} into \eqref{eq:3guassian8877}, it yields 
 \begin{align}
  \hen( \yv_t \big|  \tilde{\yv}_{t} )  \leq    \E
  \Bigl[  \log\Bigl(\pi e \bigl(  1 +     \E \bigl[    \|\hat{\hvut}\|^2 \cdot \| \xvut \|^2   +  \| \xvut \|^2       \big|  \tilde{\yv}_{t}  \bigr]    \bigl) \Bigr)   \Bigr]. \label{eq:3guassian3452}  
\end{align}
Note that  $\xvut$ and $\hat{\hvut}$ are correlated.  Without conditioning on $ \tilde{\yv}_{t}$ (a genie-aided information),  it is challenging to derive a tight bound on the expectation term  $\E \bigl[    \|\hat{\hvut}\|^2 \cdot \| \xvut \|^2  \bigr]$ in \eqref{eq:3guassian3452}.
In this step, we take  the benefit of genie-aided channel enhancement ---  which leads to a condition $\tilde{\yv}_{t}$ in the expectation term  $\E \bigl[    \|\hat{\hvut}\|^2 \cdot \| \xvut \|^2  \bigr]$ in \eqref{eq:3guassian3452} ---   and provide an upper bound  on $\hen( \yv_t \big|  \tilde{\yv}_{t} )$. 

 Since $\tilde{\yv}_{t} = \| \hat{\hvu}_{t}\| +  \tilde{ \zt}$ (cf.~\eqref{eq:genie5678}), we could bound  $\|\hat{\hvut}\|^2$ by using triangle inequality: 
  \begin{align}
   \|\hat{\hvut}\|^2  = | \tilde{\yv}_{t}  - \tilde{\zv}_{t} |^2  \leq   ( | \tilde{\yv}_{t}| + | \tilde{\zv}_{t} |) ^2     \label{eq:3triangle9952}  
\end{align}
which, together with \eqref{eq:3guassian3452}, gives
 \begin{align}
  \hen( \yv_t \big|  \tilde{\yv}_{t} )  - \log(\pi e)  &\leq    \E  \Bigl[  \log\Bigl( \E \Bigl[ 1 +       ( | \tilde{\yv}_{t}|  + | \tilde{\zv}_{t} |)^2 \cdot \| \xvut \|^2   +  \| \xvut \|^2       \Big|  \tilde{\yv}_{t}  \Bigr]   \Bigr)   \Bigr]    \label{eq:3triangle8255}   \\
    &\leq    \E  \Bigl[  \log\Bigl( \E \Bigl[ \bigl(  1+  | \tilde{\yv}_{t}| \bigr)^2 \cdot  \bigl(  1+    ( 1+    | \tilde{\zv}_{t} |)^2 \cdot \| \xvut \|^2  \bigr)  \Big|  \tilde{\yv}_{t} \Bigr]   \Bigr)   \Bigr] \label{eq:3bound2234}    \\
        & =    2  \cdot \E  \bigl[ \log  \bigl(  1+  | \tilde{\yv}_{t}| \bigr)  \bigr]  +  \E  \Bigl[   \log\bigl(   \E \bigl[   1+    ( 1+    | \tilde{\zv}_{t} |)^2 \cdot \| \xvut \|^2   \big|  \tilde{\yv}_{t} \bigr]   \bigr)   \Bigr]    \non    \\
       & \leq    2 \cdot \log  \bigl(  1+   \E [  | \tilde{\yv}_{t}| ] \bigr)  + \log\bigl(   \underbrace{ \E \bigl[   1+    ( 1+    | \tilde{\zv}_{t} |)^2 \cdot \| \xvut \|^2   \bigr] }_{= 1+   \E [ ( 1+    | \tilde{\zv}_{t} |)^2 ]  \cdot \E [ \| \xvut \|^2   ] }   \bigr)   \label{eq:3jensen2244}    \\
       & =     2 \cdot \log  \bigl(  1+   \E [  | \tilde{\yv}_{t}| ] \bigr)  + \log\bigl(      1+   \underbrace{\E \bigl[ ( 1+    | \tilde{\zv}_{t} |)^2 \bigr] }_{\leq 4}  \cdot \E \bigl[ \| \xvut \|^2   \bigr]   \bigr)   \label{eq:3ind4324}    \\
            & \leq      2 \cdot \log  \bigl(  1+  \E [  | \tilde{\yv}_{t}| ] \bigr)  + \log\bigl(      1+   4  \cdot \E \bigl[ \| \xvut \|^2   \bigr]   \bigr)   \label{eq:3ind2577}    \\  
        & \leq      2 \cdot \log  \bigl(  1+   \sqrt{  \E [  | \tilde{\yv}_{t}|^2 ] } \bigr)  + \log\bigl(      1+   4 \cdot \E \bigl[ \| \xvut \|^2   \bigr]   \bigr)     \label{eq:3triangle8386}    \\
               & \leq     2 \cdot \log  \bigl(  1+   \sqrt{ \min\{ M,  T_c\}  +  1 } \bigr)  + \log\bigl(      1+   4 \cdot \E \bigl[ \| \xvut \|^2   \bigr]   \bigr)     \label{eq:3triangle32466}  
\end{align}
where \eqref{eq:3triangle8255} is from \eqref{eq:3guassian3452} and \eqref{eq:3triangle9952};
\eqref{eq:3bound2234} follows from the identity that $1+ (a_1+ a_2)^2 \cdot a_3 + a_3\leq $ $ (1+a_1)^2\cdot  \bigl(1+ (1+a_2)^2\cdot a_3\bigr)$ for any $a_1, a_2, a_3 \geq 0$; 
\eqref{eq:3jensen2244} results from Jensen's inequality;
\eqref{eq:3ind4324} follows from the independence between $ \tilde{\zv}_{t}$ and $\xvut$;
\eqref{eq:3ind2577} stems from the fact that $\E \bigl[ ( 1+    | \tilde{\zv}_{t} |)^2 \bigr] = 1  +\E [ | \tilde{\zv}_{t} |^2] +  2 \E [ | \tilde{\zv}_{t} |]    =  2 + \sqrt{\pi}  \leq 4$, given that $ \E [ | \tilde{\zv}_{t} |] =  \sqrt{\pi}/2$ for $ \tilde{ \zt}   \sim  \Cc\Nc (  0 ,   1) $;  
\eqref{eq:3triangle8386} uses the fact that  $ 0\leq   \E [  | \tilde{\yv}_{t}|^2 ]  - ( \E [  | \tilde{\yv}_{t}| ] )^2 $  since  $0 \leq  \E (|\tilde{\yv}_{t}| -  \E | \tilde{\yv}_{t}| )^2=  \E [  | \tilde{\yv}_{t}|^2 ]  - ( \E [  | \tilde{\yv}_{t}| ] )^2$;
\eqref{eq:3triangle32466}  follows from the fact that $\E [  | \tilde{\yv}_{t}|^2 ] = \E [  | \| \hat{\hvu}_{t}\|  + \tilde{ \zt} |^2 ] = \E [ \| \hat{\hvu}_{t}\|^2] + \E[ |  \tilde{ \zt} |^2 ] = \E [ \| \hat{\hvu}_{t}\|^2] + 1 \leq    \min\{ M,  \   [  (t -1) \  \text{mod} \  T_c]\}  +  1 \leq   \min\{ M,  T_c\}  +  1$, by using \eqref{eq:chbound1} and \eqref{eq:chbound2} in Lemma~\ref{lm:ExpCht}.

 \subsection{Step 7:  derive  a final capacity upper bound with  the optimal power allocation   \label{sec:ub665} }

Finally, by plugging \eqref{eq:3triangle32466}  into \eqref{eq:3bound11} we have 
\begin{align}
 nR \!-\! n \epsilon_n  &\leq   \sum_{t=1}^n  \log \bigl(   1 +  \min\{ M,  T_c\}  \bigl)      
  +   \sum_{t=1}^n    \log  \bigl(  1+   \sqrt{  \min\{ M,  T_c\}  +  1 } \bigr)^2     +   \sum_{t=1}^n  \log\bigl(      1+   4  \E \bigl[ \| \xvut \|^2   \bigr]   \bigr)   \label{eq:3bound211}    \\
    &\leq   2n  \log \bigl(   1 +  3(\min\{ M,  T_c\} +1)  \bigl)     +   \sum_{t=1}^n  \log\bigl(      1+   4  \E \bigl[ \| \xvut \|^2   \bigr]   \bigr)   \label{eq:3bound3652}    \\
& \leq    2n  \log \bigl(   1 +  3(\min\{ M,  T_c\} +1)  \bigl)   +   \max_{   \sum_{t =1}^n \E[ \|\xvu_{t}\|^2 ] \leq n  P  }   \sum_{t=1}^n  \log\bigl(      1+   4 \E \bigl[ \| \xvut \|^2   \bigr]   \bigr)  \label{eq:3bound662}    \\
 & \leq   2n  \log \bigl(   1 +  3(\min\{ M,  T_c\} +1)  \bigl)   +  n   \log\bigl(      1+    4 P    \bigr)   \label{eq:3power999}    
\end{align}
where \eqref{eq:3bound211} follows from  \eqref{eq:3bound11} and \eqref{eq:3triangle32466};
\eqref{eq:3bound3652} stems from the identity that $(1+a) < (1+\sqrt{1+a})^2 \leq 1+ 3(a+1)  $ for any $a\geq 0$; 
\eqref{eq:3bound662} results from   maximizing  the RHS of \eqref{eq:3bound3652}  under a second moment constraint; 
\eqref{eq:3power999} follows from Lemma~\ref{lm:max}. 
At this point,  as $n\to \infty$, we have the bound  $R  \leq   2 \log \bigl(  4 +  3 \cdot \min\{ M,  T_c\}  \bigl)   +   \log\bigl(      1+    4 P    \bigr)  $ and complete the proof.

\section{Achievability \label{sec:misoc} }

This section will prove the following theorem for the MISO channel defined in Section~\ref{sec:system}.
\vspace{5pt}
 \begin{theorem} [Lower bound] \label{thm:MISOlower}
For the MISO  channel  defined in Section~\ref{sec:system}, the  capacity  is lower bounded by
\[ C \geq   \! \frac{T_c \!- \!\Tt}{T_c}  \cdot  \log \Bigl(   1  + \frac{ P  \cdot \max\{ \! (\Tt \! -\! 1) , 1/2 \} }{  2  + \frac{1}{P}    } -    \frac{1 }{   \max\{ \Tt , 2  \}  }  \Bigr)   \] 
under the second moment input constraint; while under the fourth moment input constraint, the  capacity  is lower bounded by
\[ C \geq  \! \frac{T_c \!-\! \Tt}{T_c}  \cdot  \log \Bigl(   1  + \frac{ P_o  \cdot \max\{ \!(\Tt \!-\! 1) , 1/2 \} }{  2  + \frac{1}{P_o}    } -    \frac{1 }{   \max\{ \Tt , 2  \}  }  \Bigr)     \]  
  where  $P_o \defeq \frac{\kappa P}{\sqrt{3}}$ and  $\Tt \defeq \lceil \frac{\min\{M, T_c\}}{ \log  \max\{4, \  \min\{M, T_c\}\}}\rceil  $. 
\end{theorem}
\vspace{5pt}

\vspace{5pt}
\begin{remark} [Proof of Theorems~\ref{thm:MISOb2} and~\ref{thm:MISOb4}, lower bound]  \label{rmk:prooflb24}
From the capacity lower bounds in Theorem~\ref{thm:MISOlower}, one can easily derive  a  lower bound of the beamforming gain:
\begin{align}
b(\alpha)   &\geq  \min \{1  ,  \alpha\}   \non
\end{align}
under each of the two input constraints (cf.~\eqref{eq:powerEve} and~\eqref{eq:power4m}). 
It then proves the beamforming gain  lower bounds for Theorem~\ref{thm:MISOb2} and Theorem~\ref{thm:MISOb4}.
\end{remark}
\vspace{5pt}

Specifically, an achievability scheme is provided in this section  for the MISO channel with feedback. To this end, the proposed scheme can achieve a rate  $R$ (bits/channel use)  that is lower bounded by
\begin{align}
R \geq   \frac{T_c - \Tt}{T_c}  \cdot  \log \Bigl(   1  + \frac{ P  \cdot \max\{ (\Tt -1) , 1/2 \} }{  2  + \frac{1}{P}    } -    \frac{1 }{   \max\{ \Tt , 2  \}  }  \Bigr)     \label{eq:rate222}          
\end{align}
under the \emph{second} moment input constraint.  
For the case with a \emph{fourth} moment input constraint (cf.~\eqref{eq:power4m}), the proposed scheme achieves the similar rate 
$R$ with difference being that in the latter case $P$ is replaced with  $P_o$.  Note that by replacing the input power  $P$ with $P_o$, the proposed scheme will satisfy the  \emph{fourth} moment input constraint and achieve the declared rate.  In the following we will just describe the scheme for the case with a \emph{second} moment input constraint. 
Note that the  lower bounds in Theorem~\ref{thm:MISOlower} can be further improved since we just focus on the \emph{simple} scheme.

The proposed scheme is  a simple scheme that uses  no more than  $T_c$ number of transmit-antennas.  
The scheme consists of a downlink training phase and a data transmission phase for \emph{each} coherence block of the channel (see Fig.~\ref{fig:MISOTtTd}).
The choice of  phase duration is critical to the scheme performance,  because  with too small duration for training phase there is not enough time for the channel training,  while with too large duration for training phase there is not enough time for the data transmission.  
In this scheme we set the  durations of the training phase and  data transmission phase  as  
\begin{align}
\Tt = \Bigl\lceil \frac{\min\{M, T_c\}}{ \log  \max\{4, \ \min\{M, T_c\}\}}\Bigr\rceil  ,   \quad  \quad   \Td =T_c  - \Tt   \label{eq:TtTd}
\end{align}
respectively.  
The above design of $\Tt$ implies that,  the training phase takes  a relatively small fraction of the channel coherence length. Considering the typical case of $M\geq T_c$, this fraction is  roughly $\frac{1}{ \log T_c}$, which can be ignored when $M$ and $T_c$ are very large.  On the other hand, we  show that this small fraction of channel coherence length is sufficient to obtain a relatively good channel training and achieve  a relatively good beamforming gain performance (see Theorems~\ref{thm:MISOb2} and~\ref{thm:MISOb4}).  We conjecture that the achievable beamforming gain is optimal  for the setting with second moment input constraint. Note that,  for the setting with the fourth moment input constraint,   the achievable beamforming gain is optimal (see Theorem~\ref{thm:MISOb4}).
Without loss of generality  we focus on the scheme description for the \emph{first} channel block, 
corresponding to the time index $t\in \{1,2,\cdots,T_c\}$.  Note that $\hvut[1]=\hvut[2]=\cdots =\hvut[T_c]$ and  $\hvut[1]= [\hvot[1], \hvot[2], \cdots, \hvot[M]]^\T.$

\begin{figure}
\centering
\includegraphics[width=9cm]{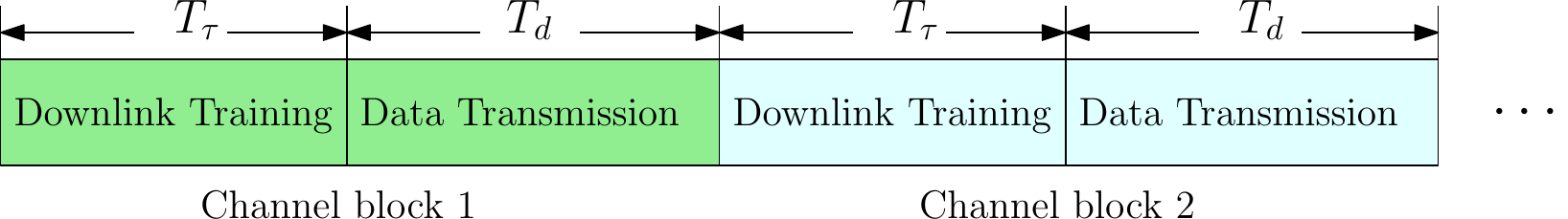}
\caption{The model of downlink training and data transmission, where the downlink training and data transmission are operated over $\Tt$ and $T_d$ channel uses of each channel block.}
\label{fig:MISOTtTd}
\end{figure}

\subsection{Downlink training  \label{sec:training}}

The goal of the downlink training  phase with feedback is to allow both user and transmitter to   learn the channel state information. 
At time $t$, $t \in \{1,2, \cdots, \Tt\}$, the downlink training is operated over the $t$th transmit-antenna in order to estimate the channel $\hvot[t]$,  where $ \hvot[t]$ denotes the channel coefficient between the $t$th transmit antenna and the user during the first channel block. By setting the pilot signal as  $\xvut =\sqrt{P} [ 0 , 0 , \cdots, 0 , 1 , 0,  \cdots, 0 ]^\T$, where the nonzero value is placed at the $t$th element,  then the received signal of user at time $t$ is given as 
\begin{align}
\yt &=   \sqrt{P} \hvot[t] +   \zt, \quad t=1,2, \cdots, \Tt.  \label{eq:ph1t}
\end{align}
As a result,  the user observes $\Tt$ channel training outputs that can be written in a vector form: 
\begin{align}
\yvu_{\tau} &=   \sqrt{P} \hvu_{\tau} +   \zvu_{\tau},  \label{eq:ph1y}
\end{align}
where $\yvu_{\tau} \defeq [ \yt[1], \yt[2], \cdots, \yt[\Tt] ]^\T$, $\hvu_{\tau} \defeq [ \hvot[1], \hvot[2] \cdots, \hvot[\Tt] ]^\T$ and $\zvu_{\tau} \defeq [ \zt[1], \zt[2], \cdots, \zt[\Tt]]^\T.$

After receiving the channel training outputs,  the user can estimate  channel $\hvu_{\tau}$ with MMSE estimator:
\begin{align}
\hat{\hvu}_{\tau} &=      \frac{\sqrt{P} }{P+1} \yvu_{\tau}  .  \label{eq:ph1csit}
\end{align}
The MMSE estimate $\hat{\hvu}_{\tau}$ and estimation error $\tilde{\hvu}_{\tau} \defeq \hvu_{\tau} - \hat{\hvu}_{\tau}$ are  two \emph{independent} complex Gaussian vectors,  
where  
$\hat{\hvu}_{\tau} \sim \Cc\Nc(\zerou,  \frac{P}{P+1} I)$ and $\tilde{\hvu}_{\tau} \sim \Cc\Nc(\zerou, \frac{1}{P+1} I )$.

After MMSE estimation,  the user  feeds back the value of $\hat{\hvu}_{\tau}$ to the transmitter over an independent feedback link (the transmitter can also obtain the MMSE estimate $\hat{\hvu}_{\tau}$ if the user feeds back the channel outputs to the transmitter).

\subsection{Data transmission   \label{sec:data}}

After obtaining the channel state information of $\hat{\hvu}_{\tau}$ (CSIT), the transmitter  sends the data information with linear precoding:  \[\xvut = \sqrt{P} \frac{\hat{\hvu}_{\tau}^{*}}{\|\hat{\hvu}_{\tau}\|} \st,  \quad t  = \Tt+1,  \Tt+2, \cdots,  T_c\] (focusing on the first channel block), where $\st$ denotes the information symbol with unit average power. 
The corresponding  signal received at the user is given as:
\begin{align}
\yv_t &=   \hvu_{\tau}^\T \xvu_t  + \zv_t     \non  \\
    & = \sqrt{P}   (  \hat{\hvu}_{\tau} +\tilde{\hvu}_{\tau})^\T  \frac{\hat{\hvu}_{\tau}^{*}}{\|\hat{\hvu}_{\tau}\|} \sv_t  + \zv_t    \non\\
  & = \sqrt{P} \|\hat{\hvu}_{\tau}\| \sv_t  +  \sqrt{P} \tilde{\hvu}_{\tau}^\T  \frac{\hat{\hvu}_{\tau}^{*}}{\|\hat{\hvu}_{\tau}\|} \sv_t   + \zv_t , \quad  t  = \Tt+1,  \Tt+2, \cdots,  T_c \label{eq:datay}
\end{align}
(again, focusing on the first channel block).
The channel input-output relationship in \eqref{eq:datay} can be further expressed in a vector form:
\begin{align}
\yvu_d  & = \sqrt{P} \|\hat{\hvu}_{\tau}\| \svu_d +  \sqrt{P} \tilde{\hvu}_{\tau}^\T  \frac{\hat{\hvu}_{\tau}^{*}}{\|\hat{\hvu}_{\tau}\|} \svu_d  + \zvu_d    \label{eq:datayv}
\end{align}
where $\yvu_{d} \defeq [ \yt[ \Tt+1], \yt[ \Tt+2], \cdots, \yt[ T_c] ]^\T$, $\svu_d \defeq [  \sv_{\Tt+1} , \sv_{\Tt+2} \cdots,\sv_{T_c} ]^\T$ and $\zvu_{d} \defeq [ \zt[\Tt+1], \zt[\Tt+2], \cdots, \zt[T_c]]^\T.$
Note that the conditional distribution of $\tilde{\hvu}_{\tau}^\T  \frac{\hat{\hvu}_{\tau}^{*}}{\|\hat{\hvu}_{\tau}\|}$ given $\hat{\hvu}_{\tau}$  is  a Gaussian distribution, that is, $\tilde{\hvu}_{\tau}^\T  \frac{\hat{\hvu}_{\tau}^{*}}{\|\hat{\hvu}_{\tau}\|} \ \big|  \hat{\hvu}_{\tau}  \sim \Cc\Nc(0, \frac{1}{P+1}  )$.

\vspace{5pt}

\underline{\emph{Rate analysis}:}  \
We now analyze the achievable rate of the proposed scheme.
At first we assume that the input symbol $\sv_t, \forall t,$ is  circularly symmetric complex Gaussian distributed, i.e., $\sv_t \sim \Cc\Nc ( 0, 1)$, and is independent of $\hat{\hvu}_{\tau}$ and $\hvu_{\tau}$. The following  proposition provides a \emph{lower bound} on the achievable ergodic rate. 
 \vspace{5pt}
\begin{proposition}  \label{pro:rate}
The achievable ergodic rate for the scheme with Gaussian input, training and feedback, and data transmission as described in Sections~\ref{sec:training} and \ref{sec:data} is bounded  as 
\begin{align*}
R &\geq     \frac{T_c - \Tt}{T_c}  \cdot  \log \Bigl(   1  + \frac{ P  \cdot \max\{ (\Tt -1) , 1/2 \} }{  2  + \frac{1}{P}    } -    \frac{1 }{   \max\{ \Tt , 2  \}  }  \Bigr)       
\end{align*}
under the second moment input constraint (cf.~\eqref{eq:powerEve}), where $\Tt = \lceil \frac{\min\{M, T_c\}}{ \log  \max\{4, \min\{M, T_c\}\}}\rceil$. 
\end{proposition}
\vspace{5pt}
\begin{proof}
The proof is shown in Appendix~\ref{sec:prorate}.
\end{proof}

\section{Conclusion and discussion} \label{sec:concl}

In this work we provide capacity bounds for the MISO block fading channel with a noiseless feedback link, under the second  and fourth moment input constraints, respectively. 
The result reveals that, increasing the  transmit-antenna number $M$  to infinity  will \emph{not}  yield an infinite capacity,  for the case with a finite coherence length and a finite input constraint on the second or fourth moment. 
In addition to the capacity bounds, this work also provides a characterization on the channel's beamforming gain for some cases. 
Specifically, for the case with a finite fourth-moment input constraint, the result reveals that $\alpha=1$ is sufficient for achieving a full beamforming gain. When $0\leq \alpha \leq 1$, the beamforming gain increases linearly with  $\alpha$.
The result has provided some practical insights for the massive MIMO system operating with FDD mode where  transmitter and receiver acquire the CSIT/CSIR via  downlink training  and feedback.  One practical insight provided in this work is that,  using more transmit antennas than the coherence length does not yield a significant gain in capacity in an asymptotic sense,   under a finite fourth-moment input constraint. 

In what follows we compare  our work with some previous works, and discuss the difficulty of  our converse proof and  the extension to the multiuser broadcast channel.

\subsection{Comparison between our work and some previous works}

In this work, we focus on the MISO block fading channel with a noiseless feedback link, where the transmitter and receiver have no prior knowledge of the channel state realizations, but the transmitter and receiver can acquire the CSIT/CSIR via downlink training  and feedback. 

In the direction with channel training and feedback,  the previous work in \cite{SH:10} has considered, among others, a MISO block fading channel with dedicated training and limited feedback, under the assumptions of \emph{linear} coding schemes and  a \emph{fixed ratio} $T_c/M$, corresponding to a specific case of $\alpha=1$ in our setting.  For that MISO setting with linear coding schemes and $\alpha=1$, the work in \cite{SH:10}  showed that the (linear) capacity is scaled as $\log M + o(\log M)$, or equivalently, the corresponding beamforming gain is $b=1$,  which matches our beamforming gain lower bound when $\alpha=1$.   
In fact, our beamforming gain lower bound is achieved by a  simple linear scheme that holds for any $\alpha \in [0, \infty)$.
So far, it remains open if  the linear schemes are optimal in terms of the beamforming gain, under the second moment input constraint. There is still a gap between our beamforming gain upper and lower bounds. We conjecture that the lower bound is tight and the linear schemes could be optimal in terms of the beamforming gain.
For the other case with the fourth moment input constraint, our derived beamforming gain upper bound reveals that the linear schemes, including the scheme proposed in \cite{SH:10},  indeed can be optimal in terms of the beamforming gain.

In the direction with channel training and feedback,  the other previous work in \cite{Caire+:10m} investigated the achievable ergodic rates of a MIMO block fading broadcast channel with dedicated training and noisy feedback, under the assumption of \emph{linear} coding schemes. Specifically, the work in \cite{Caire+:10m} derived  the lower  and upper bounds of the achievable rate  as the \emph{expectation} of some functions of the channel estimates.  
In our work we consider a different setting, i.e., a  MISO block fading channel with a noiseless feedback link, without the assumption of  linear coding schemes.  
In our setting, computing the capacity might be NP-hard \cite{Kramer:14} (see the discussion in the following subsection). 
Therefore, we mainly focus on the beamforming gain and the derived bounds depend on the parameter $\alpha$  only.
Furthermore, in the setting considered by \cite{Caire+:10m}, the  time overhead of the channel training is not taken into account in the rate analysis. However, in our setting with a large number of antennas, the time overhead of the channel training might be significant and cannot be ignored.

\subsection{Difficulty of  the converse proof}

The challenge of our proof is mainly due to the correlation between the channel inputs and the  channel outputs (see \eqref{eq:mapx} and \eqref{eq:mapcodeword}), and the high dimension of the channel inputs, equipped with a large number of antennas.
Our channel can be considered as a specific block fading channel with in-block memory, in which the capacity is generally \emph{NP-hard} to compute \cite{Kramer:14}. 
Specifically, the capacity of our setting is a \emph{multiletter} expression 
\begin{align*}
C =    \max_{p_{\xvu^{T_c}}}  \Imu(\xvu^{T_c}; \yv^{T_c})  /T_c     
\end{align*}
and finding  the optimal input distribution $p_{\xv^{T_c}}$ is NP-hard  \cite{Kramer:14}. Recall that  the channel input  $\xvu_t \in \Cc^{M}$ at each time $t$ in our setting is a function of the previous channel outputs and the message, i.e., $\xvu_t(\Wme, \yv^{t-1})$. 

Note that, under the assumptions of  linear coding schemes and a dedicated channel training, bounding the capacity (or called as the achievable rate, cf.~\cite{Caire+:10m}) may be reduced to  bounding a \emph{single-letter} expression (cf.~\cite{SH:10,Caire+:10m}). 
For example, let us consider a  setting with linear coding schemes and a dedicated channel training, such as:  1) at first a certain fraction of each channel block is used  for the channel training; 2) the transmitter and receiver(s) acquire the CSIT/CSIR from those training observations only;  3) the remaining fraction of the channel block is used for data transmission only, under the linear coding strategy. Then, after the channel training phase, the channel can be considered as a \emph{non-feedback} channel with imperfect CSIT/CSIR. In that case, the (linear) capacity bound can be reduced to a \emph{single-letter} expression (cf.~\cite{SH:10,Caire+:10m}).
However, in our setting, feedback \emph{cannot} be removed at any point of time.  Therefore, the previous approaches used in the settings with linear schemes and dedicated channel training (cf.~\cite{SH:10,Caire+:10m}) might not be directly applied in our setting.

In our converse proof, we transform the \emph{NP-hard} capacity problem into a relaxed problem that is  computable.  Note that we focus on the beamforming gain performance, as tight capacity bounds are still hard to compute. 
In our proof, a genie-aided channel enhancement is applied.  Although the  genie-aided channel enhancement  leads to a penalty on the beamforming gain, it is an important step that allows us to bound the  involved terms in a computable way. 
Our difficulty lies in Steps~3-6 (see Section~\ref{sec:converse2g}), which deal with  the correlation between the channel inputs and the channel outputs, and the high dimension of the channel inputs.
Specifically, a lemma, corresponding to the MMSE estimator (see Lemma~\ref{lm:densityh}), is used in our proof.

\subsection{Extension to the multiuser broadcast channel}

Due to the difficulty of the converse (as discussed in the previous subsection), in this work we just focus on the MISO  channel with noiseless feedback. 
Even for this setting, the optimal beamforming gain is still unknown so far under the second moment input constraint --- there is still a gap between the derived beamforming gain upper bound and lower bound.  We conjecture that the derived lower bound is optimal. 

In the future work, we will extend our results to the multiuser broadcast channel.  Note that the proposed scheme and the converse can be  extended to a $K$-user MISO broadcast channel with some modifications. In fact, based on our previous approach, we can easily prove that the sum beamforming gain  of a $K$-user MISO broadcast channel with feedback is upper bounded by  $\min\{2\alpha K, K\}$, under the second moment input constrain. This is because $K$-user MISO broadcast channel can be enhanced to $K$ \emph{parallel} MISO channels, and the beamforming gain of each  MISO channel is upper bounded by   $\min\{2 \alpha, 1\}$ according to our result (see Theorem~\ref{thm:MISOb2}).
We also conjecture that $\min\{\alpha K , K\}$ is the optimal sum beamforming gain for the $K$-user MISO broadcast channel. In the future work we will prove this conjecture, which is also related to the conjecture of the MISO  channel.

\appendices

\section{Proofs of Proposition~\ref{thm:idealc}}   \label{sec:idealc}

In this section we provide the proof of Proposition~\ref{thm:idealc}, for the \emph{ideal} case of MISO  channel with  perfect CSIT and  CSIR, and with a \emph{second} moment input constraint.
According to the previous works in \cite{JP:03, GV:97, Gallager:68, KAC:90, AC:91}, for this ideal case,  the channel capacity is characterized  as 
\begin{align}
 \Ci  =   \max_{\bar{P} (\gamma): \ \int_{\gamma}\bar{P} (\gamma) f_{\gamma}(\gamma) d\gamma =P}  \int_{\gamma}   \log \bigl(1 +   \bar{P} (\gamma) \cdot \gamma \bigr) f_{\gamma} (\gamma) d \gamma \label{eq:idealc1}  
\end{align}
 where $\gamma \defeq  \|\hvu_t \|^2$, \  $f_{\gamma} (\gamma)$ is the probability  density function of $\gamma$, \ $\bar{P} (\gamma)$ is the power allocation  function and the optimal solution  of $\bar{P} (\gamma)$ is based on a water-filling algorithm.
We here focus on the asymptotic analysis when  the antenna-number $M$ is large. 

For the capacity $\Ci$ expressed in \eqref{eq:idealc1}, it can be  upper bounded as:
\begin{align}
 \Ci  &=   \max_{\bar{P} (\gamma):  \  \E_{\gamma} [\bar{P} (\gamma) ] =P}    \E_{\gamma}\bigl[    \log \bigl(1 +   \bar{P} (\gamma) \cdot \gamma \bigr) \bigr] \non  \\
 & \leq  \max_{\bar{P} (\gamma):  \  \E_{\gamma} [\bar{P} (\gamma) ] =P}          \E_{\gamma}\bigl[    \log \bigl(1 +   \bar{P} (\gamma)  \bigr) \bigr]   +   \E_{\gamma}\bigl[    \log \bigl(1 + \gamma \bigr) \bigr]   \label{eq:idealc1351}    \\
 & \leq  \max_{\bar{P} (\gamma):  \   \E_{\gamma} [\bar{P} (\gamma) ] =P}           \log \bigl(1 +   \underbrace{\E_{\gamma} [   \bar{P} (\gamma)  ] }_{=P} \bigr)     +  \log \bigl(1 +      \underbrace{\E_{\gamma}[ \gamma ] }_{=M} \bigr)    \label{eq:jensen256}    \\
    & =      \log \bigl(1 +   P \bigr)     +  \log \bigl(1 +    M \bigr)        \label{eq:chiexp9256}    \\
        & =      \log \bigl(1 +   PM + P+M \bigr)         \label{eq:ub8225}   
\end{align}
where  \eqref{eq:idealc1351} results from the identity that $\log (1+ a_1 a_2) \leq \log (1+ a_1 ) + \log (1+  a_2)$  for any  $a_1 \geq 0$ and $a_2 \geq 0$;
\eqref{eq:jensen256} stems from Jensen's inequality; 
\eqref{eq:chiexp9256}  follows from the fact that $ \E_{\gamma}[ \gamma ] =   \E [\|\hvu_t \|^2] = M$.  

Let us now focus on the lower bound on $\Ci $ expressed in \eqref{eq:idealc1}.  Since $\Ci$ is determined by the optimal power allocation of $\bar{P} (\gamma)$ over all possible  power allocation strategies. Clearly, setting $\bar{P} (\gamma) = P$, $\forall \gamma$ (equal power allocation) gives a lower bound on $\Ci$.  
Therefore,
\begin{align}
 \Ci  &=   \max_{\bar{P} (\gamma):  \  \E_{\gamma} [\bar{P} (\gamma) ] =P}    \E_{\gamma}\bigl[    \log \bigl(1 +   \bar{P} (\gamma) \cdot \gamma \bigr) \bigr] \non  \\
 & \geq     \E_{\gamma}\bigl[    \log \bigl(1 +    P \cdot \gamma \bigr) \bigr]   \label{eq:idealc2667}    \\
 & \geq    \Bigl( \E_{\gamma}\bigl[    \log \bigl( P \cdot \gamma \bigr) \bigr]  \Bigr)^{+}  \label{eq:idealc3567}    \\
 &=      \Bigl(  \E_{\gamma}\bigl[    \log \bigl(  2 \gamma \bigr) \bigr]   +       \log \bigl( \frac{P}{2}  \bigr)   \Bigr)^{+}      \non     \\
 &\geq   \Bigl(  \log \max\{ 2M - 2, 1 \}   +       \log \bigl( \frac{P}{2}  \bigr)    \Bigr)^{+}      \label{eq:ub3678}     \\ 
 & =    \Bigl(  \log \max\{ (M - 1) P , P/2 \}        \Bigr)^{+}     \non     \\ 
  & \geq      \log (1+ (M - 1) P )   -1    \label{eq:lb2955}    
 \end{align}
where \eqref{eq:idealc2667} uses a suboptimal power allocation, i.e., $\bar{P} (\gamma) = P$, $\forall \gamma$, which will not increase the value of $ \Ci$;
\eqref{eq:idealc3567} uses the notation of $(\bullet)^{+} =  \max\{\bullet, 0 \}$;
\eqref{eq:ub3678} stems from Lemma~\ref{lm:chi} (see below), that is, $\E_{\gamma}\bigl[    \log \bigl(  2 \gamma \bigr) \bigr]  \geq    \log \max\{ 2M-2, 1 \} $, given that  $2 \gamma =  2 \|\hvu_t \|^2 \sim \Xc^2(2M)$;
\eqref{eq:lb2955} follows from the identity that $ \bigl( \log x \bigr)^{+} \geq   \log(1+ x) -1 $ for a positive $x$.
Therefore, combining the upper bound and lower bound in \eqref{eq:ub8225} and  \eqref{eq:lb2955} leads to the following conclusion:
\begin{align}
      \log (1+ (M - 1) P )   -1    \leq  \Ci    \leq   \log \bigl(1 +   PM + P+M \bigr).    \non 
 \end{align}
 For a finite $P$, we have 
\begin{align}
    \lim_{M\to \infty} \frac{ \log (1+ (M - 1) P )   -1 }{\log (1+ P M)}  =1  \quad \text{and} \quad     \lim_{M\to \infty} \frac{ \log \bigl(1 +   PM + P+M \bigr) }{\log (1+ P M)}  =1   
 \end{align}
which imply that $ \lim_{M\to \infty} \frac{ \Ci   }{\log (1+ P M)}  =1 $, $\Ci = \log(1+ PM) + o(\log M)$,  and that  $ b(\alpha) = \lim_{M\to \infty}  \frac{ \Ci  }{ \log M} = 1$.   
At this point, we complete the proof.

\begin{lemma}     \label{lm:chi}
If  $\uv \sim \Xc^2(k) $ is  a chi-square random variable with $k \geq 2$ degrees of freedom,  $k$ is an even number, then 
\begin{align}
\E   [  \log  \uv ]  & \geq    \log \max\{ k-2, 1 \}  .  \non
\end{align}
 \end{lemma}
\vspace{2pt}

\begin{proof} 

If  $\uv$ is  a chi-square random variable with $k \geq 2$ degrees of freedom, its probability density function is given by   
\begin{align}
f_{\Xc}(\uv)& = 
\begin{cases}
  \frac{\uv^{k/2-1}  e^{-\uv/2 } }{  2^{k/2} \Gamma (k/2)  }  & \quad   \uv  > 0   \\
   0   & \quad \text{else}   \\
\end{cases}
\label{eq:chisquared11}  
\end{align}
where $\Gamma ( \bullet) $ is a Gamma function (cf.~\cite{Garcia:08}). 
When $k \geq 2$ and $k$ is an even number,   we have 
\begin{align}
\E [\ln \uv ] = \psi(k/2)  + \ln 2      \non
\end{align}
(see~4.352-1 in \cite{TableGR:96}), where  $\psi(x) $ is the digamma function. Note that $\psi(1) = - \gamma_o$, where $\gamma_o \approx 0.57721566$ is Euler's  constant, and for any integer $x>1$ the digamma function $\psi(x)$ can be expressed as
\begin{align}
\psi(x) = -\gamma_o + \sum_{p=1}^{x-1} \frac{1}{p}    \non 
\end{align}
(cf.~\cite{HandbookAS:64, SG:02}). 
Therefore, when $k > 2$ and $k$ is an even number, we have
\begin{align}
\E [\ln \uv ]  &=   \psi(k/2)  + \ln 2  \non \\
& =  -\gamma_o + \sum_{p=1}^{k/2-1} \frac{1}{p}   + \ln 2   \non\\ 
&\geq     \ln (k/2-1)   +  \ln 2   \label{eq:bound2566}   \\
&=      \ln (k-2)    \label{eq:bound89256}
\end{align}
where \eqref{eq:bound2566}  uses  the identity of Harmonic series  $ \sum_{p=1}^{m} \frac{1}{p}  \geq \ln m + \gamma_o$ for any positive natural number $m$ (cf.~\cite{CQ:03}).  
When $k=2$, then 
\begin{align}
\E [\ln \uv ]  &=   \psi(1)  + \ln 2  \non \\
& =  -\gamma_o    + \ln 2   \non\\ 
&\geq 0.   \label{eq:bound02996}
\end{align}
Finally, by combining \eqref{eq:bound89256} and \eqref{eq:bound02996}, we have  $\E   [  \log   \uv]   = \frac{1}{\ln 2}\E   [  \ln   \uv ]    \geq    \frac{1}{\ln 2} \ln (\max\{ k-2, 1 \}) = \log (\max\{ k-2, 1 \}) $.
\end{proof}

\section{Converse: the case with  a  fourth moment input constraint   \label{sec:converse} }

This section provides a capacity upper bound for the MISO  channel defined in Section~\ref{sec:system}, under a fourth moment input constraint (cf.~\eqref{eq:power4m}).    The result of capacity upper bound  is summarized in the following theorem. 
\vspace{5pt}
\begin{theorem} [Upper bound, fourth moment]  \label{thm:MISO}
For the MISO  channel with feedback defined in Section~\ref{sec:system}, the  capacity is upper bounded by 
\[C \leq   \log\bigl(   1 + \min\bigl\{M +2 ,  \      \sqrt{2}( T_c +1)   \bigr\}   \cdot  \kappa P \bigr)  \]
under the fourth moment input constraint in \eqref{eq:power4m}. 
\end{theorem}
\vspace{5pt}

\vspace{5pt}
\begin{remark} [Proof of Theorem~\ref{thm:MISOb4}, converse]  \label{rmk:proofuMISOb4}
From the capacity upper bound in Theorem~\ref{thm:MISO}, we can easily derive  an upper bound on the  beamforming gain: 
\begin{align}
b(\alpha)  &\leq   \lim_{  M\to \infty}   \frac{   \log\bigl(   1 + \min\bigl\{M +2 ,  \      \sqrt{2}( M^{\alpha} +1)   \bigr\}   \cdot  \kappa P \bigr) }{ \log M}     \non   \\
  &=  \min \{1  ,  \alpha\} 
  \end{align}
  under the fourth moment input constraint. 
It then proves the converse of Theorem~\ref{thm:MISOb4}.
\end{remark}
\vspace{5pt}

In what follows we provide the proof of Theorem~\ref{thm:MISO}. 
The proof for this case with fourth moment  input constraint  is  slightly different from that for the cases with second moment  input constraint (see Section~\ref{sec:converse2g}). In this case, the genie-aided channel enhancement, used in the previous case, is not used here. For this case, we will use a Cauchy-Schwarz inequality and Lemmas~\ref{lm:densityh}, \ref{lm:estimate1} and \ref{lm:ExpCht} (see Section~\ref{sec:converse2g}).

Beginning with Fano's inequality, we bound the rate of this setting as follows:
\begin{align}
nR &\leq \Imu(\Wme; \yv^{n})  +  n \epsilon_n  \non \\ 
& =  \sum_{t=1}^n  \bigl(  \hen( \yt \big| \yv^{t-1}) - \hen( \yt \big| \Wme, \yv^{t-1})   \bigr)  +  n \epsilon_n \non \\
& \leq  \sum_{t=1}^n  \bigl(  \hen( \yt ) - \hen( \yt \big| \Wme, \yv^{t-1},\hvut, \xvut)   \bigr)   +  n \epsilon_n  \label{eq:4cond2566} \\
& =  \sum_{t=1}^n     \hen ( \yt  )    - n\log (\pi e )   +  n \epsilon_n \label{eq:bound1}  
\end{align}
where 
\eqref{eq:4cond2566} uses  the fact that conditioning reduces differential entropy;
\eqref{eq:bound1} results from  the fact that $\hen( \yt \big| \Wme, \yv^{t-1}, \hvut, \xvut) =   \hen( \zt) = \log (\pi e )$.
We proceed to  upper bound the  differential entropy $ \hen \bigl( \yt  \bigr) $ in \eqref{eq:bound1}. 
Note that the average power of $\yt$ is 
\[  \E \bigl[ | \yt |^2  \bigr]   =   1 + \E \bigl[ |\hvut^\T \xvut |^2  \bigr] . \]
Again, by using the fact that differential entropy is maximized by a circularly symmetric complex Gaussian distribution with the same average power, we have 
\begin{align}
   \hen\bigl( \yt  \bigr)   \leq  \log\Bigl(\pi e \bigl(   1 + \E \bigl[ |\hvut^\T \xvut |^2  \bigr]  \bigr)\Bigr).  \label{eq:guassian8143}  
\end{align}
Then, by combining \eqref{eq:bound1}  and   \eqref{eq:guassian8143} it yields the following  bound on the rate: 
\begin{align}
 nR - n \epsilon_n    & \leq   \sum_{t=1}^n    \log\Bigl(\pi e \bigl(  1 + \E \bigl[ |\hvut^\T \xvut |^2  \bigr] \bigr)\Bigr)  - n\log (\pi e )    \nonumber \\
 & =   \sum_{t=1}^n     \log\Bigl(   1 + \E \bigl[ |\hvut^\T \xvut |^2  \bigr] \Bigr).   \label{eq:bound987} 
\end{align}
Let us now focus on the term $\E \bigl[ |\hvut^\T \xvut |^2  \bigr] $ in \eqref{eq:bound987}.
Similarly to the previous cases,   computing the value of $\E \bigl[ |\hvut^\T \xvut |^2  \bigr] $  could be challenging in general,  since $\xvut$ and $\hvut$ are correlated.  
 By following the similar steps in \eqref{eq:3exex}-\eqref{eq:3trmax366}, we bound the value of $\E \bigl[ |\hvut^\T \xvut |^2  \bigr]$ as 
 \begin{align}
 \E \Bigl[ |\hvut^\T \xvut |^2 \Bigr] \leq   \E \Bigl[    \|\hat{\hvut}\|^2 \cdot \| \xvut \|^2   +  \| \xvut \|^2      \Bigr] \label{eq:trmax366} 
\end{align}
where  $\hat{\hvut}$ is defined in \eqref{eq:mmse424}. 
Similarly to the steps in \eqref{eq:3exex}-\eqref{eq:3trmax366}, \eqref{eq:trmax366} uses the facts that $ \tilde{\hvut}    \  \big| \  (\yv^{t-1}, \Wme)  \sim  \Cc\Nc ( \zerou ,   \Omega_{t})$ (see~Lemma~\ref{lm:densityh} in Section~\ref{sec:converse2g}) and that $\lambda_{\max}( \Omega_{t}) \leq 1$  (see Lemma~\ref{lm:estimate1} in Section~\ref{sec:converse2g}), where   $\Omega_{t}$ and $\tilde{\hvut}$ are defined in \eqref{eq:mmse9832} and \eqref{eq:mmse924}.
At this point, by combining  \eqref{eq:trmax366} and \eqref{eq:bound987} we bound the rate as
\begin{align}
 nR - n \epsilon_n    &\leq   \sum_{t=1}^n    \log\Bigl(   1 +  \E \bigl[   \|\hat{\hvut}\|^2 \cdot \| \xvut \|^2 \bigr] +   \E \bigl[  \| \xvut \|^2\bigr]      \Bigr)  \nonumber\\
  & =    \sum_{t=1}^n    \log\Bigl(   1 +  \E \bigl[  ( \|\hat{\hvut}\|^2 +1) \cdot  \| \xvut \|^2 \bigr]      \Bigr).     \label{eq:bound0255}   
\end{align}
In order to bound $\E \bigl[  ( \|\hat{\hvut}\|^2 +1) \cdot  \| \xvut \|^2 \bigr] $ in  \eqref{eq:bound0255},  we use Cauchy-Schwarz inequality, that is, $\E[\boldsymbol{a}  \boldsymbol{b} ] \leq \sqrt{\E[ |\boldsymbol{a}|^2 ]}\cdot \sqrt{\E[ |\boldsymbol{b}|^2 ]}$ for any two random variables $\boldsymbol{a}$ and $\boldsymbol{b}$. With this inequality we have 
\[ \E \bigl[  ( \|\hat{\hvut}\|^2 +1)  \cdot \| \xvut \|^2  \bigr]  \leq  \sqrt{ \E \bigl[  ( \|\hat{\hvut}\|^2 +1)^2\bigr]}  \cdot \sqrt{ \E \bigl[   \|\xvut\|^4\bigr]}   \]
which, together with \eqref{eq:bound0255}, gives the following bound on the rate  
\begin{align}
 & nR - n \epsilon_n    \non \\
 &\leq   \sum_{t=1}^n    \log\Bigl(   1 +  \sqrt{ \E \bigl[  ( \|\hat{\hvut}\|^2 +1)^2\bigr]}  \cdot \sqrt{ \E \bigl[   \|\xvut\|^4\bigr]}      \Bigr)    \label{eq:bound443}   \\ 
 &\leq  \sum_{t=1}^n    \log\Bigl(   1 + \sqrt{ \min\bigl\{M^2 \!+\! 4M \!+\! 1 ,  \    2 [ (t-1) \  \text{mod} \  T_c]^2 + 7 [(t -1) \ \text{mod} \  T_c] +1    \bigr\}  } \cdot  \sqrt{ \E \bigl[   \|\xvut\|^4\bigr]}    \Bigr)  \label{eq:lmchannel}   \\ 
    & \leq  \sum_{t=1}^n    \log\Bigl(   1 + \min\bigl\{M +2 ,  \      \sqrt{2}( T_c +1)   \bigr\}   \cdot  \sqrt{ \E \bigl[   \|\xvut\|^4\bigr]}    \Bigr)  \label{eq:bound4266}   \\ 
    & \leq  \max_{   \sum_{\ell =1}^n \E[ \|\xvu_{\ell}\|^4 ] \leq n  \kappa^2  P^2  } \sum_{t=1}^n    \log\Bigl(   1 + \min\bigl\{M +2 ,  \      \sqrt{2}( T_c +1)   \bigr\}   \cdot  \sqrt{ \E \bigl[   \|\xvut\|^4\bigr]}    \Bigr)     \label{eq:bound7662}   \\ 
        & =   n \log\bigl(   1 + \min\bigl\{M +2 ,  \      \sqrt{2}( T_c +1)   \bigr\}   \cdot  \kappa P \bigr)  \label{eq:power2466}  
\end{align}
where \eqref{eq:bound443} results from \eqref{eq:bound0255} and Cauchy-Schwarz inequality;
  \eqref{eq:lmchannel} follows  from  \eqref{eq:chbound5} in  Lemma~\ref{lm:ExpCht} (see Section~\ref{sec:converse2g}); 
  \eqref{eq:bound4266}  stems from that  $M^2 + 4M +1 < (M +2)^2$ and  that $ 2 [ (t-1) \  \text{mod} \  T_c]^2 + 7 [(t -1) \ \text{mod} \  T_c] +1$  $ \leq 2  (T_c-1)^2 +  7 (T_c -1) +1 $  $ < 2(T_c +1)^2 $;
\eqref{eq:bound7662} results from   maximizing  the RHS of \eqref{eq:bound4266} under a fourth moment constraint  (cf.~\eqref{eq:power4m}); 
\eqref{eq:power2466} follows from Lemma~\ref{lm:maxsqrt} (see below). 
At this point,  as $n\to \infty$, we have the bound  $R  \leq     \log\bigl(   1 + \min\bigl\{M +2 ,  \      \sqrt{2}( T_c +1)   \bigr\}   \cdot  \kappa P \bigr) $ and complete the proof. 
The following lemma was used in  our proof. 
 \vspace{5pt}
\begin{lemma}   \label{lm:maxsqrt}
The solution for the following maximization problem 
\begin{align}
\text{maximize} \quad       &\sum_{t=1}^n    \log(1 +    c \sqrt{s_t} )  \non\\
 \text{subject to} \quad  &  \sum_{t=1}^n  s_t \leq  m  \non \\
       &   s_t \geq 0 ,  \quad t =1,2,\cdots,n  \non
\end{align}
is  $ s_1^{\star} =  s_2^{\star} = \cdots =  s_n^{\star} =  m/n$, for  constants $m>0$ and $c>0$.
\end{lemma} 
 \begin{proof}
This lemma  follows directly from Jensen's inequality.  By applying Jensen's inequality to the  concave function $ f(x) = \log(1 + c  \sqrt{x} )$, we have 
\[ \frac{1}{n} \sum_{t=1}^n    \log(1 +     c \sqrt{s_t} )   \leq     \log \Bigl(1 +   c  \sqrt{  \frac{1}{n} \sum_{t=1}^n  s_t  }\Bigr)     \]
 which, together with the constraint of $ \sum_{t=1}^n  s_t \leq  m $, gives the bound $ \sum_{t=1}^n    \log(1 +   c\sqrt{  s_t } )   \leq  n \log (1 +     c\sqrt{ \frac{m}{n} } ) $. The equality holds when $s^{\star}_1 =  s^{\star}_2 = \cdots =  s^{\star}_n = m/n$.
\end{proof}

\section{Proof  of Proposition~\ref{pro:rate}  \label{sec:prorate} }

In this section we provide the proof of Proposition~\ref{pro:rate}.
Note that our rate analysis is closely inspired by \cite{LS:02} and \cite{Caire+:10m}.
For the proposed scheme with Gaussian input, training and feedback described in Sections~\ref{sec:training} and \ref{sec:data}, the scheme achieves the following ergodic rate  \[ R= \frac{1}{T_c} \Imu(\svu_d; \    \yvu_{\tau}, \yvu_d)  \] by encoding the message over sufficiently large number of channel blocks, where the relationship between $\svu_d$,  $\yvu_{\tau}$ and $ \yvu_d$ are given in  \eqref{eq:ph1y} and \eqref{eq:datayv}.
The achievable rate can be lower bounded as: 
\begin{align}
T_c R  
 =& \Imu(\svu_d ; \    \yvu_{\tau}, \  \yvu_d )   \non \\
= &\Imu(\svu_d ; \  \hat{\hvu}_{\tau} ,   \  \yvu_{\tau}, \  \yvu_d)   \label{eq:dataine4367} \\
\geq &\Imu(\svu_d ; \  \hat{\hvu}_{\tau}, \yvu_d)     \label{eq:ratefdatai} \\
 = &  \Imu(\svu_d ; \  \hat{\hvu}_{\tau})     +   \Imu(\svu_d ; \  \yvu_d \  \big| \  \hat{\hvu}_{\tau})   \non \\
= &   \Imu(\svu_d ; \  \yvu_d \  \big| \ \hat{\hvu}_{\tau})   \label{eq:ratef9527}\\
= &   \Imu(\svu_d, \  \sqrt{P} \|\hat{\hvu}_{\tau}\| \svu_d   \ ; \  \yvu_d \  \big| \ \hat{\hvu}_{\tau})   \label{eq:ratef2962}\\
\geq  &   \Imu( \sqrt{P} \|\hat{\hvu}_{\tau}\| \svu_d ; \  \yvu_d \  \big| \  \hat{\hvu}_{\tau})    \label{eq:ratefdatai11} \\
 = & \hen( \sqrt{P} \|\hat{\hvu}_{\tau}\|\svu_d  \ \big| \hat{\hvu}_{\tau}  )    -  \hen ( \sqrt{P} \|\hat{\hvu}_{\tau}\|\svu_d   \   \big| \yvu_d , \hat{\hvu}_{\tau})     \non \\
  = &  \Td \cdot  \E [ \log (\pi e P  \|\hat{\hvu}_{\tau}\|^2) ]    -   \hen ( \sqrt{P} \|\hat{\hvu}_{\tau}\|\svu_d   \   \big| \yvu_d , \hat{\hvu}_{\tau})     \label{eq:diffb8424} 
  \end{align}
where \eqref{eq:dataine4367} results from the fact that $\|\hat{\hvu}_{\tau}\|$ is a deterministic function of $\yvu_{\tau}$;
 \eqref{eq:ratefdatai} and \eqref{eq:ratefdatai11} are from the fact that adding more information will not reduce the mutual information;
 \eqref{eq:ratef9527} is from our input assumption that $\svu_d $ and $\|\hat{\hvu}_{\tau}\|$ are independent;
 \eqref{eq:ratef2962} uses the fact that $\sqrt{P} \|\hat{\hvu}_{\tau}\| \svu_d $ is a deterministic function of $\svu_d$ and $ \hat{\hvu}_{\tau}$;
 \eqref{eq:diffb8424}  follows from the fact that  $\svu_d  \sim \Cc\Nc (\zerou, I_{\Td})$, where  $\Td =T_c  - \Tt$ (cf.~\eqref{eq:TtTd}).
 Let us focus on the second term in \eqref{eq:diffb8424}, which can be upper bounded as:
 \begin{align}
\hen ( \sqrt{P} \|\hat{\hvu}_{\tau}\|\svu_d   \   \big| \yvu_d , \hat{\hvu}_{\tau})     
    \leq  &  \sum_{t = \Tt +1 }^{T_c}  \hen ( \sqrt{P} \|\hat{\hvu}_{\tau}\| \sv_t  \   \big| \yv_t ,   \hat{\hvu}_{\tau} )     \label{eq:diffb35787} \\
= &    \sum_{t = \Tt +1 }^{T_c}   \hen ( \sqrt{P} \|\hat{\hvu}_{\tau}\| \sv_t -  \beta_t \yv_t   \   \big| \yv_t ,   \hat{\hvu}_{\tau} )   \label{eq:cond4256}\\
\leq  &    \sum_{t = \Tt +1 }^{T_c}   \hen ( \sqrt{P} \|\hat{\hvu}_{\tau}\| \sv_t -  \beta_t \yv_t   \   \big|  \hat{\hvu}_{\tau} )   \label{eq:cond7246}\\
 \leq &    \sum_{t = \Tt +1 }^{T_c}   \E \Bigl[ \log \Bigl(\pi e \cdot  \E \Bigl[ \  \bigl|   \sqrt{P} \|\hat{\hvu}_{\tau}\| \sv_t     - \beta_t \yv_t   \bigr|^2  \ \Big|   \ \hat{\hvu}_{\tau} \Bigr ] \Bigr)   \Bigr]  \label{eq:Gau6245}
\end{align}
where \eqref{eq:diffb35787} is from chain rule and the fact that conditioning reduces differential entropy, where $\yvu_{d} \defeq [ \yt[ \Tt+1], \yt[ \Tt+2], \cdots, \yt[ T_c] ]^\T$, $\svu_d \defeq [  \sv_{\Tt+1} , \sv_{\Tt+2} \cdots,\sv_{T_c} ]^\T$ and 
 \begin{align}
 \yv_t  = \sqrt{P} \|\hat{\hvu}_{\tau}\| \sv_t  +  \sqrt{P} \tilde{\hvu}_{\tau}^\T  \frac{\hat{\hvu}_{\tau}^{*}}{\|\hat{\hvu}_{\tau}\|} \sv_t   + \zv_t   , \quad  t =\Tt+1, \Tt+2, \cdots, T_c    \label{eq:datay22}
\end{align}
(cf.~\eqref{eq:datay});
 \eqref{eq:cond4256} results from that  $  \hen ( \sqrt{P} \|\hat{\hvu}_{\tau}\| \sv_t  \   \big| \yv_t ,   \hat{\hvu}_{\tau} )  =   \hen ( \sqrt{P} \|\hat{\hvu}_{\tau}\| \sv_t -  \beta_t \yv_t   \   \big| \yv_t ,   \hat{\hvu}_{\tau} )    $ for any deterministic function $\beta_t$ of $\yv_t$ and $\hat{\hvu}_{\tau}$;
 \eqref{eq:cond7246} is from the fact that conditioning reduces differential entropy;
 \eqref{eq:Gau6245}  uses the fact that Gaussian distribution is the differential entropy maximizer given the same second moment of $ \E \Bigl[ \  \bigl|   \sqrt{P} \|\hat{\hvu}_{\tau}\| \sv_t     - \beta_t \yv_t   \bigr|^2  \ \Big|   \ \hat{\hvu}_{\tau} \Bigr ]$.  

  In the next step we will focus on  a single term inside the summation in \eqref{eq:Gau6245}. 
Specifically, we will choose a proper $\beta_t$ to  minimize $ \E \bigl[ |\sqrt{P} \|\hat{\hvu}_{\tau}\|\sv_t   - \beta_t \yv_t |^2 \big| \hat{\hvu}_{\tau}  \bigr ]$, which will in turn  tighten the bound in \eqref{eq:Gau6245}, where $\yv_t$ is expressed in \eqref{eq:datay22}.  
This is equivalent to the MMSE estimation problem.
For the MMSE estimation problem, the optimal $c$ to minimize $ \E \bigl[ | \uv   - c \vv|^2 ]$ is $c_{\star} = \frac{ \E[\uv \vv^{*}] }{\E[ | \vv|^2 ]}$ and in this case $ \E \bigl[ | \uv   -c_{\star} \vv |^2  ] = \E[|\uv|^2 ]  - \frac{ | \E[\uv \vv^{*}]  |^2}{\E[ | \vv|^2 ]} $, for two random variables  $\uv$ and $\vv$ with \emph{zero} means. Therefore,  the optimal $\beta_t$ can be chosen as 
\begin{align}
\beta_t = \frac{ \E \bigl[\sqrt{P} \|\hat{\hvu}_{\tau}\| \sv_t  \yv^{*}_t \  \big| \  \hat{\hvu}_{\tau} \bigr] }{ \E \bigl[ \ |\yv_t|^2  \  \big|   \  \hat{\hvu}_{\tau}  \bigr] } =     \frac{P \|\hat{\hvu}_{\tau}\|^2 }{P \|\hat{\hvu}_{\tau}\|^2 + P \sigma^2 +1 }        \label{eq:betadesign}
\end{align}
where  \[ \sigma^2 \defeq \frac{1}{P+1}\] corresponding to the variance of $\tilde{\hvu}_{\tau}^\T  \frac{\hat{\hvu}_{\tau}^{*}}{\|\hat{\hvu}_{\tau}\|}$  given $\hat{\hvu}_{\tau}$.  Remind that  $\hat{\hvu}_{\tau}$ and $\tilde{\hvu}_{\tau}$ are \emph{independent} with each other,    $\hat{\hvu}_{\tau} \sim \Cc\Nc(\zerou,  \frac{P}{P+1} I)$ and $\tilde{\hvu}_{\tau} \sim \Cc\Nc(\zerou, \frac{1}{P+1} I )$.
By setting  $\beta_t$ as in \eqref{eq:betadesign},  we have 
\begin{align}
  \E \bigl[ |\sqrt{P} \|\hat{\hvu}_{\tau}\|\sv_t   - \beta_t \yv_t |^2 \big|  \hat{\hvu}_{\tau}   \bigr ]     
  & =  \E \bigl[\  |\sqrt{P} \|  \hat{\hvu}_{\tau} \| \sv_t  |^2 \ \big|  \hat{\hvu}_{\tau} \bigr ]       -  \frac{  \bigl|\ \E \bigl[  \sqrt{P} \|\hat{\hvu}_{\tau}\| \sv_t   \yv^{*}_t \ \big|  \ \hat{\hvu}_{\tau} \bigr] \  \bigr|^2    }{\E [|\yv_t|^2 \ \big|  \  \hat{\hvu}_{\tau}  ]}    \non\\
    &=   P \|\hat{\hvu}_{\tau}\|^2   -    \frac{ \bigl( P \|\hat{\hvu}_{\tau}\|^2 \bigr)^2    }{ P \|\hat{\hvu}_{\tau}\|^2 + P \sigma^2 +1 }      \non\\
        &= \frac{  P \|\hat{\hvu}_{\tau}\|^2 \cdot  ( P \sigma^2 +1 )   }{ P \|\hat{\hvu}_{\tau}\|^2   + P \sigma^2 +1 } .     \label{eq:betadesign21355} 
 \end{align}
By plugging \eqref{eq:Gau6245}  and \eqref{eq:betadesign21355} into  \eqref{eq:diffb8424}, we have: 
\begin{align}
T_c R& \geq \Td  \cdot \E \bigl[  \log \bigl(\pi e P  \|\hat{\hvu}_{\tau}\|^2\bigr)  \bigr]   -   \Td  \cdot \E \Bigl[   \log \Bigl(\pi e \cdot    \frac{  P \|\hat{\hvu}_{\tau}\|^2 \cdot  ( P \sigma^2 +1 )   }{ P \|\hat{\hvu}_{\tau}\|^2   + P \sigma^2 +1 }  \Bigr)   \Bigr]    \non  \\   
&=  \Td \cdot \E  \Bigl[  \log \bigl(1 +  \frac{ P  \|\hat{\hvu}_{\tau}\|^2}{P \sigma^2 +1   } \bigr) \Bigr] .   \label{eq:bound98244}  
\end{align}
Note that   $\hat{\hvu}_{\tau} \sim \Cc\Nc(\zerou,  \frac{P}{P+1} I )$ and $\delta \hat{\hvu}_{\tau} \sim \Cc\Nc(\zerou,  2 I )$, for \[ \delta \defeq \sqrt{\frac{2(P+1)}{P}}.\]  It then implies that  $\| \delta \hat{\hvu}_{\tau}\|^2 $ is  chi-squared distributed with $2\Tt$ degrees of freedom,  that is, $\| \delta\hat{\hvu}_{\tau}\|^2 \sim \Xc^2(2\Tt)$.
If  $\uv$ is  a chi-square random variable with $k \geq 2$ degrees of freedom, its probability density function is given by  \eqref{eq:chisquared11} and its probability density function is zero  when $\uv \leq  0$.
Therefore, without loss of generality we consider $\| \delta \hat{\hvu}_{\tau}\|^2 $ as  a \emph{positive}  chi-squared random variable  with $2\Tt$ degrees of freedom. Then, from \eqref{eq:bound98244} we further have
\begin{align}
T_c R&\geq  \Td \cdot \E  \Bigl[  \log \bigl(1 +  \frac{ \frac{P}{\delta^2}  \| \delta \hat{\hvu}_{\tau}\|^2}{P \sigma^2 +1   } \bigr) \Bigr]  \non\\
& =  \Td \cdot  \E  \bigl[   \log  ( \| \delta \hat{\hvu}_{\tau}\|^2 )  \bigr]    + \Td \cdot   \E  \bigl[    \log  \bigl( \frac{1}{\|\delta \hat{\hvu}_{\tau}\|^2 } + \frac{ P/\delta^2  }{  P \sigma^2 +1    } \bigr) \bigr]     \non\\ 
&\geq  \Td \cdot  \E  \bigl[   \log  ( \|\delta \hat{\hvu}_{\tau}\|^2 )  \bigr]   + \Td \cdot  \log \Bigl( \frac{1}{\E  \|\delta \hat{\hvu}_{\tau}\|^2}  + \frac{  P/\delta^2  }{  P \sigma^2 +1    }  \Bigr)    \label{eq:bound9385}  \\ 
& =  \Td \cdot  \E  \bigl[   \log  ( \| \delta \hat{\hvu}_{\tau}\|^2 )  \bigr]  + \Td \cdot  \log \Bigl( \frac{1}{  2 \Tt   }  + \frac{  P/\delta^2 }{  P \sigma^2 +1    }  \Bigr)    \label{eq:bound2386}  
\end{align}
where \eqref{eq:bound9385} follows from the fact that $g(x) = \log (\frac{1}{x} + c )$ is a convex function since $\frac{\partial^2 g(x)}{\partial x^2} \geq 0$ for any $x>0$,  where $c> $ is a constant;
\eqref{eq:bound2386} results from that $\E  \| \delta  \hat{\hvu}_{\tau}\|^2 =  2 \Tt $,   since $\| \delta\hat{\hvu}_{\tau}\|^2 \sim \Xc^2(2\Tt)$.
Let us now  focus on the first term in  \eqref{eq:bound2386}. 
From Lemma~\ref{lm:chi} described in Appendix~\ref{sec:idealc}, we note that if  $\uv \sim \Xc^2(k) $ is  a chi-square random variable with $k \geq 2$ degrees of freedom,  $k$ is an even number, then 
\begin{align}
\E   [  \log  \uv ]  & \geq    \log \max\{ k-2, 1 \}  \non
\end{align}
which, together with the fact that $\| \delta\hat{\hvu}_{\tau}\|^2 \sim \Xc^2(2\Tt)$, implies that 
\begin{align}
\E   [  \log  ( \| \delta \hat{\hvu}_{\tau}\|^2 ) ]   & \geq    \log (\max\{ 2(\Tt -1), 1 \}).  \label{eq:bound3677}   
\end{align}

Finally, by plugging \eqref{eq:bound3677} into  \eqref{eq:bound2386} we have: 
\begin{align}
T_c R&\geq   \Td \cdot  \log \bigl(\max\{2(\Tt -1) , \  1 \}\bigr)     + \Td \cdot  \log \Bigl( \frac{1}{  2 \Tt   }  + \frac{  P/\delta^2 }{  P \sigma^2 +1    }  \Bigr)          \non   \\ 
& =    \Td \cdot  \log \Bigl( \frac{ \max\{2(\Tt -1) , \  1 \} }{  2 \Tt   }  + \frac{ \frac{ P}{\delta^2 } \cdot \max\{2(\Tt -1) , \  1 \} }{  P \sigma^2 +1    }  \Bigr)              \non \\ 
& =    \Td \cdot  \log \Bigl(   1 -    \frac{1 }{   \max\{\Tt , 2 \}  }  + \frac{ P  \cdot\max\{2(\Tt -1) , \  1 \} }{  P \sigma^2\delta^2  +\delta^2    }  \Bigr)              \non \\ 
& =    ( T_c - \Tt)  \cdot  \log \Bigl(   1 -    \frac{1 }{   \max\{ \Tt , 2  \}  }  + \frac{ P  \cdot \max\{ (\Tt -1) , 1/2 \} }{  2  + \frac{1}{P}    }  \Bigr)               \label{eq:ratelbfinal6288}  
\end{align}
where $\delta^2 \defeq \frac{2(P+1)}{P}$,  $\sigma^2 \defeq \frac{1}{P+1}$ and  $\Td \defeq  ( T_c - \Tt) $. By dividing the two sides of \eqref{eq:ratelbfinal6288} with $T_c$, it gives the final lower bound on the achievable rate of the proposed scheme. At this point we complete the proof.

\section{Proof of Lemma~\ref{lm:estimate1}  \label{sec:estimate1}}

In this section we will prove  Lemma~\ref{lm:estimate1}.
Specifically, we will prove that, for any  vector $\evu_i \in \Cc^{M \times 1 }$ for $i \in \Zc $, and for
\begin{align}
  \Km_{t}  &  \defeq I_M -  \sum_{ i =  1 }^{t-1}   \frac{\Km_{i} \evu_i^{*}  \evu^\T_i  \Km_{i}}{\evu^\T_i \Km_{i} \evu_i^{*}  +1 },    \quad    t =2, 3, 4, \cdots      \label{eq:1estm3255}
  \end{align}
and $\Km_{1}    \defeq I_M $, then 
\[ \mathbf{0} \preceq \Km_{t}  \preceq  I_M,  \quad \forall t \in \{2,3, \cdots \}.\]

From the definition in \eqref{eq:1estm3255}, we have 
 \begin{align}
\Km_{t+1}  =   \Km_{t}  -  \frac{\Km_{t} \evu_t^{*}  \evu^\T_t  \Km_{t}}{\evu^\T_t \Km_{t} \evu_t^{*}  +1} ,  \quad   t \in \{1,2,3, \cdots \}. \label{eq:inequality1122}
  \end{align}
  One can easily check from \eqref{eq:inequality1122} that,   if $\Km_{t}$ is a Hermitian  matrix, then $\Km_{t+1}$ is also a Hermitian   matrix for  $t \in \{1,2,3, \cdots \}$.  Since  $\Km_{1}    \defeq I_M $ is a Hermitian  matrix, then from the above recursive argument  it is  true that  $\Km_{t}$ is a Hermitian   matrix for  $t \in \{1,2,3, \cdots \}$.

In the second step, we will prove that if the Hermitian matrix $\Km_{t}$ is  positive semidefinite, then the Hermitian matrix $\Km_{t+1}$ is also  positive semidefinite  for  $t \in \{1,2,3, \cdots \}$.  Specifically, if the Hermitian matrix $\Km_{t}$ is  positive semidefinite, $t \in \{1,2,3, \cdots \}$, then  for  any  vector $\xvu \in \Cc^{M \times 1 }$ we have 
  \begin{align}
  \xvu^\H \Km_{t+1} \xvu  
   &=  \xvu^\H\bigl( \Km_{t} -   \frac{\Km_{t} \evu_t^{*}  \evu^\T_t  \Km_{t}}{ \evu^\T_t \Km_{t} \evu_t^{*}  +1}   \bigr)\xvu   \label{eq:inequality2567}  \\
 & =  \xvu^\H\Km_{t} \xvu  -  \frac{ |\xvu^\H\Km_{t} \evu_t^{*} |^2}{ \evu^\T_t \Km_{t} \evu_t^{*}  +1}     \non\\
 &=  \bvu^\H \bvu -  \frac{ |\bvu^\H \cvu |^2}{ \cvu^\H\cvu +1}      \non\\
 &=   \frac{ \|\bvu\|^2   +  \|\bvu\|^2\|\cvu\|^2  - |\bvu^\H \cvu |^2  }{ \|\cvu\|^2 +1}      \non\\
 &\geq    \frac{\|\bvu\|^2  +  \|\bvu\|^2\|\cvu\|^2 -  \|\bvu\|^2\|\cvu\|^2   }{ \|\cvu\|^2 +1}    \label{eq:inequality134} \\
 &\geq  0   
  \end{align}
where \eqref{eq:inequality2567} is from the definition in \eqref{eq:inequality1122};
  the Hermitian positive semidefinite matrix $\Km_{t}$ is  decomposed  as $\Km_{t} \defeq  \Um \Lambda \Um^\H =  \Um \Lambda^{1/2} \Um^\H\Um \Lambda^{1/2} \Um^\H$ using singular value decomposition method, where $\Um$ and $\Lambda$ are the unitary matrix and diagonal matrix respectively,   $\bvu \defeq   \Um \Lambda^{1/2} \Um^\H  \xvu$ and $\cvu \defeq   \Um \Lambda^{1/2} \Um^\H \evu^{*}$;  \eqref{eq:inequality134} results from Cauchy-Schwarz inequality, i.e., $|\bvu^\H \cvu |^2  \leq  \|\bvu\|^2\|\cvu\|^2$.   
Since the Hermitian matrix $\Km_{1}$ is  positive semidefinite, then  from the above recursive argument  it is  true that   the  Hermitian   matrix $\Km_{t}$ is positive semidefinite, $t \in \{1, 2,3, \cdots \}$.

 From the above steps we have proved that  the  matrix $\Km_{t}$ is Hermitian  positive semidefinite, $t \in \{1, 2,3, \cdots \}$, which means that 
  \[\mathbf{0} \preceq \Km_{t} ,  \quad \forall t \in \{1, 2,3, \cdots \}. \]
 In the next step we will prove that 
   \[ \Km_{t}  \preceq  I_M  ,  \quad \forall t \in \{1, 2,3, \cdots \}. \]
 From the definition in  \eqref{eq:1estm3255}, we have 
 \begin{align}
 I_M - \Km_{t}  \defeq   \sum_{ i =  1 }^{t-1}   \frac{\Km_{i} \evu_i^{*}  \evu^\T_i  \Km_{i}}{\evu^\T_i \Km_{i} \evu_i^{*}  +1 } ,  \quad    t =2, 3, 4, \cdots .      \label{eq:inequality2466} 
  \end{align}
 Since matrix $\Km_{t}$ is Hermitian  positive semidefinite, $t \in \{1, 2,3, \cdots \}$,  it holds true that 
  \begin{align}
 \frac{\Km_{t} \evu_t^{*}  \evu^\T_t  \Km_{t}}{\evu^\T_t \Km_{t} \evu_t^{*}  +1 }  \succeq  \mathbf{0} , \quad    t =1, 2, 3, \cdots   \label{eq:inequality4788} 
  \end{align}
because  for  any  vector $\xvu \in \Cc^{M \times 1 }$ we have 
$   \xvu^\H\bigl(   \frac{\Km_{t} \evu_t^{*}  \evu^\T_t  \Km_{t}}{ \evu^\T_t \Km_{t} \evu_t^{*}  +1}   \bigr)\xvu   =  \frac{ |\xvu^\H\Km_{t} \evu_t^{*} |^2}{ \evu^\T_t \Km_{t} \evu_t^{*}  +1}  \geq 0 $.
Then combining \eqref{eq:inequality4788}  and  \eqref{eq:inequality2466} it gives  
 \begin{align}
 I_M - \Km_{t}  \defeq   \sum_{ i =  1 }^{t-1}   \frac{\Km_{i} \evu_i^{*}  \evu^\T_i  \Km_{i}}{\evu^\T_i \Km_{i} \evu_i^{*}  +1 } \succeq  \mathbf{0} ,  \quad    t =2, 3, 4, \cdots     \non 
  \end{align}
which implies that   \[ I_M   \succeq    \Km_{t} , \quad    t =2, 3, 4, \cdots .\] 
At this point, we complete the proof.

\section{Proof of Lemma~\ref{lm:densityh}  \label{sec:lmdensityh}}

This section provides the proof of Lemma~\ref{lm:densityh} (see Section~\ref{sec:converse2g}). 
For the ease of description, we rewrite $\hat{\uvu}_{t}$ and $\Omega_{t}$ (see \eqref{eq:mmseuv11} and \eqref{eq:mmseuv22}) with the following forms:
 \begin{align}
\hat{\uvu}_{t} \defeq \hat{\uvu}_{t-1}  +   \hat{\uvu}_{t-1,t}, \quad \quad \Omega_t \defeq\Omega_{t-1} - \Omega_{t-1,t} \label{eq:mmseuv12} 
\end{align}
where
\begin{align*}
  \hat{\uvu}_{t-1,t} &\defeq \Omega_{t-1} \Am^\H_{t-1}( \Am_{t-1}\Omega_{t-1}\Am^\H_{t-1}  + I_N )^{-1} (\yvut[t-1] - \Am_{t-1} \hat{\uvu}_{t-1} )  \\
\Omega_{t-1,t}  &\defeq \Omega_{t-1} \Am^\H_{t-1} ( \Am_{t-1}\Omega_{t-1}\Am^\H_{t-1}  + I_N )^{-1}  \Am_{t-1}\Omega_{t-1}.
\end{align*}
 for $t = 2,3, \cdots, T$.  
Lemma~\ref{lm:densityh} is the extension of the well-known result of MMSE estimator (see, for example, \cite[Chapter 15.8]{Kay:93}) that is expressed in the following Lemma~\ref{lm:mmse}.

\vspace{5pt}
\begin{lemma} \cite[Chapter 15.8]{Kay:93}  \label{lm:mmse}
Consider two independent random vectors $\uvu\in \Cc^{M\times 1} \sim \Cc\Nc (\hat{\uvu}_1,   \Omega_{1})$ and $\zvu \in \Cc^{N\times 1}\sim \Cc\Nc ( \zerou ,  I_N)$,  for some fixed $\hat{\uvu}_1$ and Hermitian positive semidefinite $ \Omega_{1}$.   Let \[\yvu=A\uvu+ \zvu\] where   $A\in \Cc^{N\times M}$  is  a fixed matrix.   Then, the conditional density of $\uvu$ given $\yvu$ is  \[ \uvu  | \yvu  \sim \Cc\Nc ( \hat{\uvu},   \Omega)\] where 
\begin{align*}
\hat{\uvu} &=\hat{\uvu}_1  + \Omega_{1} A^\H ( A \Omega_{1}A^\H  + I_N )^{-1} (\yvu - A \hat{\uvu}_1) \\
\Omega &=\Omega_{1} - \Omega_{1}A^\H ( A\Omega_{1}A^\H  + I_N )^{-1}  A\Omega_{1} .
\end{align*}
Furthermore,  the two random vectors $\hat{\uvu}$ and $\vvu \defeq \uvu - \hat{\uvu}$ are independent, and we have  \[  \vvu  | \yvu  \sim \Cc\Nc ( \zerou,   \Omega).\] 
\end{lemma}
\vspace{5pt}

Note that in Lemma~\ref{lm:mmse},  $\hat{\uvu}$ and $\vvu$ are two jointly proper complex Gaussian vectors and the covariance matrix of those two vectors vanishes, which implies that   $\hat{\uvu}$ and $\vvu$ are independent. The lack of correlation implies independence for two jointly proper Gaussian vectors (see, e.g., \cite{NM:93}).  
The proof of Lemma~\ref{lm:densityh} are described as follows. 

\subsection{Proof for the case with $t=2$}

We first consider the simple case with $t=2$. 
From Lemma~\ref{lm:mmse} we  conclude that the conditional density of $\uvu$ given $(\yvut[1], w)$ is
\begin{align}
 \uvu |(\yvut[1], w)  \sim  \Cc\Nc ( \hat{\uvu}_{2},   \Omega_{2})  \label{eq:mmseb9275}
\end{align}
where 
\begin{align}
\hat{\uvu}_{2} &=  \hat{\uvu}_{1}  + \Omega_{1}\Am_1^\H ( \Am_1\Omega_{1}\Am_1^\H  + I_N )^{-1} (\yvut [1] - \Am_1 \hat{\uvu}_{1} )          \label{eq:mmseb3924}
\end{align}
and
\begin{align}
\Omega_{2} &=  \Omega_{1} - \Omega_{1} \Am_1^\H ( \Am_1\Omega_{1}\Am_1^\H  + I_N )^{-1}  \Am_1\Omega_{1}        \label{eq:mmseb0221}
\end{align}
where $\Am_1$ is a deterministic function of $w$ by definition.
It follows from Lemma~\ref{lm:mmse} that  $\hat{\uvu}_{2}$ and $\vvu_2 \defeq \uvu  - \hat{\uvu}_{2}$ are independent;  the conditional density of $\vvu_2$ given $(\yvut [1], w)$ is \[ \vvu_2 |(\yvut [1], w) \sim \Cc\Nc ( \zerou,   \Omega_{2}).\] 

\subsection{Proof for the case with $t=3$}

We then consider the case with $t=3$ ($T \geq 3$).  
By using the result in \eqref{eq:mmseb9275}, that is,  $\uvu | ( \yvut[1], w) \sim  \Cc\Nc ( \hat{\uvu}_{2},   \Omega_{2})$, it yields the following conclusion:
\begin{align}
 \Bmatrix{ \zvu_2   \\  \uvu  } \Big| (\yvu_1,  w ) \  \sim \  \Cc\Nc \Bigl( \Bmatrix{  \zerou  \\  \hat{\uvu}_{2}  },    \Bmatrix{   I_N   &   \mathbf{0}_{N \times M}   \\   \mathbf{0}_{M\times N}  &   \Omega_{2}   } \Bigr) .     \label{eq:cg29563}
 \end{align} 
Let us now look at the following vector  
\begin{align}
 \Bmatrix{\yvu_2 \\ \uvu }  = \Bmatrix{ I_N     &  \Am_2   \\    \mathbf{0}_{M\times N}    &      I_M  } \Bmatrix{  \zvu_2 \\  \uvu }     \label{eq:jointG8823}
 \end{align}
 where $\Am_2$ is a deterministic function of $(\yvu_1, w)$.  It is well known that the affine transformation of a complex proper Gaussian  vector  also yields a complex proper Gaussian  vector, that is, if $\evu \in \Cc^{q\times 1}  \sim  \Cc\Nc( \bunderline{\mu}, Q)$, then it holds true that $B\evu  \sim  \Cc\Nc(B \bunderline{\mu}, BQB^{\H})$ for  fixed $\bunderline{\mu} \in  \Cc^{q\times 1} $,  $B\in \Cc^{p\times q}$ and $Q\in \Cc^{q\times q}$ (see, e.g., \cite{Tel:99,NM:93}).  
 Therefore, by combining \eqref{eq:cg29563} and \eqref{eq:jointG8823} it gives 
\begin{align}
 \Bmatrix{ \yvu_2 \\ \uvu } \Big| (\yvu_1,  w )  \sim  \Cc\Nc \Bigl( \Bmatrix{  \Am_2\hat{\uvu}_{2} \\ \hat{\uvu}_{2}  },     \underbrace{ \Bmatrix{\Km_{1,1}  &  \Km_{1,2}  \\  \Km_{2,1}  &     \Km_{2,2}   }}_{\defeq \Km}\Bigr)     \label{eq:condproy1y7234}
 \end{align} 
where  
\begin{align}
 \Km_{2,2}  =  \Omega_{2},   \quad 
    \Km_{2,1} =  \Omega_{2} \Am_2^\H ,   \quad
        \Km_{1,2}  =  \Am_2\Omega_{2}   ,  \quad 
        \Km_{1,1}   = \Am_2\Omega_{2} \Am_2^\H +  I_N.   \label{eq:K22}
\end{align} 
Let us consider a new vector obtained from the following affine transformation:
\begin{align*}
  \Bmatrix{ \yvu_2 \\  \uvu  \!- \! \Km_{2,1}  \Km_{1,1}^{-1}  \yvu_2 }  = B \Bmatrix{ \yvu_2 \\ \uvu }  \quad \text{where}  \quad B=\Bmatrix{ I_N       &   \mathbf{0}_{N\times M}    \\     -\Km_{2,1}  \Km_{1,1}^{-1}  & I_M  }.
 \end{align*} 
As mentioned, affine transformation of a complex proper Gaussian  vector  also yields a complex proper Gaussian  vector. Therefore, we have
 \begin{align}
   \Bmatrix{ \yvu_2 \\  \uvu  \!- \! \Km_{2,1}  \Km_{1,1}^{-1}  \yvu_2 }   \Big| (\yvu_1,  w )  \sim  \Cc\Nc \Bigl( \Bmatrix{  \Am_2\hat{\uvu}_{2} \\ \hat{\uvu}_{2}  - \Km_{2,1}  \Km_{1,1}^{-1}  \Am_2\hat{\uvu}_{2}  },   \Bmatrix{\Km_{1,1}  &  \mathbf{0}_{N\times M}   \\   \mathbf{0}_{M\times N}   &     \Km_{2,2}    \!- \! \Km_{2,1}  \Km_{1,1}^{-1}\Km_{1,2}  }\Bigr).     \label{eq:congau3884}
 \end{align}  
 The  result in \eqref{eq:congau3884} implies that
  \begin{align}
    \uvu  \!- \! \Km_{2,1}  \Km_{1,1}^{-1}  \yvu_2 \  \big| ( \yvu_1,  w )  \sim  \Cc\Nc \Bigl(  \hat{\uvu}_{2}  - \Km_{2,1}  \Km_{1,1}^{-1}  \Am_2\hat{\uvu}_{2}  ,    \   \Km_{2,2}    - \Km_{2,1}  \Km_{1,1}^{-1}\Km_{1,2}  \Bigr).     \label{eq:congau5799}
 \end{align}   
 The  result in \eqref{eq:congau3884} also implies that  the two vectors $ \yvu_2$ and $ \uvu  -  \Km_{2,1}  \Km_{1,1}^{-1}  \yvu_2$ are conditionally independent given $ (\yvu_1,  w )$  because their conditional cross-covariance   vanishes.  Based on this independence and \eqref{eq:congau5799},  it gives 
 \begin{align}
    \uvu  -  \Km_{2,1}  \Km_{1,1}^{-1}  \yvu_2 \  \big| (\yvu_2, \yvu_1,  w )  \sim  \Cc\Nc \Bigl(  \hat{\uvu}_{2}  - \Km_{2,1}  \Km_{1,1}^{-1}  \Am_2\hat{\uvu}_{2}  ,   \    \Km_{2,2}    - \Km_{2,1}  \Km_{1,1}^{-1}\Km_{1,2}  \Bigr)    \label{eq:congau3578}
 \end{align}  
and 
 \begin{align}
    \underbrace{  \uvu \! -\! \Km_{2,1}  \Km_{1,1}^{-1}  \yvu_2    \!+\!   \Km_{2,1}  \Km_{1,1}^{-1}  \yvu_2}_{\uvu}  \  \big| (\yvu_2, \yvu_1,  w )  \sim  \Cc\Nc \bigl(  \underbrace{ \hat{\uvu}_{2} \! +\!  \Km_{2,1}  \Km_{1,1}^{-1} (\yvu_2 \! -\! \Am_2\hat{\uvu}_{2}) }_{\hat{\uvu}_{3}} ,   \   \underbrace{ \Km_{2,2}   \! -\! \Km_{2,1}  \Km_{1,1}^{-1}\Km_{1,2}}_{   \Omega_{3}}  \bigr).     \label{eq:congau9925}
 \end{align}  
Finally, plugging \eqref{eq:K22} into  \eqref{eq:congau9925} leads to the following conclusion:   
\[  \uvu | (\yvu_2, \yvu_1, w) \sim \Cc\Nc ( \hat{\uvu}_{3},   \Omega_{3}) \]
where $\hat{\uvu}_{3} =  \hat{\uvu}_{2}  +   \Omega_{2} \Am^\H_{2}( \Am_{2}\Omega_{2}\Am^\H_{2}  + I_N )^{-1} (\yvut[2] - \Am_{2} \hat{\uvu}_{2} )$ and $\Omega_{3}= 
\Omega_{2}- \Omega_{2} \Am^\H_{2} ( \Am_{2}\Omega_{2}\Am^\H_{2}  + I_N )^{-1}  \Am_{2}\Omega_{2}$, as defined in \eqref{eq:mmseuv12}.
Let $\vvu_3 \defeq \uvu - \hat{\uvu}_{3}$. Then, the conditional density of $\vvu_3 $ given $(\yvu_2,\yvu_1 w) $ is 
\[ \vvu_3 | (\yvu_2, \yvu_1, w) \sim  \Cc\Nc ( \zerou,   \Omega_{3}).\] 
Note that  $\hat{\uvu}_{3}$ is conditionally independent of  $ \vvu_3$ given $(\yvu_1 , w)$,
 since 
 \begin{align}
     \text{Cov} (\hat{\uvu}_{3}, \vvu_3  \big| \yvu_1 , w)    \defeq \E\Bigl[  \bigl(\hat{\uvu}_{3} -  \E[ \hat{\uvu}_{3} | \yvu_1, w] \bigr) \bigl( \vvu_3  -  \E[ \vvu_3 | \yvu_1, w] \bigr)^\H  \big| \yvu_1 , w \Bigr]    =\mathbf{0} \non
 \end{align} 
 and  the  vectors  $\hat{\uvu}_{3}$ and   $ \vvu_3$ are two jointly proper Gaussian vectors given $(\yvu_1 , w)$.
The lack of correlation implies independence for two jointly proper Gaussian vectors.

\subsection{Proof for the general case when $t=4, 5,\cdots, T$}

For the general case when $t=4, 5,\cdots, T$, the proof is similar to the previous case.
At this point we complete the proof.

\section{Proof of Lemma~\ref{lm:ExpCht}  \label{sec:lmExpCht}}

For $\hat{\hvu}_t$ defined in  \eqref{eq:mmse424} and  \eqref{eq:mmse9832}, we will prove the following bounds
\begin{align}
\E   \bigl[  \|\hat{\hvu}_t\|^2  \bigr]  & \leq     [ (t -1) \  \text{mod} \  T_c]    \label{eq:chbound11}\\ 
\E   \bigl[  \|\hat{\hvu}_t\|^2  \bigr]   &\leq    M     \label{eq:chbound12}  \\ 
\E   \bigl[  \|\hat{\hvu}_t\|^4  \bigr]   &\leq      2 [ (t -1) \ \text{mod} \  T_c]^2 + 5 [ (t -1) \  \text{mod} \  T_c]      \label{eq:chbound13}    \\ 
\E   \bigl[  \|\hat{\hvu}_t\|^4  \bigr]   &\leq     M^2 +2M    \label{eq:chbound14}     \\ 
 \E \bigl[  ( \|\hat{\hvut}\|^2 +1)^2\bigr]  & \leq   \min\bigl\{M^2 + 4M +1 ,  \     2 [ (t-1) \  \text{mod} \  T_c]^2 + 7 [ (t -1) \ \text{mod} \  T_c] +1    \bigr\}  \label{eq:chbound15}  
\end{align}
for $t=1,2, \cdots, n$.
\vspace{5pt}

Let us provide some lemmas that will be used in our proof. At first we rewrite the definitions of $\hat{\hvut}$ and $\Omega_{t}$  in \eqref{eq:mmse424} and \eqref{eq:mmse9832}  as
\begin{align}
\hat{\hvu}_{t+1} & \defeq \hat{\hvu}_{t}   +  \hat{\hvu}_{t, t+1} , \quad \quad      \hat{\hvu}_{t,t+1}  \defeq   \frac{\Omega_{t} \xvu^{*}_{t} (\yt[t]  - \xvu^\T_{t} \hat{\hvu}_{t} )  }{  \xvu^\T_{t} \Omega_{t} \xvu^{*}_{t}  +1 }   \quad \text{for}\quad t +1 \neq  \ell T_c +1   \label{eq:mmse424a}  \\
  \Omega_{t+1}  & \defeq \Omega_{t} -   \frac{\Omega_{t} \xvu^{*}_{t}  \xvu^\T_{t}  \Omega_{t}^\H}{\xvu^\T_{t} \Omega_{t} \xvu^{*}_{t}  +1 }  \  \quad\quad\quad\quad\quad  \quad\quad\quad\quad\quad   \quad \text{for}\quad t +1 \neq  \ell T_c +1 \label{eq:mmse9832a}
  \end{align}
 and   $\hat{\hvu}_{\ell T_c +1} = \zerou$, \ $ \Omega_{\ell T_c +1} = I_M $,  $ \forall \ell  \in \{0, 1, \cdots, L-1\}$.

\vspace{5pt}
\begin{lemma}   \label{lm:independent1}
For   $\hat{\hvu}_{t+1}$ and $\hat{\hvu}_{t,t+1}$  defined as  in \eqref{eq:mmse424a} and \eqref{eq:mmse9832a},  $t \in \{ 1, 2,  \cdots, T_c -1\}$, we have 
\begin{align*}
\E  \bigl[ || \hat{\hvu}_{1} + \hat{\hvu}_{1,2} +\hat{\hvu}_{2,3} + \cdots +  \hat{\hvu}_{t,t+1} ||^2 \bigr]   =\E \bigl[  ||\hat{\hvu}_{1}||^2 +  || \hat{\hvu}_{1, 2} ||^2  + || \hat{\hvu}_{2, 3} ||^2 + \cdots  + || \hat{\hvu}_{t, t+1} ||^2\bigr] .  
\end{align*}
\end{lemma}   

\begin{proof} 
See Appendix~\ref{sec:independent1}.
\end{proof}

\vspace{5pt}

\begin{lemma}     \label{lm:boundh1}
For $\hat{\hvu}_{t,t+1}$  defined as  in \eqref{eq:mmse424a} and  \eqref{eq:mmse9832a}, 
$t \in \{ 1, 2,  \cdots, T_c -1\}$, the following bounds hold   
\begin{align}
 \E \bigl[ ||\hat{\hvu}_{t, t+1} ||^2  \big|   \Wme ,  \yv^{t-1}  \bigr]   & \leq  1,   \label{eq:2boundh1} \\
  \E [ ||\hat{\hvu}_{t, t+1} ||^2 ] & \leq  1  .   \label{eq:2boundh2}
\end{align}
 \end{lemma}
\vspace{2pt}

\begin{proof} 
See Appendix~\ref{sec:boundh1}.
\end{proof}

\vspace{5pt}

\begin{lemma}     \label{lm:boundh41} 
For $\hat{\hvu}_{t,t+1}$  defined as  in \eqref{eq:mmse424a} and  \eqref{eq:mmse9832a}, 
$t \in \{ 1, 2,  \cdots, T_c -1\}$, the following inequalities hold  
\begin{align}
 \E \bigl[ ||\hat{\hvu}_{t, t+1} ||^4  \big|   \Wme ,  \yv^{t-1}  \bigr]   & \leq  3,   \label{eq:4boundh11} \\
  \E [ ||\hat{\hvu}_{t, t+1} ||^4 ]  & \leq  3 .   \label{eq:4boundh12}
\end{align}
 \end{lemma}
\vspace{2pt}

\begin{proof} 
See Appendix~\ref{sec:boundh41}.
\end{proof} 

\vspace{5pt}

\begin{lemma}   \cite[Theorem~6]{ST:96}  \label{lm:fourthmoment}
Let  $\uvu \in \Cc^{M\times 1} \  \sim \  \Cc\Nc ( \zerou ,   \Omega)$.  For a fixed Hermitian matrix $A\in \Cc^{M\times M}$, then
\begin{align*}
\E [  ( \uvu^{\H} A  \uvu )^2  ] &=  2 \trace ( A \Omega A \Omega)  + ( \trace ( A \Omega))^2.
\end{align*}
\end{lemma}

\vspace{5pt}
\begin{lemma}   \label{lm:independent}
For  $\hat{\hvu}_{t+1} =  \hat{\hvu}_{t}  +  \hat{\hvu}_{t,t+1}$  defined   in  \eqref{eq:mmse424a} and \eqref{eq:mmse9832a},  $t \in \{ 1, 2,  \cdots, T_c -1\}$, we have 
\begin{align*}
\E  \bigl[ ||  \hat{\hvu}_{t+1} ||^4 \bigr]   \leq  \E \bigl[   \| \hat{\hvu}_{t}\|^4 \bigr]     +  4 \cdot  \E \bigl[   \| \hat{\hvu}_{t} \|^2  \bigr]  + 3 . 
\end{align*}
\end{lemma}   

\begin{proof} 
See Appendix~\ref{sec:independent}.
\end{proof}

\vspace{10pt}

Now we are ready to prove \eqref{eq:chbound1}-\eqref{eq:chbound5} in Lemma~\ref{lm:ExpCht} (or equivalently, \eqref{eq:chbound11}-\eqref{eq:chbound15}).  

\vspace{10pt}

{\bf \emph{Proof of \eqref{eq:chbound1}:}}  At first we focus on the case of $t \in \{ 1, 2,  \cdots, T_c\}$  and prove \eqref{eq:chbound1} in Lemma~\ref{lm:ExpCht} (or equivalently, \eqref{eq:chbound11}):  
\begin{align}
\E    \bigl[  \|\hat{\hvu}_t\|^2  \bigr]   
&= \E \bigl[  || \hat{\hvu}_{1, 2} ||^2  + || \hat{\hvu}_{2, 3} ||^2 + \cdots  + || \hat{\hvu}_{t-1, t} ||^2 \bigr]   \label{eq:independent333}\\
& \leq  1 + \cdots + 1     \label{eq:boundh1333}\\
& = t  - 1           \label{eq:bfinal323}
\end{align}
where  \eqref{eq:independent333} results from Lemma~\ref{lm:independent1};
\eqref{eq:boundh1333} follows from Lemma~\ref{lm:boundh1}.  
For the general case of $t \in \{ 1, 2,  \cdots, n\}$, we note that  $\hat{\hvu}_t$  is a function of  $(\xvu^{t-1}_{T_c \lfloor \frac{t-1}{T_c}\rfloor  + 1} , \yv^{t-1}_{T_c \lfloor \frac{t-1}{T_c}\rfloor  + 1})$ (cf.~\eqref{eq:mmse424}, \eqref{eq:mmse9832}), where $\yv^{t-1}_{T_c \lfloor \frac{t-1}{T_c}\rfloor  + 1}$ corresponds to the  channel outputs (up to time $t-1$) within the current channel block associated with time $t$.  
We  also note that  the previous result in \eqref{eq:bfinal323}   depends only on the number of channel outputs within the current channel block.
Therefore, one can easily follow the previous steps and show that 
\begin{align}
\E   \bigl[  \|\hat{\hvu}_t\|^2  \bigr]   
&=  \sum_{ i = T_c \lfloor \frac{t-1}{T_c}\rfloor  + 1 }^{t -1 }  \E [  || \hat{\hvu}_{i, i+1} ||^2 ]  \non \\
& \leq  1+ 1 + \cdots + 1    \label{eq:boundh1888}\\
& =  [ (t-1)   \  \text{mod} \  T_c] ,    \quad  \quad  t \in \{ 1, 2,  \cdots, n\}        \label{eq:bfinal883}
\end{align}
 where $\hat{\hvu}_{i,i+1} \defeq   \frac{\Omega_{i} \xvu^{*}_{i} (\yv_{i} - \xvu^\T_{i} \hat{\hvu}_{i} )  }{  \xvu^\T_{i} \Omega_{i} \xvu^{*}_{i} +1 } $  for $ T_c \lfloor \frac{t-1}{T_c}\rfloor  + 1 \leq  i  \leq t-1$, and  $\hat{\hvu}_{i}$ and $\Omega_{i}$ are defined in \eqref{eq:mmse424} and \eqref{eq:mmse9832};  \eqref{eq:boundh1888}  is again from Lemma~\ref{lm:boundh1}.

\vspace{10pt}

{\bf \emph{Proof of \eqref{eq:chbound2}:}}  We now prove  \eqref{eq:chbound2} in Lemma~\ref{lm:ExpCht} (or \eqref{eq:chbound12}):
\begin{align}
\E    \bigl[  \|\hat{\hvu}_t\|^2  \bigr]   
&\leq    \E    \bigl[  \|\hat{\hvu}_t\|^2  \bigr]    +  \E \bigl[  \|\tilde{\hvu}_t\|^2  \bigr]       \non\\
& =     \E    \bigl[  \|\hat{\hvu}_t  +  \tilde{\hvu}_t\|^2  \bigr]  -   \underbrace{  \E \bigl[ \hat{\hvu}_t^\H  \tilde{\hvu}_t  \bigr]}_{=0}   -     \underbrace{\E \bigl[ \tilde{\hvu}_t^\H  \hat{\hvu}_t   \bigr]}_{=0}    \label{eq:boundh357789}\\
& =  \E    \bigl[  \| \hvut \|^2  \bigr]           \label{eq:boundh3688}\\
&= M  \label{eq:bfinal2578}
\end{align}
where $\tilde{\hvu}_t \defeq \hvut  - \hat{\hvu}_t$;
\eqref{eq:boundh357789} is from the identity that $ \| \avu  + \bvu \|^2  =  \|\avu\|^2 +  \|\bvu\|^2  +  \avu^\H \bvu +   \bvu^\H \avu$ for any two vectors  $\avu, \bvu \in \Cc^{M\times 1}$;
\eqref{eq:boundh3688} follows from the fact that  $\E \bigl[ \hat{\hvu}_t^\H  \tilde{\hvu}_t  \bigr] =  \E \bigl[  \E \bigl[  \hat{\hvu}_t^\H  \tilde{\hvu}_t  \big| \Wme ,  \yv^{t-1}\bigr] \bigr] = \E[ 0 ] =0$ by using the results that $ \tilde{\hvu}_t  \big| (\Wme ,  \yv^{t-1})  \sim \mathcal{CN}( \zerou, \Omega_{t}) $ and that  $ \hat{\hvu}_t$ is deterministic given $(\Wme ,  \yv^{t-1})$;  similarly, $ \E \bigl[ \tilde{\hvu}_t^\H  \hat{\hvu}_t   \bigr] =0$;
\eqref{eq:bfinal2578} is from the assumption that  $\hvut \sim \mathcal{CN}( \zerou,  I_M) $.

\vspace{10pt}

{\bf \emph{Proof of \eqref{eq:chbound3}:}} We now focus on the case of $t \in \{ 1, 2,  \cdots, T_c\}$ and  prove  \eqref{eq:chbound3} in Lemma~\ref{lm:ExpCht} (or  \eqref{eq:chbound13}):
\begin{align}
\E    \bigl[  \|\hat{\hvu}_t\|^4  \bigr]   
&\leq  \E \bigl[   \| \hat{\hvu}_{t-1}\|^4 \bigr]     +  4 \cdot  \E \bigl[   \| \hat{\hvu}_{t-1} \|^2  \bigr]  + 3    \label{eq:independent4677}\\
& \leq  \E \bigl[   \| \hat{\hvu}_{t-1}\|^4 \bigr]     +  4 (t-1)  + 3   \label{eq:boundh2577}\\
& \leq     \E \bigl[   \| \hat{\hvu}_{1}\|^4 \bigr]   +  4 \sum_{k=1}^{t-1} k   + 3 (t-1)  \label{eq:boundh3678}\\
& =  2(t-1)^2 + 5(t-1)            \label{eq:bfinal9926}
\end{align}
where \eqref{eq:independent4677} follows from Lemma~\ref{lm:independent};
\eqref{eq:boundh2577}  is from the result in \eqref{eq:bfinal323}; 
\eqref{eq:boundh3678} follows by repeating the steps of  \eqref{eq:independent4677} and \eqref{eq:boundh2577};
\eqref{eq:bfinal9926}  uses the definition that $\hat{\hvu}_{1} = \zerou$.
For the general case of $t \in \{ 1, 2,  \cdots, n\}$, we again note that  $\hat{\hvu}_t$  is a function of  $(\xvu^{t-1}_{T_c \lfloor \frac{t-1}{T_c}\rfloor  + 1} , \yv^{t-1}_{T_c \lfloor \frac{t-1}{T_c}\rfloor  + 1})$.  
Therefore, one can easily follow the previous steps and show that 
\begin{align}
\E    \bigl[  \|\hat{\hvu}_t\|^4  \bigr]   
& \leq     \E \bigl[   \| \hat{\hvu}_{T_c \lfloor \frac{t-1}{T_c}\rfloor  + 1}\|^4 \bigr]   +   \sum_{k=1}^{ [(t-1) \text{mod}  T_c ] } \!\!\!\!\! 4 k  \  +  \ 3 [(t -1) \ \text{mod} \  T_c]   \label{eq:boundh9663}\\
& =  2 [ (t -1) \  \text{mod} \  T_c]^2 + 5 [(t -1) \  \text{mod} \  T_c]           \label{eq:bfinal4377}
\end{align}
where \eqref{eq:bfinal4377}  uses the definition of $  \hat{\hvu}_{T_c \lfloor \frac{t-1}{T_c}\rfloor  + 1} = \zerou$.

\vspace{10pt}

{\bf \emph{Proof of \eqref{eq:chbound4}:}}  We now prove   \eqref{eq:chbound4} in Lemma~\ref{lm:ExpCht} (or \eqref{eq:chbound14}):
\begin{align}
&\E    \bigl[  \|\hat{\hvu}_t\|^4  \bigr]     \non \\
\leq  &  \E    \bigl[  \|\hat{\hvu}_t\|^4  \bigr]    +  \underbrace{ \E \bigl[  \|\tilde{\hvu}_t\|^4  \bigr]}_{\geq 0}   +   \underbrace{  \E    \bigl[ 2 \|\hat{\hvu}_t\|^2 \|\tilde{\hvu}_t\|^2 \bigr] }_{\geq 0}  +  \underbrace{  \E    \bigl[ 4\text{Re}^2(\hat{\hvu}_t^\H \tilde{\hvu}_t)\bigr] }_{\geq 0}  + \underbrace{ \E    \bigl[ 4 ( \|\hat{\hvu}_t\|^2 +  \|\tilde{\hvu}_t\|^2) \cdot \text{Re} (\hat{\hvu}_t^\H \tilde{\hvu}_t) \bigr] }_{= 0}      \label{eq:boundh3788}\\
=  &   \E    \bigl[  \|\hat{\hvu}_t  +  \tilde{\hvu}_t\|^4  \bigr]    \label{eq:boundh3677}\\
 =&  \E    \bigl[  \| \hvut \|^4  \bigr]           \non\\
=&  M^2 +2M  \label{eq:bfinal8266}
\end{align}
where $\tilde{\hvu}_t \defeq \hvut  - \hat{\hvu}_t$;
\eqref{eq:boundh3677} stems from the identity that 
\begin{align}
 &   \| \avu  + \bvu \|^4  =  \|\avu\|^4 +  \|\bvu\|^4  +  2   \|\avu\|^2 \|\bvu\|^2  + 4   \text{Re}^2(\avu^\H \bvu) + 4 ( \|\avu\|^2 +  \|\bvu\|^2) \cdot \text{Re} (\avu^\H \bvu)   \non
\end{align}
for any two vectors  $\avu, \bvu \in \Cc^{M\times 1}$, where $\text{Re}(\bullet )$ denotes the real part of the argument;
\eqref{eq:bfinal8266} follows from Lemma~\ref{lm:fourthmoment};
\eqref{eq:boundh3788} results from the fact that 
\begin{align}  
 &  \E \bigl[ 4 ( \|\hat{\hvu}_t\|^2 +  \|\tilde{\hvu}_t\|^2) \cdot \text{Re} (\hat{\hvu}_t^\H \tilde{\hvu}_t) \bigr]   \non \\
 = &\E \bigl[ \ \E \bigl[ 4 ( \|\hat{\hvu}_t\|^2 +  \|\tilde{\hvu}_t\|^2) \cdot \text{Re} (\hat{\hvu}_t^\H \tilde{\hvu}_t) \ \big| (\Wme ,  \yv^{t-1}) \bigr] \ \bigr]   \non\\
 =& \E \bigl[  \  4  \|\hat{\hvu}_t\|^2  \cdot   \underbrace{ \E \bigl[  \text{Re} (\hat{\hvu}_t^\H \tilde{\hvu}_t) \ \big| (\Wme ,  \yv^{t-1}) \bigr] }_{= 0} +  \underbrace{  \E \bigl[ 4  \|\tilde{\hvu}_t\|^2 \cdot \text{Re} (\hat{\hvu}_t^\H \tilde{\hvu}_t) \ \big| (\Wme ,  \yv^{t-1}) \bigr]}_{= 0}    \ \bigr]  \label{eq:boundh92467}\\
=  & \E [ 0 + 0 ]     \label{eq:boundh9356}\\
= &0   \non
\end{align}
where \eqref{eq:boundh92467}   results from the fact that  $\hat{\hvu}_{t}$ is deterministic given  $(\Wme ,  \yv^{t-1})$; \eqref{eq:boundh9356} follows from the identities that $ \E [ \text{Re} (\au^\H \bvu) ]  = \text{Re} ( \E[ \au^\H \bvu ])  = 0$  and  $ \E [   \|\bvu\|^2 \cdot \text{Re} (\au^\H \bvu) ]  = \text{Re} ( \E [    \au^\H \bvu \cdot \|\bvu\|^2  ])  = 0   $ for a fixed  vector $\au$ and a Gaussian  vector $\bvu \sim \Cc\Nc ( \zerou,   \Km )$. Note that the \emph{odd}-order moments of a complex proper Gaussian vector are  zeros (see, e.g., \cite{Kostas:02}).

\vspace{10pt}

{\bf \emph{Proof of \eqref{eq:chbound5}:}}  Finally,  \eqref{eq:chbound5} in Lemma~\ref{lm:ExpCht} follows from \eqref{eq:chbound1}-\eqref{eq:chbound4}. Specifically,  combining  \eqref{eq:chbound2} and  \eqref{eq:chbound4} gives $ \E \bigl[  ( \|\hat{\hvut}\|^2 +1)^2\bigr]  =   \E [  \|\hat{\hvut}\|^4 ]  +  2 \E [ \|\hat{\hvut}\|^2 ]  +1 \leq  M^2 +4M +1$, while combining  \eqref{eq:chbound1} and  \eqref{eq:chbound3}  gives  $ \E \bigl[  ( \|\hat{\hvut}\|^2 +1)^2\bigr]  =   \E [  \|\hat{\hvut}\|^4 ]  +  2 \E [ \|\hat{\hvut}\|^2 ]  +1  \leq  2 [ (t -1) \  \text{mod} \  T_c]^2 + 7 [ (t -1) \  \text{mod} \  T_c] +1$.  At this point it proves  Lemma~\ref{lm:ExpCht}.  For Lemmas~\ref{lm:independent1}-\ref{lm:boundh41} and Lemma~\ref{lm:independent},  which have been used above, the proofs are given as follows.

\subsection{ Proof of Lemma~\ref{lm:independent1}  \label{sec:independent1}}

We here prove that, for   $\hat{\hvu}_{t+1}$ and $\hat{\hvu}_{t,t+1}$  defined as  in \eqref{eq:mmse424a} and \eqref{eq:mmse9832a},  $t \in \{ 1, 2,  \cdots, T_c -1\}$, we have 
\begin{align*}
\E  \bigl[ || \hat{\hvu}_{1} + \hat{\hvu}_{1,2} +\hat{\hvu}_{2,3} + \cdots +  \hat{\hvu}_{t,t+1} ||^2 \bigr]   =\E \bigl[  ||\hat{\hvu}_{1}||^2 +  || \hat{\hvu}_{1, 2} ||^2  + || \hat{\hvu}_{2, 3} ||^2 + \cdots  + || \hat{\hvu}_{t, t+1} ||^2\bigr] .  
\end{align*}

 For  the case of $t \in \{ 1, 2,  \cdots, T_c-1\}$, we have
\begin{align}
 & \E  \Bigl[ || \hat{\hvu}_{1} + \hat{\hvu}_{1,2} +\hat{\hvu}_{2,3} + \cdots +  \hat{\hvu}_{t,t+1} ||^2  \Bigr]      \nonumber\\
  &=  \E  \Bigl[ ||  \hat{\hvu}_{t}+  \hat{\hvu}_{t,t+1} ||^2  \Bigr]     \label{eq:split9942}\\  
 & =  \E  \Bigl[  \  \E  \Bigl[  || \hat{\hvu}_{t} + \hat{\hvu}_{t, t+1} ||^2 \  \big|   \Wme,  \yv^{t-1} \Bigr] \   \Bigr]  \label{eq:split8325}\\ 
  & =  \E  \Bigl[ \   \E  \Bigl[    || \hat{\hvu}_{t} ||^2    +   ||\hat{\hvu}_{t, t+1} ||^2 +  \hat{\hvu}_{t}^{\H} \ \hat{\hvu}_{t, t+1}  +   \hat{\hvu}_{t, t+1}^{\H}  \ \hat{\hvu}_{t}      \  \big|   \Wme,  \yv^{t-1} \Bigr]    \ \Bigr]    \nonumber\\ 
  & =  \E  \Bigl[   \   || \hat{\hvu}_{t}  ||^2 +\E  \bigl[   || \hat{\hvu}_{t, t+1} ||^2  \ \big| \Wme,  \yv^{t-1} \bigr]  +0 + 0  \  \Bigr]  \label{eq:ind828}\\ 
  & =  \E  \bigl[     || \hat{\hvu}_{t}  ||^2  \bigr]  +   \E  \Bigl[ \E  \bigl[   || \hat{\hvu}_{t, t+1} ||^2  \ \big| \Wme,  \yv^{t-} \bigr]  \  \Bigr]  \nonumber \\ 
    & =  \E  \bigl[     || \hat{\hvu}_{t}  ||^2  \bigr]  +   \E  \bigl[   || \hat{\hvu}_{t, t+1} ||^2  \bigr]   \nonumber \\ 
    & =  \E  \bigl[     || \hat{\hvu}_{t-1}  +  \hat{\hvu}_{t-1, t}  ||^2  \bigr]  +   \E  \bigl[   || \hat{\hvu}_{t, t+1} ||^2 \bigr]   \label{eq:split0193}\\ 
        & = \E  \bigl[     || \hat{\hvu}_{t-1}   ||^2  \bigr] +   \E  \bigl[     ||  \hat{\hvu}_{t-1, t}  ||^2  \bigr]  +   \E  \bigl[   || \hat{\hvu}_{t, t+1} ||^2 \bigr]   \label{eq:split9245}\\   
& \quad \vdots  \nonumber\\
& =   \E [ || \hat{\hvu}_{1} ||^2 ] +  \E[ || \hat{\hvu}_{1, 2} ||^2]  +   \E  [|| \hat{\hvu}_{2, 3} ||^2 ] + \cdots  + \E[ || \hat{\hvu}_{t, t+1} ||^2  ]  \label{eq:split42653}
\end{align}
where \eqref{eq:split9942} uses the definitions of  $\hat{\hvu}_{t+1}$ and  $\hat{\hvu}_{t,t+1}$ in \eqref{eq:mmse424a}  and  \eqref{eq:mmse9832a};  
\eqref{eq:split8325} is from the identity that  $\E[\boldsymbol{a}] = \E[\E[\boldsymbol{a} | \boldsymbol{b} ]]$ for  random $\boldsymbol{a}$ and $\boldsymbol{b}$;
 \eqref{eq:ind828} follows from the fact that   $\hat{\hvu}_{t, t+1}$ is a complex Gaussian  vector with zero mean given $(\Wme,  \yv^{t-1} )$ (cf.~Lemma~\ref{lm:densityh}) and the fact that $\hat{\hvu}_{t}$ is a deterministic function of $(\Wme,  \yv^{t-1} )$ given the encoding maps in \eqref{eq:mapx}; 
 \eqref{eq:split0193} uses the  definitions of  $ \hat{\hvu}_{t-1}$  and  $\hat{\hvu}_{t-1, t}$ in \eqref{eq:mmse424a}  and  \eqref{eq:mmse9832a}; 
 \eqref{eq:split9245} follows from the previous steps in \eqref{eq:split9942}-\eqref{eq:split0193}; 
  \eqref{eq:split42653} follows from the same step   in  \eqref{eq:split9245}.  Note that $\hat{\hvu}_{1} = \zerou$.

\subsection{ Proof of Lemma~\ref{lm:boundh1}  \label{sec:boundh1}}

We will prove that, for $\hat{\hvu}_{t,t+1}$  defined as  in \eqref{eq:mmse424a} and  \eqref{eq:mmse9832a}, 
$t \in \{ 1, 2,  \cdots, T_c -1\}$, the following bounds hold   
\begin{align*}
 \E \bigl[ ||\hat{\hvu}_{t, t+1} ||^2  \big|   \Wme ,  \yv^{t-1}  \bigr]   & \leq  1,   \\
  \E [ ||\hat{\hvu}_{t, t+1} ||^2 ] & \leq  1  . 
\end{align*}
We will just prove the first  inequality, as the second inequality follows immediately from the first  inequality and  the identity that  $ \E [ ||\hat{\hvu}_{t, t+1} ||^2   ]  =   \E \bigl[ \   \E \bigl[  ||\hat{\hvu}_{t, t+1} ||^2   \big|   \Wme ,  \yv^{t-1}  \bigr]  \ \bigr] $.

Given that  $\hat{\hvu}_{t, t+1} = \frac{\Omega_{t} \xvu^{*}_{t} (\yv_{t} -  \xvu^\T_{t} \hat{\hvu}_{t} )  }{  \xvu^\T_{t} \Omega_{t} \xvu^{*}_{t}  +1 }$, for $t \in \{ 1, 2,  \cdots, T_c -1\}$, we have
\begin{align}
 & \E \bigl[  ||\hat{\hvu}_{t, t+1} ||^2  \big|   \Wme ,  \yv^{t-1} \bigr]        \nonumber\\
  & = \E \Bigl[  \Bigl(\frac{\Omega_{t} \xvu^{*}_{t} (\yv_{t} \!- \! \xvu^\T_{t} \hat{\hvu}_{t} )  }{  \xvu^\T_{t} \Omega_{t} \xvu^{*}_{t}  +1 } \Bigr)^\H  \Bigl(\frac{\Omega_{t} \xvu^{*}_{t} (\yv_{t} \!-\! \xvu^\T_{t} \hat{\hvu}_{t} )  }{  \xvu^\T_{t} \Omega_{t} \xvu^{*}_{t}  +1 } \Bigr)  \Big|   \Wme ,  \yv^{t-1} \Bigr] \nonumber\\
 & = \E \Bigl[ \   \frac{ \bigl( \Omega_{t} \xvu^{*}_{t} ( \xvu^\T_{t}  \tilde{\hvut}   + \zv_{t} )   \bigr)^\H  \bigl( \Omega_{t} \xvu^{*}_{t} ( \xvu^\T_{t} \tilde{\hvut}  + \zv_{t} ) \bigr) }{\bigl( \xvu^\T_{t} \Omega_{t} \xvu^{*}_{t}  +1 \bigr)^2}  \Big|   \Wme ,  \yv^{t-1} \Bigr]  \label{eq:tildeh992}\\ 
  & = \E \Bigl[ \   \frac{   \xvu^\T_{t}  \Omega_{t} \Omega_{t} \xvu^{*}_{t}  \cdot  | (\xvu^\T_{t}  \tilde{\hvut}   + \zv_{t} )|^2}{\bigl( \xvu^\T_{t} \Omega_{t} \xvu^{*}_{t}  +1 \bigr)^2}  \Big|   \Wme ,  \yv^{t-1} \Bigr]  \non \\ 
   & =    \frac{   \xvu^\T_{t}  \Omega_{t} \Omega_{t} \xvu^{*}_{t}   \cdot \E \bigl[    ( \xvu^\T_{t}  \tilde{\hvut}   \tilde{\hvut}^{\H} \xvu^{*}_{t}     + \zv_{t}  \zv_{t}^{*}  + \xvu^\T_{t}  \tilde{\hvut} \zv_{t}^{*} + \zv_{t} \tilde{\hvut}^{\H} \xvu^{*}_{t}    ) \big|   \Wme ,  \yv^{t-1}  \bigr] }{\bigl( \xvu^\T_{t} \Omega_{t} \xvu^{*}_{t}  +1 \bigr)^2}      \label{eq:lmest8355}  \\ 
  & =   \frac{   \xvu^\T_{t}  \Omega_{t} \Omega_{t} \xvu^{*}_{t}   (    \xvu^\T_{t}   \Omega_{t}   \xvu^{*}_{t}     +  1  + 0 + 0    )}{\bigl( \xvu^\T_{t} \Omega_{t} \xvu^{*}_{t}  +1 \bigr)^2}       \label{eq:lmest5267}  \\ 
  & =  \frac{    \xvu^\T_{t} \Omega_{t} \Omega_{t} \xvu^{*}_{t}  }{  \xvu^\T_{t} \Omega_{t} \xvu^{*}_{t}  +1 }   \nonumber \\
& \leq   \frac{    \xvu^\T_{t} \Omega_{t} \xvu^{*}_{t}  +1  }{  \xvu^\T_{t} \Omega_{t} \xvu^{*}_{t}  +1 }     \label{eq:lmest92834}\\
& = 1  
\end{align}
where  $\hat{\hvu}_{t,t+1}$  is defined  in \eqref{eq:mmse424a} and  \eqref{eq:mmse9832a};
 \eqref{eq:tildeh992} uses the definition of  $\tilde{\hvut}  \defeq \hvut  - \hat{\hvu}_{t}$ and the fact that $\yv_{t} - \xvu^\T_{t} \hat{\hvu}_{t} = \xvu^\T_{t}  \tilde{\hvut}   + \zv_{t} $;  
\eqref{eq:lmest8355} results from  the facts that $\xvut$ is a deterministic function of $(\yv^{t-1}, \Wme)$ and  that $\Omega_{t}$ is a deterministic function of $(\yv^{t-2}, \Wme)$ given the encoding maps in \eqref{eq:mapx};
\eqref{eq:lmest5267} follows from  the facts that $\tilde{\hvut} |  (\yv^{t-1}, \Wme) \sim \Cc\Nc ( \zerou,   \Omega_{t})$ and that $ \zv_{t}$ is independent of $\tilde{\hvut}$; 
   \eqref{eq:lmest92834} follows from that  
\begin{align}
 \xvu^\T_{t} \Omega_{t} \Omega_{t} \xvu^{*}_{t}  \leq   \xvu^\T_{t} \Omega_{t}\xvu^{*}_{t}   \leq     \xvu^\T_{t} \Omega_{t}\xvu^{*}_{t}    + 1,  \non
 \end{align}
 where the first inequality follows from that $ \mathbf{0} \preceq \Omega_{t}  \preceq  I_M $ (cf.~Lemma~\ref{lm:estimate1}) and that 
$ \xvu^\T_{t} \Omega_{t}\xvu^{*}_{t} - \xvu^\T_{t} \Omega_{t} \Omega_{t} \xvu^{*}_{t} \defeq  \xvu^\T_{t} \Um \Lambda \Um^\H  \xvu^{*}_{t} - \xvu^\T_{t} \Um \Lambda  \Um^\H  \Um \Lambda \Um^\H    \xvu^{*}_{t} = \xvu^\T_{t} \Um (\Lambda - \Lambda^2)  \Um^\H  \xvu^{*}_{t}  \geq 0$  by using the  singular value decomposition  of  $\Omega_{t} \defeq  \Um \Lambda \Um^\H$, where $\Um$ and $\Lambda$ are the unitary matrix and diagonal matrix respectively. Note that if $ \mathbf{0} \preceq \Omega_{t}  \preceq  I_M $, then $\Um ( \Lambda -  \Lambda^2) \Um^\H   \succeq  \mathbf{0}$.  At this point we complete the proof.

\subsection{ Proof of Lemma~\ref{lm:boundh41}  \label{sec:boundh41}}

We will prove that, for $\hat{\hvu}_{t,t+1}$  defined as  in \eqref{eq:mmse424a} and  \eqref{eq:mmse9832a}, 
$t \in \{ 1, 2,  \cdots, T_c -1\}$, the following inequalities hold  
\begin{align}
 \E \bigl[ ||\hat{\hvu}_{t, t+1} ||^4  \big|   \Wme ,  \yv^{t-1}  \bigr]   & \leq  3,   \label{eq:4boundh111} \\
  \E [ ||\hat{\hvu}_{t, t+1} ||^4 ]  & \leq  3 .   \label{eq:4boundh121}
\end{align}

We will just prove the first  inequality in \eqref{eq:4boundh111}, as the second inequality in \eqref{eq:4boundh121} follows immediately from  \eqref{eq:4boundh111} and  the identity that  $ \E [ ||\hat{\hvu}_{t, t+1} ||^4 ]  =   \E \bigl[  \E \bigl[ ||\hat{\hvu}_{t, t+1} ||^4  \big|   \Wme ,  \yv^{t-1}  \bigr] \bigr] $.

The proof of \eqref{eq:4boundh111} follows from the proof steps of  Lemma~\ref{lm:boundh1}. For $\hat{\hvu}_{t,t+1}$  defined as  in \eqref{eq:mmse424a} and  \eqref{eq:mmse9832a},   $t \in \{ 1, 2,  \cdots, T_c -1\}$, we have
\begin{align}
 \E \bigl[  ||\hat{\hvu}_{t, t+1} ||^4   \big|   \Wme ,  \yv^{t-1}  \bigr]     & = \E \Bigl[ \ \Bigl(\Bigl(\frac{\Omega_{t} \xvu^{*}_{t} (\yv_{t} \!- \! \xvu^\T_{t} \hat{\hvu}_{t} )  }{  \xvu^\T_{t} \Omega_{t} \xvu^{*}_{t}  +1 } \Bigr)^\H  \Bigl(\frac{\Omega_{t} \xvu^{*}_{t} (\yv_{t} \!-\! \xvu^\T_{t} \hat{\hvu}_{t} )  }{  \xvu^\T_{t} \Omega_{t} \xvu^{*}_{t}  +1 } \Bigr) \Bigr)^2   \Big|   \Wme ,  \yv^{t-1}   \Bigr] \nonumber\\
 & = \E \Bigl[ \ \frac{\bigl( \bigl( \Omega_{t} \xvu^{*}_{t} ( \xvu^\T_{t}  \tilde{\hvut}   + \zv_{t} )   \bigr)^\H   \Omega_{t} \xvu^{*}_{t} ( \xvu^\T_{t} \tilde{\hvut}  + \zv_{t} ) \bigr)^2 }{\bigl( \xvu^\T_{t} \Omega_{t} \xvu^{*}_{t}  +1 \bigr)^4} \Big|   \Wme ,  \yv^{t-1}  \Bigr]  \label{eq:4tildeh992}\\ 
  & = \E \Bigl[ \   \frac{\bigl(    \xvu^\T_{t} \Omega_{t} \Omega_{t}   \xvu^{*}_{t}\bigr)^2  \cdot    \bigl| \bigl( \xvu^\T_{t} \tilde{\hvut}  + \zv_{t} \bigr)\bigr|^4 }{\bigl( \xvu^\T_{t} \Omega_{t} \xvu^{*}_{t}  +1 \bigr)^4}  \Big|   \Wme ,  \yv^{t-1}   \Bigr]  \non\\ 
        & =    \frac{\bigl(    \xvu^\T_{t} \Omega_{t} \Omega_{t}   \xvu^{*}_{t}\bigr)^2  }{\bigl( \xvu^\T_{t} \Omega_{t} \xvu^{*}_{t}  +1 \bigr)^4}   \ \cdot \E \Bigl[      \bigl| \bigl( \xvu^\T_{t} \tilde{\hvut}  + \zv_{t} \bigr)\bigr|^4     \Big|   \Wme ,  \yv^{t-1}  \Bigr]   \label{eq:4det256}
\end{align}
where \eqref{eq:4tildeh992} uses the definition of  $\tilde{\hvut}  \defeq \hvut  - \hat{\hvu}_{t}$;  
\eqref{eq:4det256} follows from  the facts that  $\xvut$ is a deterministic function of $(\yv^{t-1}, \Wme)$ and that $\Omega_{t}$ is a deterministic function of $(\yv^{t-2}, \Wme)$ given the encoding maps in \eqref{eq:mapx}.
Let us focus on the inner expectation term in \eqref{eq:4det256}. 
Note that,  for  two \emph{complex numbers} $\av$ and $\bv$, we have 
\begin{align}
 &   | ( \av  + \bv )|^4  =  |\av|^4 +  |\bv|^4  +  2   |\av|^2 |\bv|^2  + 4   \text{Re}^2(\av \bv^*) + 4 ( |\av|^2 +  |\bv|^2) \cdot \text{Re} (\av \bv^*) .  \label{eq:expend3245}
\end{align}
In the following, we will replace $\av$ and $\bv$ with $\xvu^\T_{t} \tilde{\hvut}$ and $ \zv_{t}$ respectively and compute $\E \bigl[      \bigl| \bigl( \xvu^\T_{t} \tilde{\hvut}  + \zv_{t} \bigr)\bigr|^4     \big|   \Wme ,  \yv^{t-1}  \bigr]$. At first we note that given $ \zv_{t} \sim \Cc\Nc ( 0,   1)$ and $\tilde{\hvut} |  (\yv^{t-1}, \Wme) \sim \Cc\Nc ( \zerou,   \Omega_{t})$, the following equalities hold true:
\begin{align}
  \E[  \zv_{t} ] &= 0           \label{eq:moment1}  \\
  \E[  |\zv_{t}|^2  ] &= 1   \label{eq:moment2} \\
    \E[  \zv_{t} \zv_{t}  ] &= 0   \label{eq:moment222} \\
    \E[ \zv_{t} \cdot  |\zv_{t}|^2  ] &= 0   \label{eq:moment22} \\
  \E[  |\zv_{t}|^4  ] &=  3    \label{eq:moment3} \\
    \E[  \xvu^\T_{t} \tilde{\hvut}   \ |   \Wme ,  \yv^{t-1}]  &=   0   \label{eq:moment4} \\
  \E[  |\xvu^\T_{t} \tilde{\hvut} |^2  \ |   \Wme ,  \yv^{t-1}]  &=   \xvu^\T_{t} \Omega_{t} \xvu^{*}_{t}   \label{eq:moment5} \\
  \E[  |\xvu^\T_{t} \tilde{\hvut} |^4  \  |   \Wme ,  \yv^{t-1}]  &=   3 ( \xvu^\T_{t} \Omega_{t} \xvu^{*}_{t}  )^2     \label{eq:moment6}
\end{align} 
where \eqref{eq:moment6}  follows from Lemma~\ref{lm:fourthmoment} (shown at the beginning of this section), i.e.,
$\E[  |  \xvu^\T_{t}   \tilde{\hvu}_t   |^4  |   \Wme ,  \yv^{t-1}]$ 
= $ \E[ (   \tilde{\hvu}^{\H}_t \xvu^{*}_{t} \xvu^\T_{t}   \tilde{\hvut} )^2  |   \Wme ,  \yv^{t-1}] $  =$   2 \trace ( \xvu^{*}_{t} \xvu^\T_{t}  \Omega_{t}   \xvu^{*}_{t} \xvu^\T_{t}  \Omega) $  + $ ( \trace (  \xvu^{*}_{t} \xvu^\T_{t}  \Omega_{t} ))^2$  =  $3 ( \xvu^\T_{t} \Omega_{t} \xvu^{*}_{t}  )^2 $; 
\eqref{eq:moment3} also follows from Lemma~\ref{lm:fourthmoment}.  
By using  \eqref{eq:expend3245}-\eqref{eq:moment6}, we have 
\begin{align}
   &\E \Bigl[      \bigl| \bigl( \xvu^\T_{t} \tilde{\hvut}  + \zv_{t} \bigr)\bigr|^4     \Big|   \Wme ,  \yv^{t-1}  \Bigr]    \non\\
   = & \E \Bigl[   |\xvu^\T_{t} \tilde{\hvut} |^4   +   |\zv_{t}|^4       \!+\!   2 |\xvu^\T_{t} \tilde{\hvut} |^2  \cdot   |\zv_{t}|^2 \!+\!     4   \text{Re}^2 \bigl(\xvu^\T_{t} \tilde{\hvut}   \zv_{t}^* \bigr )    \!+\!  4 ( |\xvu^\T_{t} \tilde{\hvut} |^2 +  |\zv_{t}|^2) \cdot \text{Re} \bigl(\xvu^\T_{t} \tilde{\hvut}   \zv_{t}^* \bigr )  \Big|   \Wme ,  \yv^{t-1}  \Bigr]   \label{eq:expend8525} \\    
 = & 3 ( \xvu^\T_{t} \Omega_{t} \xvu^{*}_{t}  )^2   +  3  +     2   \xvu^\T_{t} \Omega_{t} \xvu^{*}_{t}   +     \E \Bigl[ 4   \text{Re}^2 \bigl(\xvu^\T_{t} \tilde{\hvut}   \zv_{t}^* \bigr )  \Big|   \Wme ,  \yv^{t-1}  \Bigr]  + 0  \label{eq:expend4266} \\  
  = & 3 ( \xvu^\T_{t} \Omega_{t} \xvu^{*}_{t}  )^2   +  3  +     2   \xvu^\T_{t} \Omega_{t} \xvu^{*}_{t} +      2  \E \bigl[  |\xvu^\T_{t} \tilde{\hvut} |^2  \big|   \Wme ,  \yv^{t-1}  \bigr]  \cdot \E  |\zv_{t}|^2      \label{eq:expend9255}   \\   
    = & 3 ( \xvu^\T_{t} \Omega_{t} \xvu^{*}_{t}  )^2   +  3  +     4   \xvu^\T_{t} \Omega_{t} \xvu^{*}_{t}      \label{eq:expend425}                         
\end{align}
where \eqref{eq:expend8525} is from \eqref{eq:expend3245};
\eqref{eq:expend4266} follows from \eqref{eq:moment1}-\eqref{eq:moment6} as well as  the fact that $\zv_{t}$ is independent of $\xvu_{t}$ and $ \tilde{\hvut}$;
\eqref{eq:expend9255} stems from the following conclusion  for  two \emph{independent} complex random variables $\av$ and $\bv $, $\bv \sim \Cc\Nc (0, 1)$,  that is,  $ \E [4\text{Re}^2 (\av \bv^*) ] =  \E [(  \av \bv^*  + \av^* \bv    )(  \av \bv^*  + \av^* \bv    ) ]  =  \E [  2 |\av|^2 |\bv|^2    +  \av \av  (\bv \bv)^* +    (\av \av)^*\bv \bv  ]   =  2 \E [ |\av|^2 |\bv|^2] $  (cf.~\eqref{eq:moment222}). In the above we replace $\av$ and $\bv$ with $\xvu^\T_{t} \tilde{\hvut}$ and $ \zv_{t}$ respectively.  The last step in \eqref{eq:expend425}  follows from \eqref{eq:moment5}. 

By plugging \eqref{eq:expend425} into \eqref{eq:4det256}, we have 
\begin{align}
  \E [  ||\hat{\hvu}_{t, t+1} ||^4  |   \Wme ,  \yv^{t-1} ]     & =\frac{\bigl(    \xvu^\T_{t} \Omega_{t} \Omega_{t}   \xvu^{*}_{t}\bigr)^2  }{\bigl( \xvu^\T_{t} \Omega_{t} \xvu^{*}_{t}  +1 \bigr)^4}   \ \cdot \bigl( 3 ( \xvu^\T_{t} \Omega_{t} \xvu^{*}_{t}  )^2   +  3  +     4   \xvu^\T_{t} \Omega_{t} \xvu^{*}_{t}  \bigr)    \non \\
  & \leq     \frac{\bigl(    \xvu^\T_{t} \Omega_{t} \Omega_{t}   \xvu^{*}_{t}\bigr)^2  }{\bigl( \xvu^\T_{t} \Omega_{t} \xvu^{*}_{t}  +1 \bigr)^4}   \ \cdot 3 \bigl( \xvu^\T_{t} \Omega_{t} \xvu^{*}_{t}    + 1  \bigr)^2     \label{eq:4det4525} \\
    & =  3   \cdot  \frac{\bigl(    \xvu^\T_{t} \Omega_{t} \Omega_{t}   \xvu^{*}_{t}\bigr)^2  }{\bigl( \xvu^\T_{t} \Omega_{t} \xvu^{*}_{t}  +1 \bigr)^2}       \non \\ 
 & \leq  3   \cdot  \frac{\bigl( \xvu^\T_{t} \Omega_{t} \xvu^{*}_{t}  +1 \bigr)^2   }{\bigl( \xvu^\T_{t} \Omega_{t} \xvu^{*}_{t}  +1 \bigr)^2}   \label{eq:4det3567} \\  
      & =   3       \label{eq:4det0144} 
\end{align}
where \eqref{eq:4det4525} uses the fact that  $\xvu^\T_{t} \Omega_{t} \xvu^{*}_{t} \geq 0$ since   $ \Omega_{t} \succeq \mathbf{0} $ (cf.~Lemma~\ref{lm:estimate1});
\eqref{eq:4det3567} follows from the same step in \eqref{eq:lmest92834}, i.e.,  $\xvu^\T_{t} \Omega_{t} \Omega_{t}   \xvu^{*}_{t} \leq  \xvu^\T_{t} \Omega_{t} \xvu^{*}_{t}  +1$. At this point we complete the proof.

\subsection{ Proof of Lemma~\ref{lm:independent}  \label{sec:independent}}

We will prove that, for  $\hat{\hvu}_{t+1} =  \hat{\hvu}_{t}  +  \hat{\hvu}_{t,t+1}$  defined   in  \eqref{eq:mmse424a} and \eqref{eq:mmse9832a},  $t \in \{ 1, 2,  \cdots, T_c -1\}$, we have 
\begin{align*}
\E  \bigl[ ||  \hat{\hvu}_{t+1} ||^4 \bigr]   \leq  \E \bigl[   \| \hat{\hvu}_{t}\|^4 \bigr]     +  4 \cdot  \E \bigl[   \| \hat{\hvu}_{t} \|^2  \bigr]  + 3 . 
\end{align*}

 We will at first focus on the upper bound  of $\E \bigl[      \|   \hat{\hvu}_{t} +  \hat{\hvu}_{t,t+1} \|^4     \big|   \Wme ,  \yv^{t-1}  \bigr] $.  Remind that  $\hat{\hvu}_{t, t+1} = \frac{\Omega_{t} \xvu^{*}_{t} (\yv_{t} - \xvu^\T_{t} \hat{\hvu}_{t} )  }{  \xvu^\T_{t} \Omega_{t} \xvu^{*}_{t}  +1 } = \frac{\Omega_{t} \xvu^{*}_{t} ( \xvu^\T_{t}  \tilde{\hvut}  + \zv_t)  }{  \xvu^\T_{t} \Omega_{t} \xvu^{*}_{t}  +1 }$ and that $ \tilde{\hvut}  |  (\yv^{t-1}, \Wme) \sim \Cc\Nc ( \zerou,   \Omega_{t})$  (cf.~Lemma~\ref{lm:densityh}), where $\tilde{\hvut} \defeq  \hvut  - \hat{\hvu}_{t}$. Thus, one can easily conclude that 
\begin{align}
 \hat{\hvu}_{t,t+1} |  (\yv^{t-1}, \Wme) \sim \Cc\Nc ( \zerou,   \frac{\Omega_{t}  \xvu^{*}_{t}  \xvu^\T_{t} \Omega_{t}  }{  \xvu^\T_{t} \Omega_{t} \xvu^{*}_{t}  +1 }).    \label{eq:hatht1gaussian}
\end{align}
Note that, for any two vectors  $\avu, \bvu \in \Cc^{M\times 1}$,   $ \| \avu  + \bvu \|^4$ can be expanded as in \eqref{eq:expend3245}.
Then, by replacing $\avu$ and $\bvu$ with $\hat{\hvu}_{t}$ and $\hat{\hvu}_{t,t+1}$ respectively, we have 
\begin{align}
 &   \E \bigl[      \|  \hat{\hvu}_{t} +  \hat{\hvu}_{t,t+1}\|^4     \big|   \Wme ,  \yv^{t-1}  \bigr]   \non\\ 
  = &\E \Bigl[   \| \hat{\hvu}_{t}\|^4 \! + \! \|\hat{\hvu}_{t,t+1}\|^4  \! + \!  2   \| \hat{\hvu}_{t} \|^2 \|\hat{\hvu}_{t,t+1}\|^2  \! +\!  4   \text{Re}^2( \hat{\hvu}_{t}^\H \hat{\hvu}_{t,t+1}) \!+ \!4 ( \| \hat{\hvu}_{t}\|^2 \!+\!  \|\hat{\hvu}_{t,t+1}\|^2) \cdot \text{Re} ( \hat{\hvu}_{t}^\H \hat{\hvu}_{t,t+1}) \   \Big|   \Wme ,  \yv^{t-1}  \Bigr]  \label{eq:expend25677} \\
  = &  \| \hat{\hvu}_{t}\|^4  +  \E \bigl[   \|\hat{\hvu}_{t,t+1}\|^4   \big|   \Wme ,  \yv^{t-1}  \bigr]    +  2   \| \hat{\hvu}_{t} \|^2   \cdot \E \bigl[  \|\hat{\hvu}_{t,t+1}\|^2   \big|   \Wme ,  \yv^{t-1}  \bigr]  +  4   \E \bigl[  \text{Re}^2( \hat{\hvu}_{t}^\H \hat{\hvu}_{t,t+1})   \   \big|   \Wme ,  \yv^{t-1}  \bigr]   \label{eq:expend9285} \\
 \leq    &  \| \hat{\hvu}_{t}\|^4  +  3     +  2   \| \hat{\hvu}_{t} \|^2    +  4   \E \bigl[  \text{Re}^2( \hat{\hvu}_{t}^\H \hat{\hvu}_{t,t+1})   \   \big|   \Wme ,  \yv^{t-1}  \bigr]   \label{eq:expend4378} \\ 
  \leq    &  \| \hat{\hvu}_{t}\|^4  +  3     +  2   \| \hat{\hvu}_{t} \|^2    +   2   \| \hat{\hvu}_{t} \|^2    \label{eq:expend0135} 
\end{align}
 where  \eqref{eq:expend25677} results from \eqref{eq:expend3245}; 
 \eqref{eq:expend9285} follows from the fact that  $\hat{\hvu}_{t}$ is deterministic given  $(\Wme ,  \yv^{t-1})$, and the identities that $ \E [ \text{Re} (\au^\H \bvu) ]  = \text{Re} ( \E[ \au^\H \bvu ])  = 0$  and  $ \E [   \|\bvu\|^2 \cdot \text{Re} (\au^\H \bvu) ]  = \text{Re} ( \E [    \au^\H \bvu \cdot \|\bvu\|^2  ])  = 0   $ for a fixed  vector $\au$ and a Gaussian  vector $\bvu \sim \Cc\Nc ( \zerou,   \Km )$; note that the \emph{odd}-order moments of a complex proper Gaussian vector are  zeros (see, e.g., \cite{Kostas:02});  \eqref{eq:expend4378} results from Lemma~\ref{lm:boundh1} and Lemma~\ref{lm:boundh41}, i.e.,   $\E \bigl[   \|\hat{\hvu}_{t,t+1}\|^4   \big|   \Wme ,  \yv^{t-1}  \bigr]  \leq  3$ and $\E \bigl[  \|\hat{\hvu}_{t,t+1}\|^2   \big|   \Wme ,  \yv^{t-1}  \bigr] \leq 1$;  
 \eqref{eq:expend0135} follows from that 
 \begin{align}
 &   4   \E \bigl[  \text{Re}^2( \hat{\hvu}_{t}^\H \hat{\hvu}_{t,t+1})   \   \big|   \Wme ,  \yv^{t-1}  \bigr]   \non \\ 
  = &  \E \bigl[    2 \hat{\hvu}_{t}^\H \hat{\hvu}_{t,t+1} \hat{\hvu}_{t,t+1}^\H \hat{\hvu}_{t}  \   \big|   \Wme ,  \yv^{t-1}  \bigr]    +  2 \text{Re}\bigl( \E \bigl[  \hat{\hvu}_{t}^\H \hat{\hvu}_{t,t+1}  \hat{\hvu}_{t}^\H \hat{\hvu}_{t,t+1}  \   \big|   \Wme ,  \yv^{t-1}  \bigr]  \bigr)   \label{eq:expend5826}  \\ 
        = &      2 \cdot \hat{\hvu}_{t}^\H  \cdot  \frac{  \Omega_{t}  \xvu^{*}_{t}  \xvu^\T_{t} \Omega_{t}   }{  \xvu^\T_{t} \Omega_{t} \xvu^{*}_{t}  +1 } \cdot \hat{\hvu}_{t}     +  2 \text{Re}\bigl( \E \bigl[  \hat{\hvu}_{t}^\H \hat{\hvu}_{t,t+1}  \hat{\hvu}_{t}^\H \hat{\hvu}_{t,t+1}  \   \big|   \Wme ,  \yv^{t-1}  \bigr]  \bigr)   \label{eq:expend3588}  \\ 
        = &       2 \cdot \hat{\hvu}_{t}^\H  \cdot  \frac{  \Omega_{t}  \xvu^{*}_{t}  \xvu^\T_{t} \Omega_{t}   }{  \xvu^\T_{t} \Omega_{t} \xvu^{*}_{t}  +1 } \cdot \hat{\hvu}_{t}   \label{eq:expend37889}  \\   
        = &       2 \cdot  \trace\Bigl(\hat{\hvu}_{t} \hat{\hvu}_{t}^\H  \cdot  \frac{  \Omega_{t}  \xvu^{*}_{t}  \xvu^\T_{t} \Omega_{t}   }{  \xvu^\T_{t} \Omega_{t} \xvu^{*}_{t}  +1 }  \Bigr)  \non  \\   
       \leq  &       2 \cdot   \| \hat{\hvu}_{t} \|^2  \cdot \trace\Bigl( \frac{  \Omega_{t}  \xvu^{*}_{t}  \xvu^\T_{t} \Omega_{t}   }{  \xvu^\T_{t} \Omega_{t} \xvu^{*}_{t}  +1 }  \Bigr)  \label{eq:expend9286}  \\  
         =   &       2 \cdot   \| \hat{\hvu}_{t} \|^2  \cdot \trace\Bigl( \frac{  \xvu^\T_{t} \Omega_{t} \Omega_{t}  \xvu^{*}_{t}     }{  \xvu^\T_{t} \Omega_{t} \xvu^{*}_{t}  +1 }  \Bigr)  \non  \\  
        \leq   &       2 \cdot   \| \hat{\hvu}_{t} \|^2  \cdot \trace\Bigl( \frac{  \xvu^\T_{t} \Omega_{t} \xvu^{*}_{t}  +1  }{  \xvu^\T_{t} \Omega_{t} \xvu^{*}_{t}  +1 }  \Bigr)   \label{eq:expend9983}   \\ 
    =    &   2   \| \hat{\hvu}_{t} \|^2    \label{eq:expend2446} 
\end{align}
where \eqref{eq:expend5826} follows from the identity that  $4\text{Re}^2(\avu^\H \bvu)  = ( \avu^\H \bvu + \bvu^\H \avu) ( \avu^\H \bvu + \bvu^\H \avu)=2 \avu^\H \bvu \bvu^\H \avu + 2 \text{Re}(\avu^\H \bvu \avu^\H \bvu) $ for two vectors $\avu$ and $\bvu$ with the same dimension;
\eqref{eq:expend3588} stems from \eqref{eq:hatht1gaussian}, i.e., $\hat{\hvu}_{t,t+1} |  (\yv^{t-1}, \Wme) \sim \Cc\Nc ( \zerou,   \frac{\Omega_{t}  \xvu^{*}_{t}  \xvu^\T_{t} \Omega_{t}  }{  \xvu^\T_{t} \Omega_{t} \xvu^{*}_{t}  +1 })$;
\eqref{eq:expend37889} follows from the identity that $\E[ \au^\H \bvu  \au^\H \bvu ]  = 0 $ for a fixed  vector $\au$ and a complex Gaussian  vector $\bvu \sim \Cc\Nc ( \zerou,   \Km )$; 
note that if $\bvu \sim \Cc\Nc ( \zerou,   \Km )$, then   $\cv\defeq  \au^\H \bvu  \sim \Cc\Nc ( 0,   \au^\H \Km \au )$ and  $\E[ \cv \cdot \cv ]  = 0 $;
\eqref{eq:expend9286} follows from the identities that  $\trace( A B)  \leq  \lambda_{\max}(A ) \trace(B)$ for positive semidefinite $m\times m$ Hermitian matrices $A, B$, and that $\lambda_{\max}(\hat{\hvut}\hat{\hvut}^\H) = \|\hat{\hvu}_{t}\|^2$; 
\eqref{eq:expend9983} follows from the same step in \eqref{eq:lmest92834}, i.e.,  $\xvu^\T_{t} \Omega_{t} \Omega_{t}   \xvu^{*}_{t} \leq  \xvu^\T_{t} \Omega_{t} \xvu^{*}_{t}  +1$.

Finally, from the step in \eqref{eq:expend0135}, we have the following inequality \[\E \bigl[      \|  \hat{\hvu}_{t} +  \hat{\hvu}_{t,t+1}\|^4     \big|   \Wme ,  \yv^{t-1}  \bigr] \leq \| \hat{\hvu}_{t}\|^4   +  4   \| \hat{\hvu}_{t} \|^2  + 3. \]  By taking the expectation on both sides of the above inequality, and using the identity that $ \E \bigl[      \|  \hat{\hvu}_{t} +  \hat{\hvu}_{t,t+1}\|^4    \bigr]  =  \E \bigl[  \  \E \bigl[      \|  \hat{\hvu}_{t} +  \hat{\hvu}_{t,t+1}\|^4     \big|   \Wme ,  \yv^{t-1} \bigr]  \ \bigr]$, it yields  
\begin{align}
  \E \bigl[      \|  \hat{\hvu}_{t} +  \hat{\hvu}_{t,t+1}\|^4    \bigr]
  \leq      \E \bigl[   \| \hat{\hvu}_{t}\|^4 \bigr]     +  4 \cdot  \E \bigl[   \| \hat{\hvu}_{t} \|^2  \bigr]  + 3    \label{eq:expend9844} 
\end{align}
which completes the proof.

\section*{Acknowledgement}

We wish to thank Ayfer \"Ozg\"ur and Andrea Goldsmith for helpful comments during the early stage of this work.



\end{document}